\documentclass[12pt]{article}

\usepackage{xr-hyper}
\usepackage{amsmath, amsfonts, amssymb, amsthm, fullpage, color, bbm, graphicx, epstopdf, enumerate, titlesec, mathtools, pdflscape, afterpage}
\usepackage[T1]{fontenc}
\usepackage{lmodern}
\usepackage[title]{appendix}

\usepackage[onehalfspacing]{setspace}
\usepackage{graphicx}
\usepackage{float}
\usepackage[caption = false]{subfig}

\titleformat*{\paragraph}{\sc}

\definecolor{darkred}{rgb}{0.5,0,0}
\usepackage{hyperref}
\hypersetup{allbordercolors={1 1 1},allcolors=darkred,colorlinks=true}

\usepackage[capitalize,nameinlink,noabbrev]{cleveref}

\usepackage[round]{natbib}
\allowdisplaybreaks

\usepackage[runin]{abstract}

\theoremstyle{plain}
\newtheorem{asn}{Assumption}
\crefname{asn}{Assumption}{Assumptions}
\newtheorem{lem}{Lemma}
\crefname{lem}{Lemma}{Lemmas}
\numberwithin{lem}{section}

\newtheorem{prop}{Proposition}

\newtheorem{defn}{Definition}
\newtheorem*{claim*}{Claim}
\newtheorem*{cor*}{Corollary}

\DeclareMathOperator*{\tr}{trace}

\DeclareMathOperator*{\diag}{diag}

\DeclareMathOperator*{\var}{Var}

\crefname{sappsec}{Supplemental Appendix}{Supplemental Appendices}
\crefname{sappsubsec}{Supplemental Appendix}{Supplemental Appendices}
\crefname{sappsubsubsec}{Supplemental Appendix}{Supplemental Appendices}
\crefname{appsec}{Appendix}{Appendices}

\newcommand{\alert}[1]{#1}
\crefname{sappsec}{Supplemental Appendix}{Supplemental Appendices}
\crefname{sappsubsec}{Supplemental Appendix}{Supplemental Appendices}
\crefname{sappsubsubsec}{Supplemental Appendix}{Supplemental Appendices}
\crefname{appsec}{Appendix}{Appendices}

\begin{document}

\title{\texorpdfstring{\vspace{-2em}}{}Local Projection Inference is Simpler \texorpdfstring{\\}{} and More Robust Than You Think\thanks{Email: {\tt jm4474@columbia.edu}, {\tt mikkelpm@princeton.edu}. We are grateful for comments from \alert{two anonymous referees}, Jushan Bai, Otavio Bartalotti, Guillaume Chevillon, Max Dovi, Marco Del Negro, Domenico Giannone, Nikolay Gospodinov, Michael Jansson, \`Oscar Jord\`a, Lutz Kilian, Michal Koles\'{a}r, Simon Lee, Sophocles Mavroeidis, Konrad Menzel, Ulrich M\"{u}ller, Serena Ng, Elena Pesavento, Mark Watson, Christian Wolf, Tao Zha, and numerous seminar participants. We would like to especially thank Atsushi Inoue and Anna Mikusheva for a very helpful discussion of our paper. Montiel Olea would like to thank Qifan Han and Giovanni Topa for excellent research assistance.  Plagborg-M{\o}ller acknowledges that this material is based upon work supported by the NSF under Grant {\#}1851665. Any opinions, findings, and conclusions or recommendations expressed in this material are those of the authors and do not necessarily reflect the views of the NSF.}}
\author{Jos\'{e} Luis Montiel Olea \\ Columbia University \and Mikkel Plagborg-M{\o}ller \\ Princeton University}
\date{First version: March 17, 2020 \texorpdfstring{\\[0.5ex]}{}This version: December 4, 2020}

\maketitle

\begin{abstract}
Applied macroeconomists often compute confidence intervals for impulse responses using local projections, i.e., direct linear regressions of future outcomes on current covariates. This paper proves that local projection inference robustly handles two issues that commonly arise in applications: highly persistent data and the estimation of impulse responses at long horizons.  We consider local projections that control for lags of the variables in the regression. We show that lag-augmented local projections with normal critical values are asymptotically valid uniformly over (i) both stationary and non-stationary data, and also over (ii) a wide range of response horizons. Moreover, lag augmentation obviates the need to correct standard errors for serial correlation in the regression residuals. Hence, local projection inference is arguably both simpler than previously thought and more robust than standard autoregressive inference, whose validity is known to depend sensitively on the persistence of the data and on the length of the horizon.
\end{abstract}

\noindent \emph{Keywords:} impulse response, local projection, long horizon, uniform inference.

\section{Introduction}
Impulse response functions are key objects of interest in empirical macroeconomic analysis. It is increasingly popular to estimate these parameters using the method of \emph{local projections} \citep{Jorda2005}:  simple linear regressions of a future outcome on current covariates \citep{Ramey2016,Angrist2018,Nakamura2018,Stock2018}. Since local projection estimators are regression coefficients, they have a simple and intuitive interpretation. Moreover, inference can be carried out using textbook standard error formulae, adjusting for serial correlation in the (multi-step forecast) regression residuals. 

Despite its popularity, there exist no theoretical results justifying the use of local projection \emph{inference} over autoregressive procedures. From an identification and estimation standpoint, \citet{Kilian2017} and \citet{PlagborgMoller2019} argue that neither local projections nor Vector Autoregressions (VARs) dominate the other in terms of mean squared error \alert{in finite samples}, and in population the two methods are equivalent. However, from an inference perspective, the only available guidance on the relative performance of local projections comes in the form of a small number of simulation studies, which by necessity cannot cover the entire range of empirically relevant data generating processes. 

In this paper we show that---in addition to its intuitive appeal---frequentist local projection inference is robust to two common features of macroeconomic applications: highly persistent data and the estimation of impulse responses at long horizons. Key to our result is that we consider \emph{lag-augmented} local projections, which use lags of the variables in the regression as controls. Formally, we prove that standard confidence intervals based on such lag-augmented local projections have correct asymptotic coverage \emph{uniformly} over the persistence in the data generating process and over a wide range of horizons.\footnote{We focus on marginal inference on individual impulse responses, not \emph{simultaneous} inference on a vector of several response horizons \citep{IK2016,MontielOlea2019}.} This means that confidence intervals remain valid even if the data exhibits unit roots, and even at horizons $h$ that are allowed to grow with the sample size $T$, e.g., $h=h_T \propto T^{\eta}$, $\eta \in [0,1)$. In fact, when persistence is not an issue, and the data is known to be stationary, local projection inference is also valid at \emph{long} horizons; i.e., horizons that are a non-negligible fraction of the sample size ($h_T \propto T$).

Lag-augmenting local projections not only robustifies inference, it also simplifies the computation of standard errors by obviating the adjustment for serial correlation in the residuals. It is common practice in the local projections literature to compute Heteroskedasticity and Autocorrelation Consistent/Robust (HAC/HAR) standard errors \citep{Jorda2005,Ramey2016,Kilian2017,Stock2018}. Instead, we prove that the usual Eicker-Huber-White heteroskedasticity-robust standard errors suffice for \emph{lag-augmented} local projections. The reason is that, although the regression residuals are serially correlated, the \emph{regression scores} (the product of the residuals and residualized regressor of interest) are serially uncorrelated under weak assumptions. This finding further simplifies local projection inference, as it side-steps the delicate choice of HAR procedure and associated difficult-to-interpret tuning parameters \citep[e.g.,][]{Lazarus2018}.

The robustness properties of lag-augmented local projection inference stand in contrast to the well-known fragility of standard autoregressive procedures. Textbook autoregressive inference methods for impulse responses (such as the delta method) are invalid in \alert{some cases with} near-unit roots or medium-long to long horizons (e.g., $h_T \propto \sqrt{T})$, \alert{as discussed further below}. We show that lag-augmented local projection inference is valid when the data has near-unit roots and the horizon sequence satisfies $h_{T}/T \rightarrow 0$. Though the method fails in the case of unit roots and very long horizons $h_T \propto T$, existing VAR-based methods that achieve correct coverage in this case are either highly computationally demanding or result in impractically wide confidence intervals.  When the data is stationary and interest centers on short horizons, local projection inference is valid but less efficient than textbook AR inference. Thus, the robustness afforded by our recommended procedure is not a free lunch. We provide a detailed comparison with alternative inference procedures in \cref{sec:comp} below.

Our results rely on assumptions that are similar to those used in the literature on autoregressive inference. In particular, we assume that the true model is a VAR($p$) with possibly conditionally heteroskedastic innovations and known lag length. We discuss the choice of lag length $p$ in \cref{sec:conc}. The key assumption that we require on the innovations is that they are conditionally mean independent of both past and \emph{future} innovations (which is trivially satisfied for i.i.d. innovations). Our strengthening of the usual martingale difference assumption is crucial to avoid HAC inference, but we show that the assumption is satisfied for a large class of conditionally heteroskedastic innovation processes. The robustness property of local projection inference only obtains asymptotically if the researcher controls for all $p$ lags of all of the variables in the VAR system. Thus, our paper highlights the advantages of multivariate modeling even when using single-equation local projections.

To illustrate our theoretical results, we present a small-scale simulation study suggesting that lag-augmented local projection confidence intervals achieve a favorable trade-off between coverage and length. Since local projection estimation is subject to small-sample biases just like VAR estimation \citep{Herbst2020}, we consider a simple and computationally convenient bootstrap implementation of local projection. The simulations suggest that non-augmented autoregressive procedures with delta method standard errors have more severe under-coverage problems than local projection inference, especially at moderate and long horizons. Autoregressive confidence intervals can be meaningfully \emph{shorter} than lag-augmented local projection intervals in \emph{relative} terms, but in \emph{absolute} terms the difference in length is surprisingly modest. Our simulations also indicate that lag-augmented local projections with heteroskedasticity-robust standard errors have better coverage/length properties than more standard \emph{non-augmented} local projections with off-the-shelf HAR standard errors. Finally, although the lag-augmented autoregressive bootstrap procedure of \citet{Inoue2020} achieves good coverage, it yields prohibitively wide confidence intervals at longer horizons when the data is persistent.

\paragraph{Related Literature.}
\alert{It is well known that standard autoregressive (AR) inference on impulse responses requires an auxiliary rank condition to rule out super-consistent limit distributions, thus yielding a $\sqrt{T}$-normal limit with strictly positive variance, see Assumption B of \citet{Inoue2020}. When this rank condition holds, the textbook AR impulse response estimator is asymptotically normal even in the presence of (near-)unit roots \citep{Inoue2002}. However, there are two common features of the data that lead to violations of the rank condition. First, the condition can fail when some linear combinations of the variables exhibit no persistence \citep{Benkwitz2000}. Second, in the presence of (near-)unit roots, certain linear combinations of the autoregressive coefficients are necessarily super-consistent \citep{Sims1990}. This compromises textbook AR inference for certain combinations of impulse response horizons and parameter values that typically cannot be ruled out \emph{a priori}, especially in AR(1) or VAR(1) models, but also in higher-order autoregressions (\citealp{Phillips1998}; \citealp[Remark 3, p. 455]{Inoue2020}). In an important paper, \citet{Inoue2020} show that \emph{lag-augmented} autoregressive inference solves the rank problem caused by (near-)unit roots, but data generating processes that lack persistence still need to be ruled out \emph{a priori}. We build on their ideas, which in turn are based on \citet{Toda1995} and \citet{Dolado1996}. As we show, the validity of lag-augmented \emph{local projection} (LP) inference does not hinge on auxiliary rank conditions.}

\alert{Moreover, the validity of textbook AR inference is also compromised when the length of the impulse response horizon is large \citep{Pesavento2007,Mikusheva2012}.} Standard bootstrap methods rectify some of these problems, but not all. Several papers have proposed AR-based methods for impulse response inference at \emph{long} horizons $h=h_T \propto T$ \citep{Wright2000,Gospodinov2004,Pesavento2007,Mikusheva2012,Inoue2020}. With the exception of \citet{Mikusheva2012}, the literature \alert{on long-horizon inference} has exclusively focused on near-unit root processes as opposed to devising uniformly valid procedures. The \citet{Hansen1999} grid bootstrap analyzed by \citet{Mikusheva2012} is asymptotically valid at short and long horizons. However, it is not valid at intermediate horizons (e.g., $h_T \propto \sqrt{T}$), unlike the LP procedure we analyze. \citeauthor{Mikusheva2012} argues, though, that the grid bootstrap is \emph{close} to being valid at intermediate horizons, although it is much more computationally demanding than our recommended procedure, especially in VAR models with several parameters. \citet{Inoue2020} show that a version of the Efron bootstrap confidence interval, when applied to lag-augmented AR estimators, is valid at long horizons. We show that this procedure delivers impractically wide confidence intervals at moderately long horizons when the data is persistent, unlike lag-augmented LP.

We appear to be the first to prove the \emph{uniform} validity of lag-augmented LP inference. \citet{Mikusheva2007,Mikusheva2012} and \citet{Inoue2020} derive the uniform coverage properties of various AR inference procedures, but they do not consider LP. The \emph{pointwise} properties of LP procedures have been discussed by \citet{Jorda2005}, \citet{Kilian2017}, and \citet{Stock2018}, among others. \citet{Kilian2011} and \citet{Brugnolini2018} present simulation studies comparing AR inference and LP inference. \citet{Brugnolini2018} finds that the lag length in the LP matters, which is consistent with our theoretical results.

Though the theoretical results in this paper appear to be novel, \citet[Section 5]{Dufour2006} and \citet{Breitung2019} have discussed some of the main ideas presented herein. First, both these papers state that lag augmentation in LP avoids unit root asymptotics, but neither paper considers inference at long horizons or derives uniform inference properties. Second, \citet{Breitung2019} further argue that HAC inference in LP can be avoided if the true model is a VAR($p$), although it is not clear from their discussion what are the assumptions needed for this to be true. Neither of these papers provide results concerning the efficiency of lag-augmented LP inference relative to other lag-augmented or non-augmented inference procedures, as we do in \cref{sec:comp}. 

Local projections are closely related to multi-step forecasts. \citet{Richardson1989} and \citet{Valkanov2003} develop a non-standard limit distribution theory for long-horizon forecasts. \citet{Chevillon2017} proves a robustness property of direct multi-step inference that involves non-normal asymptotics due to the lack of lag augmentation. \citet{Phillips2013} test the null hypothesis of no long-horizon predictability using a novel approach that requires a choice of tuning parameters, but yields uniformly-over-persistence normal asymptotics. This test is based on an estimator with a faster convergence rate than ours in the non-stationary case. However, to the best of our knowledge, their approach does not carry over immediately to impulse response inference, and it is not obvious whether the procedure is uniformly valid over both short and long horizons. 

\paragraph{Outline.}
\cref{sec:ar1} provides a non-technical overview of our results in the context of a simple AR(1) model, including an illustrative simulation study. \cref{sec:comp} provides an in-depth comparison of lag-augmented LP with other inference procedures. \cref{sec:var} presents the formal uniformity result for a general VAR($p$) model. \cref{sec:boot} describes a simple bootstrap implementation of lag-augmented local projection that we recommend for practical use. \cref{sec:conc} concludes. Proofs are relegated to \cref{sec:var_proof_main} and the Online Supplement. \cref{sec:appendix,sec:var_XpX_ar1} contain further simulation and theoretical results. The supplement and a full Matlab code repository are available online.\footnote{\url{https://github.com/jm4474/Lag-augmented_LocalProjections} \label{fn:github}}

\section{Overview of the Results}
\label{sec:ar1}
This section provides an overview of our results in the context of a simple univariate AR(1) model. The discussion here merely intends to illustrate our main points. \cref{sec:var} presents general results for VAR($p$) models.

\subsection{Lag-Augmented Local Projection}
\label{sec:ar1_intuition}

\paragraph{Model.}
Consider the AR(1) model for the data $\lbrace y_t \rbrace$:
\begin{equation} \label{eqn:ar1}
y_t = \rho y_{t-1} + u_t,\quad t=1,2,\dots,T,
\quad y_0 = 0.
\end{equation}
The parameter of interest is a nonlinear transformation of $\rho$, namely the impulse response coefficient at horizon $h \in \mathbb{N}$. We denote this parameter by $\beta(\rho,h) \equiv \rho^h$. In \cref{sec:var} below we argue that the zero initial condition $y_0=0$ is not needed for our results to go through. Our main assumption in the univariate model is:

\begin{asn} \label{asn:u_mds}
$\lbrace u_t \rbrace$ is strictly stationary and satisfies $E(u_t \mid \lbrace u_s \rbrace_{s \neq t})=0$ almost surely.
\end{asn}
\noindent The assumption requires the innovations to be mean independent relative to past and future innovations. This is a slight strengthening of the usual martingale difference assumption on $u_t$.  \cref{asn:u_mds} is trivially satisfied if $\lbrace u_t \rbrace$ is i.i.d., but it also allows for stochastic volatility and GARCH-type innovation processes.\footnote{For example, consider processes $u_t = \tau_t \varepsilon_t$, where $\varepsilon_t$ is i.i.d. with $E(\varepsilon_t)=0$, and for which one of the following two sets of conditions hold: (a) $\lbrace \tau_t \rbrace$ and $\lbrace \varepsilon_t \rbrace$ are independent processes; or (b) $\tau_t$ is a function of lagged values of $\varepsilon_t^2$, and the distribution of $\varepsilon_t$ is symmetric. \cref{asn:u_mds} is in principle testable, but that is outside the scope of this paper.} 

\paragraph{Local Projections With and Without Lag Augmentation.}
We consider the local projection (LP) approach of \cite{Jorda2005} for conducting inference about the impulse response $\beta(\rho, h)$. A common motivation for this approach is  that the AR(1) model \eqref{eqn:ar1} implies 
\begin{equation} \label{eqn:ytph_decomp}
y_{t+h} = \beta(\rho,h)y_t + \xi_t(\rho,h),
\end{equation}
where the regression residual (or \emph{multi-step forecast error}), 
\[\xi_t(\rho,h) \equiv \sum_{\ell=1}^h \rho^{h-\ell}u_{t+\ell},\]
is generally serially correlated, even if the innovation $u_t$ is i.i.d.

The most straight-forward LP impulse response estimator simply regresses $y_{t+h}$ on $y_t$, as suggested by equation \eqref{eqn:ytph_decomp}, but the validity of this approach is sensitive to the persistence of the data. Specifically, this standard approach leads to a non-normal limiting distribution for the impulse response estimator when $\rho \approx 1$, since the regressor $y_t$ exhibits near-unit-root behavior in this case. Hence, inference based on normal critical values will not be valid uniformly over all values of $\rho \in [-1,1]$ even for fixed forecast horizons $h$. If $\rho$ is safely within the stationary region, then the LP estimator is asymptotically normal, but inference generally requires the use of Heteroskedasticity and Autocorrelation Robust (HAR) standard errors to account for serial correlation in the residual $\xi_t(\rho,h)$.

To robustify and simplify inference, we will instead consider a \emph{lag-augmented} local projection, which uses $y_{t-1}$ as an additional control variable. In the autoregressive literature, ``lag augmentation'' refers to the practice of using more lags for estimation than suggested by the true autoregressive model. Define the covariate vector $x_t \equiv (y_t,y_{t-1})'$. Given any horizon $h \in \mathbb{N}$, the lag-augmented LP estimator $\hat{\beta}(h)$ of $\beta(\rho,h)$ is given by the coefficient on $y_t$ in a regression of $y_{t+h}$ on $y_t$ and $y_{t-1}$:
\begin{equation}
\begin{pmatrix} \label{eqn:LPestimates}
\hat{\beta}(h) \\
\hat{\gamma}(h)
\end{pmatrix} \equiv \left(\sum_{t=1}^{T-h} x_tx_t'\right)^{-1}\sum_{t=1}^{T-h}x_t y_{t+h}.
\end{equation}
Here $\hat{\beta}(h)$ is the impulse response estimator of interest, while $\hat{\gamma}(h)$ is a nuisance coefficient.

The purpose of the lag augmentation is to make the effective regressor of interest stationary even when the data $y_t$ has a unit root. Note that equations \eqref{eqn:ar1}--\eqref{eqn:ytph_decomp} imply
\begin{equation} \label{eqn:ytph_decomp2}
y_{t+h} =  \beta(\rho,h)u_t + \beta(\rho,h+1)y_{t-1} + \xi_t(\rho,h).
\end{equation}
If $u_t$ were observed, the above equation suggests regressing $y_{t+h}$ on $u_t$, while controlling for $y_{t-1}$. Intuitively, this will lead to an asymptotically normal estimator of $\beta(\rho,h)$, since the regressor of interest $u_t$ is stationary by \cref{asn:u_mds}, and we control for the term that involves the possibly non-stationary regressor $y_{t-1}$. Fortunately, due to the linear relationship $y_t=\rho y_{t-1}+u_t$, the coefficient $\hat{\beta}(h)$ on $y_t$ in the feasible lag-augmented regression \eqref{eqn:LPestimates} on $(y_t,y_{t-1})$ precisely equals the coefficient on $u_t$ in the desired regression on $(u_t,y_{t-1})$. This argument for why lag-augmented LP can be expected to have a uniformly normal limit distribution even when $\rho \approx 1$ is completely analogous to the reasoning for using lag augmentation in AR inference \citep{Sims1990,Toda1995,Dolado1996,Inoue2002,Inoue2020}. In the LP case, lag augmentation has the additional benefit of simplifying the computation of standard errors, as we now discuss.

\paragraph{Standard Errors.} 
We now define the standard errors for the lag-augmented LP estimator. We will show that, contrary to conventional wisdom (e.g., \citealp[p. 166]{Jorda2005}; \citealp[p. 84]{Ramey2016}), HAR standard errors are \emph{not} needed to conduct inference on \emph{lag-augmented} LP, despite the fact that the regression residual $\xi_t(\rho,h)$ is serially correlated. Instead, it suffices to use the usual heteroskedasticity-robust Eicker-Huber-White standard error of $\hat{\beta}(h)$:\footnote{This is computed by the {\tt regress, robust} command in Stata, for example. The usual homoskedastic standard error formula suffices if $u_t$ is assumed to be i.i.d.}
\begin{equation}
\label{eqn:EHW}
\hat{s}(h) \equiv \frac{(\sum_{t=1}^{T-h} \hat{\xi}_t(h)^2 \hat{u}_t(h)^2)^{1/2}}{\sum_{t=1}^{T-h} \hat{u}_t(h)^2},
\end{equation}
where we define the lag-augmented LP residuals
\begin{equation}\label{eqn:estimatedresiduals}
\hat{\xi}_t(h) \equiv y_{t+h} - \hat{\beta}(h)y_t - \hat{\gamma}(h)y_{t-1},\quad t=1,2,\dots,T-h,
\end{equation}
and the residualized regressor of interest
\[\hat{u}_t(h) \equiv y_t - \hat{\rho}(h)y_{t-1},\quad t=1,2,\dots,T-h,\]
\[\hat{\rho}(h) \equiv \frac{\sum_{t=1}^{T-h} y_t y_{t-1}}{\sum_{t=1}^{T-h} y_{t-1}^2}.\]
As mentioned in the introduction, the fact that we may avoid HAR inference simplifies the implementation of LP inference, as there is no need to choose amongst alternative HAR procedures or specify tuning parameters such as bandwidths \citep{Lazarus2018}.

Why is it not necessary to adjust for serial correlation in the residuals? Since lag-augmented LP controls for $y_{t-1}$, equation  \eqref{eqn:ytph_decomp2} suggests that the estimator $\hat{\beta}(h)$ is asymptotically equivalent with the coefficient in a linear regression of the (population) residualized outcome $y_{t+h} - \beta(\rho,h+1)y_{t-1}$ on the (population) residualized regressor $u_t = y_t - \rho y_{t-1}$:
\begin{align*}
\hat{\beta}(h) &\approx \frac{\sum_{t=1}^{T-h} \lbrace y_{t+h} - \beta(\rho,h+1)y_{t-1} \rbrace u_t}{\sum_{t=1}^{T-h}u_t^2} \\
&= \beta(\rho,h) + \frac{\sum_{t=1}^{T-h} \xi_t(\rho,h) u_t}{\sum_{t=1}^{T-h}u_t^2}.
\end{align*}
The second term in the decomposition above determines the sampling distribution of the lag-augmented local projection. Although the multi-step regression residual $\xi_t(\rho,h)$ is serially correlated on its own, the \emph{regression score} $\xi_t(\rho,h)u_t$ is serially uncorrelated under \cref{asn:u_mds}.\footnote{\citet{Breitung2019} make this same observation, but they appear to claim that it is sufficient to assume that $\lbrace u_t \rbrace$ is white noise, which is incorrect.} For any $s < t$,
\begin{align}
E[\xi_t(\rho,h)u_t\xi_s(\rho,h)u_s ] &= E[E(\xi_t(\rho,h)u_t\xi_s(\rho,h)u_s \mid u_{s+1},u_{s+2},\dots)] \nonumber \\
&= E[\xi_t(\rho,h) u_t\xi_s(\rho,h)\underbrace{E(u_s \mid u_{s+1},u_{s+2},\dots)}_{=0}]. \label{eqn:scores_uncorr}
\end{align}
Thus, the heteroskedasticity-robust (but not autocorrelation-robust) standard error $\hat{s}(h)$ suffices for doing inference on $\hat{\beta}(h)$.\footnote{\citet[p. 152]{Stock2018} mention a similar conclusion for the distinct case of LP with an instrumental variable, under some conditions on the instrument.} Notice that this result crucially relies on (i) lag-augmenting the local projections and (ii) the strengthening in \cref{asn:u_mds} of the usual martingale difference assumption on $\lbrace u_t \rbrace$ (as remarked above, the strengthening still allows for conditional heteroskedasticity and other plausible features of economic shocks).\footnote{The nuisance coefficient $\hat{\gamma}(h)$ is not interesting \emph{per se}, but note that inference on this coefficient would generally require HAR standard errors, and its limit distribution is in fact non-standard when $\rho \approx 1$.} 

\alert{Though lag augmentation robustifies and simplifies local projection inference, it is not necessarily a free lunch. We show in \cref{sec:comp} that the relative efficiency of non-augmented and lag-augmented local projection estimators depends on $\rho$ and $h$.}

\paragraph{Lag-Augmented Local Projection Inference.}
Define the nominal $100 (1-\alpha)\%$ lag-augmented LP confidence interval for the impulse response at horizon $h$ based on the standard error $\hat{s}(h)$:
\[\hat{C}(h,\alpha) \equiv \left[ \hat{\beta}(h)-z_{1-\alpha/2}\: \hat{s}(h)\:,\: \hat{\beta}(h)+z_{1-\alpha/2}\: \hat{s}(h) \right],\]
where $z_{1-\alpha/2}$ is the $(1-\alpha/2)$ quantile of the standard normal distribution.

Our main result shows that the lag-augmented LP confidence interval above is valid regardless of the persistence of the data, i.e., whether or not the data has a unit root. Crucially, the result does not break down at moderately long horizons $h$. We provide a formal result for VAR($p$) models in \cref{sec:var} and for now just discuss heuristics. Consider any upper bound $\bar{h}_T$ on the horizon which satisfies $\bar{h}_T/T \to 0$. Then \cref{thm:var_lp_inference} below implies that 
\begin{equation} \label{eqn:coverage}
\inf_{\rho \in [-1,1]} \inf_{1 \leq h \leq \bar{h}_T}  P_{\rho}\left( \beta(\rho,h) \in \hat{C}(h,\alpha) \right )  \to  1-\alpha \quad \text{as } T \to \infty,
\end{equation}
where $P_\rho$ denotes the distribution of the data $\lbrace y_t \rbrace$ under the AR(1) model \eqref{eqn:ar1} with parameter $\rho$. In words, the result states that, for sufficiently large sample sizes, LP inference is valid even under the \emph{worst-case} choices of parameter $\rho \in [-1,1]$ and horizon $h \in [1,\bar{h}_T]$. As is well known, such \emph{uniform} validity is a much stronger  result than \emph{pointwise} validity for fixed $\rho$ and $h$. In fact, if we restrict attention to only the stationary region $\rho \in [-1+a,1-a]$, $a \in (0,1)$, then the statement \eqref{eqn:coverage} is true with the upper bound $\bar{h}_T = (1-a)T$ on the horizon. That is, if we know the time series is not close to a unit root, then local projection inference is valid even at long horizons $h$ that are non-negligible fractions of the sample size $T$. 

\subsection{Illustrative Simulation Study}
\label{sec:sim_ar1}

We now present a small simulation study to show that lag-augmented LP achieves a favorable trade-off between robustness and efficiency relative to other procedures. For clarity, we continue to assume the simple AR(1) model \eqref{eqn:ar1} with known lag length. Our baseline design considers homoskedastic innovations $u_t \stackrel{i.i.d.}{\sim} N(0,1)$. In \cref{sec:sim_ar1_arch} we present results for ARCH innovations.

We stress that, although we use the AR(1) model for illustration here, the central goal of this paper is to develop a procedure that is feasible even in realistic VAR($p$) models. Thus, we avoid computationally demanding procedures, such as the AR grid bootstrap, which are difficult to implement in applied settings. We provide an extensive theoretical comparison of various inference procedures in \cref{sec:comp}.

\cref{tab:TableMC} displays the coverage and median length of impulse response confidence intervals at various horizons. We consider several versions of AR inference and LP inference, either implemented using the bootstrap or using delta method standard errors. ``LP'' denotes local projection and ``AR'' autoregressive inference. ``LA'' denotes lag augmentation. The subscript ``$b$'' denotes bootstrap confidence intervals constructed from a wild recursive bootstrap design \citep{Goncalves2004}, as described in \cref{sec:boot} (for LP we use the percentile-t confidence interval). Columns without the ``$b$'' subscript use delta method standard errors. For LA-LP, we always use Eicker-Huber-White standard errors as discussed in \cref{sec:ar1_intuition}, whereas non-augmented LP always uses HAR standard errors.\footnote{As an off-the-shelf, state-of-the-art HAR procedure, we choose the Equally Weighted Cosine (EWC) estimator with degrees of freedom as recommended by \citet[equations 4 and 10]{Lazarus2018}. The degrees of freedom depend on the effective sample size $T-h$ and thus differ across horizons $h$.} The column ``AR-LA'' is the Efron bootstrap confidence interval for \emph{lag-augmented} AR estimates developed by \citet{Inoue2020} and discussed further in \cref{sec:comp}.\footnote{We use the \citet{Pope1990} bias-corrected AR estimates to generate the bootstrap samples, as recommended by \citet{Inoue2020}.} \alert{All estimation procedures include an intercept.} The sample size is $T=240$. We consider data generating processes (DGPs) $\rho \in \lbrace 0,.5,.95,1\rbrace $ and horizons $h$ up to 60 periods (25\% of the sample size, which is not unusual in applied work). The nominal confidence level is 90\%. We use 5,000 Monte Carlo repetitions, with 2,000 bootstrap draws per repetition.

\afterpage{
\begin{landscape}
\begin{table}[p]
    \centering
    \caption{Monte Carlo results: homoskedastic innovations}
    \vspace{0.5\baselineskip}
    \begin{tabular}{r|cccccc|cccccc}
& \multicolumn{6}{c|}{Coverage} & \multicolumn{6}{c}{Median length} \\
$h$ & $\text{LP-LA}_b$ & $\text{LP-LA}$ & $\text{LP}_b$ & $\text{LP}$ & $\text{AR-LA}_b$ & $\text{AR}$ & $\text{LP-LA}_b$ & $\text{LP-LA}$ & $\text{LP}_b$ & $\text{LP}$ & $\text{AR-LA}_b$ & $\text{AR}$ \\
\hline
\multicolumn{13}{c}{$\rho = 0.00$} \\
  1 & 0.902 & 0.892 & 0.912 & 0.889 & 0.891 & 0.894 & 0.218 & 0.211 & 0.233 & 0.215 & 0.211 & 0.210 \\
  6 & 0.908 & 0.899 & 0.916 & 0.898 & 0.000 & 1.000 & 0.219 & 0.214 & 0.233 & 0.220 & 0.000 & 0.000 \\
 12 & 0.909 & 0.900 & 0.903 & 0.897 & 0.000 & 1.000 & 0.222 & 0.217 & 0.230 & 0.226 & 0.000 & 0.000 \\
 36 & 0.903 & 0.895 & 0.903 & 0.898 & 0.000 & 1.000 & 0.235 & 0.229 & 0.244 & 0.239 & 0.000 & 0.000 \\
 60 & 0.898 & 0.886 & 0.894 & 0.889 & 0.000 & 0.979 & 0.252 & 0.244 & 0.261 & 0.255 & 0.000 & 0.000 \\
\multicolumn{13}{c}{$\rho = 0.50$} \\
  1 & 0.906 & 0.896 & 0.912 & 0.885 & 0.897 & 0.897 & 0.219 & 0.212 & 0.205 & 0.187 & 0.211 & 0.184 \\
  6 & 0.895 & 0.886 & 0.906 & 0.875 & 0.897 & 0.832 & 0.252 & 0.245 & 0.293 & 0.266 & 0.046 & 0.032 \\
 12 & 0.906 & 0.894 & 0.903 & 0.887 & 0.897 & 0.766 & 0.255 & 0.248 & 0.293 & 0.280 & 0.002 & 0.001 \\
 36 & 0.900 & 0.889 & 0.901 & 0.884 & 0.897 & 0.643 & 0.271 & 0.262 & 0.309 & 0.296 & 0.000 & 0.000 \\
 60 & 0.905 & 0.891 & 0.903 & 0.880 & 0.897 & 0.595 & 0.291 & 0.279 & 0.333 & 0.316 & 0.000 & 0.000 \\
\multicolumn{13}{c}{$\rho = 0.95$} \\
  1 & 0.892 & 0.878 & 0.842 & 0.827 & 0.882 & 0.850 & 0.220 & 0.212 & 0.076 & 0.072 & 0.212 & 0.075 \\
  6 & 0.903 & 0.838 & 0.851 & 0.789 & 0.882 & 0.810 & 0.523 & 0.452 & 0.395 & 0.345 & 1.011 & 0.318 \\
 12 & 0.889 & 0.806 & 0.853 & 0.752 & 0.882 & 0.769 & 0.678 & 0.550 & 0.644 & 0.518 & 1.744 & 0.430 \\
 36 & 0.885 & 0.814 & 0.865 & 0.674 & 0.882 & 0.656 & 0.728 & 0.625 & 0.859 & 0.612 & 6.567 & 0.272 \\
 60 & 0.892 & 0.833 & 0.892 & 0.693 & 0.882 & 0.595 & 0.731 & 0.651 & 0.942 & 0.641 & 23.050 & 0.095 \\
\multicolumn{13}{c}{$\rho = 1.00$} \\
  1 & 0.895 & 0.874 & 0.820 & 0.554 & 0.877 & 0.532 & 0.219 & 0.211 & 0.040 & 0.040 & 0.210 & 0.039 \\
  6 & 0.875 & 0.777 & 0.836 & 0.503 & 0.877 & 0.494 & 0.564 & 0.498 & 0.243 & 0.222 & 1.206 & 0.214 \\
 12 & 0.843 & 0.676 & 0.827 & 0.429 & 0.877 & 0.459 & 0.821 & 0.671 & 0.477 & 0.385 & 2.553 & 0.379 \\
 36 & 0.741 & 0.428 & 0.755 & 0.200 & 0.877 & 0.348 & 1.338 & 0.950 & 1.200 & 0.592 & 21.107 & 0.670 \\
 60 & 0.642 & 0.276 & 0.712 & 0.156 & 0.877 & 0.295 & 1.434 & 0.978 & 1.667 & 0.637 & 161.250 & 0.731 \\
\end{tabular}
    \label{tab:TableMC}
    \\
    \vspace{0.5\baselineskip}
\begin{minipage}{1.15\textwidth} 
{\footnotesize Coverage probability and median length of nominal 90\% confidence intervals at different horizons. AR(1) model with $\rho \in \lbrace 0,.5,.95,1\rbrace $, $T=240$, i.i.d. standard normal innovations. 5,000 Monte Carlo repetitions; 2,000 bootstrap iterations.} 
\end{minipage}
\end{table}
\end{landscape}
}

Consistent with our theoretical results, the bootstrap version of lag-augmented local projection (column 1) achieves coverage close to the nominal level in almost all cases, whereas the competing procedures either under-cover or return impractically wide confidence intervals. In contrast, non-augmented LP (columns 3 and 4) exhibits larger coverage distortions in almost all cases. As is well known, textbook AR delta method confidence intervals (column 6) severely under-cover when $\rho>0$ and the horizon is even moderately large.

It is only when both $\rho=1$ and $h \geq 36$ that lag-augmented local projection exhibits serious coverage distortions, again consistent with our theory. However, even in these cases, the coverage distortions are similar to or less pronounced than those for non-augmented LP and for delta method AR inference.

Although the \citet{Inoue2020} lag-augmented AR bootstrap confidence interval (column 5) achieves correct coverage for $\rho>0$ at all horizons, this interval is extremely wide in the problematic cases where $\rho$ is close to 1 and the horizon $h$ is intermediate or long. We explain this fact theoretically in \cref{sec:comp}. Confidence intervals with median width greater than 1 would appear to be of little practical use, since the true impulse response parameter is bounded above by 1 in the AR(1) model.\footnote{In the AR(1) model, we could intersect all confidence intervals with the interval $[-1,1]$. In this case, the median length of the \citet{Inoue2020} confidence interval is close to 1, cf. \cref{sec:comp_details_laarboot}.} Note also that the \citet{Inoue2020} interval severely under-covers when $\rho=0$ at all even (but not odd) horizons $h$, as explained theoretically in \cref{sec:comp}.\footnote{\alert{\citet{Inoue2020} assume $\rho \neq 0$ and discuss why this restriction is necessary in their case.}}

Although outperformed by bootstrap procedures, the lag-augmented local projection delta method interval (column 2) performs well among the group of delta method procedures. Its coverage distortions are much less severe than textbook AR delta method inference (column 4) and non-augmented LP inference with HAR standard errors (column 6). Recall that the lag-augmented LP confidence interval is at least as easy to compute as these other delta method confidence intervals. The reason why the bootstrap improves on the coverage properties of the delta method procedures is related to the well-known finite-sample bias of AR and LP estimators \citep{Kilian1998, Herbst2020}.\footnote{Our bootstrap implementation of \emph{non-augmented} LP also appears to be quite effective at correcting the most severe coverage distortions of the delta method procedure.}

\cref{tab:TableMC} illustrates the fact that the robustness of lag-augmented local projection inference entails an efficiency loss relative to AR inference when $\rho$ is well below 1, although this loss is not large in absolute terms. In \emph{percentage} terms, local projection confidence intervals are much wider than AR-based confidence intervals when $\rho \ll 1$ and the horizon $h$ is intermediate or long, since AR procedures mechanically impose that the impulse response function tends to 0 geometrically fast with the horizon. Yet, in \emph{absolute} terms, the median length of the LP confidence intervals is not so large as to be a major impediment to applied research. The relative efficiency of lag-augmented LP vs. non-augmented LP cannot be ranked and depends on the DGP and on the horizon. When $\rho$ is close to 1, lag-augmented LP intervals are sometimes (much) narrower than lag-augmented AR intervals. We analytically characterize the various efficiency trade-offs in \cref{sec:comp_details_releff}.

\alert{\cref{sec:sim_further} shows that the preceding qualitative conclusions extend to richer models. There we consider a bivariate VAR(4) model with varying degrees of persistence, as well as two empirically calibrated VAR(12) models with four or five observables.}

\section{Comparison With Other Inference Procedures}
\label{sec:comp}
The simulations and theoretical results in this paper suggest that lag-augmented local projection is the only known \alert{confidence interval procedure} that achieves uniformly valid coverage over the DGP and over a wide range of horizons, \alert{while preserving reasonable average length and remaining computationally feasible in realistic settings}. However, the simulations also suggest that lag-augmented local projection inference is less efficient than standard AR inference when the data is stationary. In this section we discuss in more detail the coverage and length properties of alternative confidence interval procedures for impulse responses. We review the well-known drawbacks of textbook AR inference, provide new results on the relative length of lag-augmented LP vs. non-augmented LP and lag-augmented AR, and discuss the computational challenges of the AR grid bootstrap. We refer the reader back to the small-scale simulation study in \cref{sec:sim_ar1} for illustrations of the following arguments.

\paragraph{Textbook Autoregressive Inference.}
The uniformity result (\ref{eqn:coverage}) for lag-augmented LP stands in stark contrast to textbook AR inference on impulse responses, which suffers from several well-known issues. First, for the standard OLS AR estimator, the usual asymptotic normal limiting theory is invalid when the derivative of the impulse response parameter with respect to the AR coefficients has a singular Jacobian matrix. In the AR(1) model, this occurs \alert{at all horizons $h\geq 2$} in the white noise case $\rho=0$ \citep{Benkwitz2000}. Second, as with non-augmented LP, textbook AR inference is \alert{not uniformly valid when the data is nearly non-stationary, unless one further restricts the parameter space} (\citealp[Remark 2.5]{Phillips1998}; \citealp{Inoue2002}; \citealp[Remark 3, p. 455]{Inoue2020}).\footnote{\alert{This is well known in the AR(1) model. In the AR(2) model, a non-normal limit arises at $h=2$ when there is a unit root and the autoregressive coefficients are equal \citep[Remark 3, p. 455]{Inoue2020}.}} Third, pre-testing for the presence of a unit root does not yield uniformly valid inference and can lead to poor finite sample performance \citep[e.g.,][p. 1412]{Mikusheva2007}. Fourth, plug-in AR inference with normal critical values must necessarily break down at medium-long horizons $h=h_T \propto T^{1/2}$ and at long horizons $h_T \propto T$, due to the severe nonlinearity of the impulse response transformation at such horizons \citep[Section 4.3]{Mikusheva2012}. \citet{Wright2000} and \citet{Pesavento2006,Pesavento2007} construct confidence intervals for persistent processes at long horizons $h=h_T \propto T$ by inverting the non-standard AR limit distribution, but these tailored procedures do not work uniformly over the parameter space or over the horizon.

The severe under-coverage of delta method AR inference is starkly illustrated in \cref{sec:sim_ar1} (see column 6 of \cref{tab:TableMC}). As discussed in detail by \cite{Inoue2020}, standard bootstrap approaches to AR inference do not solve all the uniformity issues.

We must emphasize, however, that if we restrict attention to stationary processes and short-horizon impulse responses, the standard OLS AR impulse response estimator is more efficient than lag-augmented LP. Hence, there is a trade-off between efficiency in benign settings and robustness to persistence and longer horizons, as is also clear in the simulation results in \cref{sec:sim_ar1}. We expand upon the efficiency properties of the standard AR estimator in \cref{sec:comp_details_releff}.

\paragraph{Lag-Augmented AR Inference.}
The above-mentioned non-uniformity of the textbook AR inference method in the case of near-non-stationary data can be remedied by lag augmentation \citep{Inoue2020}. In the case of an AR(1) model, the lag-augmented AR estimator $\hat{\beta}_\text{ARLA}(h)$ is given by $\hat{\rho}_1^h$, where $(\hat{\rho}_1,\hat{\rho}_2)$ are the OLS coefficients from a regression of $y_t$ on $(y_{t-1},y_{t-2})$ (i.e., we estimate an AR(2) model). The intuition why this guarantees a normal limiting distribution even in the unit root case is the same as in \cref{sec:ar1_intuition}. Lag-augmented AR and lag-augmented LP coincide at horizon $h=1$, but not at longer horizons. Lag augmentation involves a loss of efficiency: The lag-augmented AR estimator is strictly less efficient than the non-augmented AR estimator except when the true process is white noise (see \cref{sec:comp_details_releff}). Note that lag augmentation by itself does not solve the above-mentioned issues that occur when the Jacobian of the impulse response transformation is singular, or when doing inference at medium-long or long horizons.\footnote{The AR(1) simulations in \cref{sec:sim_ar1} show that the coverage of the \citet{Inoue2020} confidence interval is 0 at all \emph{even} horizons when $\rho=0$. This is because the true impulse response is 0, but the bootstrap samples of $\hat{\rho}_1^h$ are all strictly positive. Their procedure achieves uniformly correct coverage at \emph{odd} horizons.}

The bootstrap confidence interval for lag-augmented AR proposed by \citet{Inoue2020} has valid coverage even at long horizons. Specifically, \citet{Inoue2020} show that the \emph{Efron} bootstrap confidence interval---applied to recursive AR bootstrap samples of $\hat{\beta}_\text{ARLA}(h)$---has valid coverage even at long horizons $h=h_T \propto T$, as long as the largest autoregressive root is bounded away from 0.\footnote{For intuition, consider the AR(1) case. The Efron bootstrap preserves monotonic transformations, and the bootstrap transformation $\beta(\rho,h)=\rho^h$ is monotonic (if we restrict attention to $\rho \in (0,1]$ or $\rho \in [-1,0)$). Hence, the Efron confidence interval is valid for $\rho^h$ if it is valid for $\rho$ itself. In more general VAR($p$) models, the same argument can be applied at long horizons, since here only the largest autoregressive root matters for impulse responses (if the roots are well-separated).}

Unfortunately, we show in \cref{sec:comp_details_laarboot} that the expected length of the lag-augmented AR interval is prohibitively large when the data is persistent and the horizon is long. Precisely, in the case of an AR(1) model, $\hat{\beta}_\text{ARLA}(h)=\hat{\rho}_1^h$ is \emph{inconsistent} for sequences of DGPs $\rho=\rho_T$ and horizons $h=h_T$ such that $h_{T} \propto T^{\eta}, \eta \in [1/2,1]$, and $h_T(1-\rho_T) \to a \in [0,\infty)$. The reason is that the lag-augmented coefficient estimator $\hat{\rho}_1$ converges at rate $T^{-1/2}$ even in the unit root case, implying that the estimation error in $\hat{\rho}_1$ is not negligible when raising the estimator to a power of $h=h_T$. This implies that the Efron bootstrap confidence interval is inconsistent (i.e., its length does not shrink to 0 in probability) for such sequences $\rho_T$ and $h_T$. In fact, when $\eta>1/2$,  the width of the confidence interval for the $h_{T}$ impulse response is almost equal to the entire positive part of the parameter space $[0,1]$ with probability equal to the nominal confidence level. This contrasts with the lag-augmented LP confidence interval, which is consistent for any sequence $\rho_T \in [-1,1]$ and any sequence $h_T$ such that $h_T/T\to 0$. The large width of the \cite{Inoue2020} interval is illustrated in the simulations in \cref{sec:sim_ar1} (see the second-to-last column in \cref{tab:TableMC}).

Interestingly, \emph{if we restrict attention to stationary processes and short horizons, the relative efficiency of lag-augmented AR and lag-augmented LP inference is ambiguous}. In the context of a stationary, homoskedastic AR(1) model with a fixed horizon $h$ of interest, \cref{fig:se_indiff} shows that lag-augmented AR is more efficient than lag-augmented LP when $\rho$ is small or when the horizon $h$ is large, and vice versa. For any horizon $h$, there exists some cut-off value for $\rho \in (0,1)$, above which lag-augmented LP is more efficient. Intuitively, the nonlinear impulse response transformation $\rho \mapsto \rho^h$ is highly sensitive to values of $\rho$ near 1 whenever $h$ is large, which compounds the effects of estimation error in $\hat{\rho}$, whereas LP is a purely linear procedure.

\begin{figure}[t]
\centering
\includegraphics[width=0.8\linewidth]{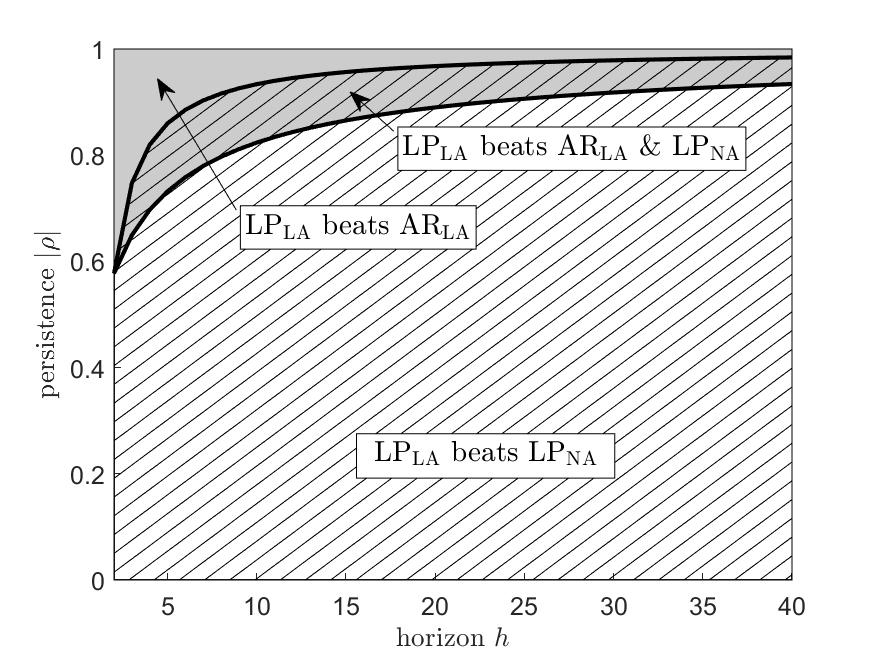}
\caption{Efficiency ranking of three different estimators of the fixed impulse response $\beta(\rho,h) = \rho^h$ in the homoskedastic AR(1) model: lag-augmented LP ($\text{LP}_\text{LA}$), non-augmented LP ($\text{LP}_\text{NA}$), and lag-augmented AR  ($\text{AR}_\text{LA}$). Gray area: combinations of $(|\rho|,h)$ for which $\text{LP}_\text{LA}$ is more efficient than $\text{AR}_\text{LA}$. Thatched area: $\text{LP}_\text{LA}$ is more efficient than $\text{LP}_\text{NA}$. See \cref{sec:comp_details_releff} for analytical derivations of the indifference curves (thick lines).} \label{fig:se_indiff}
\end{figure}

\paragraph{AR Grid Bootstrap and Projection.}
The grid bootstrap of \citet{Hansen1999} represents a computationally intensive approach to doing valid inference at fixed and long horizons, regardless of persistence, but it is invalid at intermediate horizons, as shown by \citet{Mikusheva2012}. The grid bootstrap is based on test inversion, so it requires running an autoregressive bootstrap on each point in a fine grid of potential values for the impulse response parameter of interest. It also requires estimating a constrained OLS estimator that imposes the hypothesized null on the impulse response at each point in the grid. Recall that lag-augmented LP inference is computationally simple and valid at any horizon $h=h_T$ satisfying $h_T/T \to 0$. However, in the case of unit roots and long horizons $h_T \propto T$, lag-augmented LP inference with normal critical values is not valid, while the grid bootstrap is valid \citep{Mikusheva2012}.

Another computationally intensive approach is to form a uniformly valid confidence set for the AR parameters and then map it into a confidence interval for impulse responses by projection. Although doable in the AR(1) model, this approach would appear to be computationally infeasible and possibly highly conservative in realistic VAR($p$) settings, unlike lag-augmented LP (see \cref{sec:var}).

\paragraph{Other Local Projection Approaches.}
Non-augmented LP is not robust to non-stationarity, as already discussed in \cref{sec:ar1_intuition}. \emph{If the data is stationary and the horizon $h$ is fixed, the relative efficiency of non-augmented LP and lag-augmented LP is generally ambiguous}, as shown in \cref{fig:se_indiff} in the case of a homoskedastic AR(1) model. \alert{There are two competing forces. On the one hand, as shown in \cref{sec:ar1_intuition}, non-augmented LP uses the regressor $y_t$, which has higher variance than the effective regressor $u_t$ in the lag-augmented case. By itself, this suggests that non-augmented LP should be more efficient. On the other hand, absent lag augmentation, the LP regression scores are serially correlated and thus have a larger long-run variance. On balance, \cref{sec:comp_details_releff} shows that lag-augmented LP is relatively more efficient the smaller is $\rho$ and the larger is $h$.}

In some empirical settings, the researcher may directly observe the autoregressive innovation, or some component of the innovation, for example by constructing narrative measures of economic shocks \citep{Ramey2016}. For concreteness, consider the AR(1) model \eqref{eqn:ar1} and assume we observe the innovation $u_t$. In this case, it is common in empirical practice to simply regress $y_{t+h}$ on $u_t$, without controls. Although this strategy provides consistent impulse response estimates when the data is stationary, it is inefficient relative to lag-augmented LP, since the latter approach additionally controls for the variable $y_{t-1}$, which would otherwise show up in the error term in the representation \eqref{eqn:ytph_decomp2}. Thus, lag augmentation is desirable on robustness and efficiency grounds even if some shocks are directly observed.

\paragraph{Summary.}
Existing and new theoretical results confirm the main message of our simulations in \cref{sec:sim_ar1}: Lag-augmented LP is the only known procedure that is computationally feasible in realistic problems and can be shown to have valid coverage under a wide range of DGPs and horizon lengths, without achieving such valid coverage by returning a confidence interval that is impractically wide. This robustness does come at the cost of a loss of efficiency relative to non-robust AR methods. However, the efficiency loss is large in \emph{relative} terms only in stationary, short-horizon cases, where lag-augmented LP confidence intervals do well in \emph{absolute} terms, as illustrated in \cref{sec:sim_ar1}. Based on these results, we believe that it is only in the case of highly persistent data and very long horizons $h=h_T \propto T$ that the use of alternative robust procedures should be considered, such as the computationally demanding AR grid bootstrap.

\section{General Theory for the VAR(\texorpdfstring{$p$}{p}) Model}
\label{sec:var}
This section presents the inference procedure and theoretical uniformity result for a general VAR($p$) model. In this case, the lag-augmented LP procedure controls for $p$ lags of all the time series that enter into the VAR model. We follow \citet{Mikusheva2012} and \citet{Inoue2020} in assuming that the lag length $p$ is finite and known. We also assume that the VAR process has no deterministic dynamics for simplicity. See \cref{sec:conc} for further discussion of these assumptions.

\subsection{Model and Inference Procedure} \label{sec:var_model}

Consider an $n$-dimensional VAR($p$) model for the data $y_t=(y_{1,t},\dots,y_{n,t})'$:
\begin{equation} \label{eqn:var}
y_{t} = \sum_{\ell=1}^p A_\ell y_{t-\ell} +  u_{t}, \quad t=1,2,\dots,T,\quad y_0=\dots=y_{1-p}=0,
\end{equation}
Let $A \equiv (A_1, \ldots, A_{p})$ denote the $n \times np$ matrix collecting all the autoregressive coefficients. The assumption of zero pre-sample initial conditions $y_0=\dots=y_{1-p}=0$ is made for notational simplicity and can be relaxed, as discussed below in the remarks after \cref{thm:var_lp_inference}. As in the AR(1) case, we assume that the $n$-dimensional innovation process $\lbrace u_{t} \rbrace$ satisfies the strengthening of the martingale difference condition in \cref{asn:u_mds} (which from now on will refer to the vector process $\lbrace u_t \rbrace$). 

We seek to do inference on a scalar function of the reduced-form impulse responses of the VAR model. Generalizations to \emph{structural} impulse responses and \emph{joint} inference require more notation but are otherwise straight-forward, see \cref{sec:conc}. Let $\beta_{i}(A,h)$ denote the $n \times 1$ vector containing each of variable $i$'s reduced-form impulse responses at horizon $h \geq 0$. Without loss of generality, we focus on the impulse responses of the first variable $y_{1,t}$. Thus, we seek a confidence interval for the scalar parameter $\nu'\beta_{1}(A,h)$, where $\nu \in \mathbb{R}^n \backslash \lbrace 0 \rbrace$ is a user-specified vector. For example, the choice $\nu = e_j$ (the $j$-th unit vector) selects the horizon-$h$ response of $y_{1,t}$ with respect to the $j$-th reduced-form innovation $u_{j,t}$.

Local projection estimators of impulse responses are motivated by the representation
\begin{equation} \label{eqn:var_long_regression}
y_{1,t+h} = \beta_1(A,h)' y_{t} + \sum_{\ell=1}^{p-1} \delta_{1,\ell}(A,h)^{\prime}y_{t-\ell}  + \xi_{1,t}(A,h),
\end{equation} 
see \citet{Jorda2005} and \citet[Chapter 12.8]{Kilian2017}. Here $\delta_{1,\ell}(A,h)$ is an $n \times 1$ vector of regression coefficients that can be obtained by iterating on the VAR model \eqref{eqn:var}. The model-implied multi-step forecast error in this regression is
\begin{equation}
\xi_{1,t}(A,h) \equiv  \sum_{\ell = 1}^{h} \beta_{1}(A,h-\ell)' u_{t+\ell}.
\end{equation}

\paragraph{Multivariate Lag-Augmented Local Projection.}
The lag-augmented LP estimator corresponding to the VAR model \eqref{eqn:var} is motivated by  \eqref{eqn:var_long_regression}. We regress $y_{1,t+h}$ on the $n$ variables $y_{t}$, using the $np$ variables $(y'_{t-1}, \ldots, y'_{t-p})$ as additional controls. According to equation \eqref{eqn:var_long_regression}, the population regression coefficients on the last $n$ control variables $y_{t-p}$ equal zero. Thus, we are including one additional lag in the estimation of the impulse response coefficients. Given any horizon $h \in \mathbb{N}$, the lag-augmented LP estimator $\hat{\beta}_1(h)$ of $\beta_1(A,h)$ is given by the vector of coefficients on $y_{t}$ in the regression of $y_{1,t+h}$ on $x_{t} \equiv (y_{t}', y_{t-1}', \ldots, y_{t-p}')'$:
\begin{equation}
\begin{pmatrix} \label{eqn:LPestimates_VAR}
\hat{\beta}_1(h) \\
\hat{\gamma}_1(h)
\end{pmatrix} \equiv \left(\sum_{t=1}^{T-h} x_tx_t'\right)^{-1}\sum_{t=1}^{T-h}x_t y_{1,t+h},
\end{equation}
where $\hat{\beta}_{1}(h)$ is a vector of dimension $n \times 1$. 

The usual (Eicker-Huber-White) heteroskedasticity-robust standard error for $\nu'\hat{\beta}_1(h)$ is defined as
\[\hat{s}_1(h,\nu) \equiv \frac{1}{T-h}\left\lbrace \nu'\hat{\Sigma}(h)^{-1} \left(\sum_{t=1}^{T-h} \hat{\xi}_{1,t}(h)^2\hat{u}_t(h)\hat{u}_t(h)' \right) \hat{\Sigma}(h)^{-1}\nu\right\rbrace^{1/2},\]
where
\[\hat{\xi}_{1,t}(h) \equiv y_{1,t+h} - \hat{\beta}_1(h)' y_{t} - \hat{\gamma}_{1}(h)' X_{t}, \quad X_{t} \equiv (y_{t-1}' , \ldots, y_{t-p}')', \]
\[ \hat{u}_{t}(h) \equiv y_{t} -\hat{A}(h)X_{t}, \quad \hat{A}(h) \equiv \left( \sum_{t=1}^{T-h} y_{t}X_{t}' \right) \left ( \sum_{t=1}^{T-h} X_{t} X_{t}'  \right)^{-1}, \]
and
\[\hat{\Sigma}(h) \equiv \frac{1}{T-h}\sum_{t=1}^{T-h} \hat{u}_t(h)\hat{u}_t(h)'.\]
The $1-\alpha$ confidence interval for $\nu'\beta_1(A,h)$ is defined as
\[\hat{C}_1(h,\nu,\alpha) \equiv \left[ \nu'\hat{\beta}_1(h)-z_{1-\alpha/2}\: \hat{s}_1(h,\nu)\:,\: \nu'\hat{\beta}_1(h)+z_{1-\alpha/2}\: \hat{s}_1(h,\nu) \right].\]

\paragraph{Parameter Space.}
We consider a class of VAR processes with possibly multiple unit roots combined with arbitrary stationary dynamics. Specifically, we will prove that the confidence interval $\hat{C}_1(h,\nu,\alpha)$ has uniformly valid coverage over the following parameter space. Let $\|M\| \equiv \sqrt{\tr(M'M)}$ denote the Frobenius matrix norm, and let $I_n$ denote the $n \times n$ identity matrix.

\begin{defn}[VAR parameter space] \label{dfn:param_space}
Given constants $a \in [0,1)$, $C>0$, and $\epsilon \in (0,1)$, let $\mathcal{A}(a,C,\epsilon)$ denote the space of autoregressive coefficients $A=(A_1, \ldots, A_{p})$ such that the associated $p$-dimensional lag-polynomial $A(L)=I_n-\sum_{\ell=1}^p A_\ell L^\ell$ admits the factorization
\begin{equation}
A(L) = B(L) (I_n - \diag(\rho_1, \ldots, \rho_n) L),
\end{equation}
where $\rho_i \in [a-1,1-a]$ for all $i=1,\dots,n$, and $B(L)$ is a lag polynomial of order $p-1$ with companion matrix $\mathbf{B}$ satisfying $\| \mathbf{B}^\ell \| \leq C (1-\epsilon)^\ell$ for all $\ell = 1,2,\dots$.\footnote{See \cref{sec:var_proof_main} for the standard definition of a companion matrix.}
\end{defn}

This parameter space contains any stationary VAR process (for sufficiently small $a,\epsilon$ and sufficiently large $C$) as well as many---but not all---non-stationary processes. Lag polynomials $A(L)$ in this parameter space imply that the process $\lbrace y_t \rbrace$ can be written in the form $y_t = \diag(\rho_1,\dots,\rho_n)y_{t-1} + \tilde{y}_t$, where $\tilde{y}_t \equiv B(L)^{-1}u_t$ is a stationary process whose impulse responses at horizon $\ell$ decay at the geometric rate $(1-\epsilon)^\ell$. We allow all the roots $\rho_1,\dots,\rho_n$ to be potentially close to or equal to 1. \citet[Section 4.2]{Mikusheva2012} considers the same class of processes but with $\rho_2=\dots=\rho_n=0$. We are not aware of other uniform inference results that allow multiple near-unit roots. Although the parameter space in \cref{dfn:param_space} appears more restrictive than the local-to-unity framework of \citet[Eqn. 2]{Phillips1988}, we argue below that our uniform coverage result applied to the parameter space $\mathcal{A}(a,C,\epsilon)$ immediately implies an extended result that also covers processes with cointegration among the control variables $y_{2,t},\dots,y_{n,t}$. However, we do impose the restriction that the response variable of interest $y_{1,t}$ has at most one root near unity, as in \citet{Wright2000}, \citet{Pesavento2006}, \citet{Mikusheva2012}, and \citet{Inoue2020}.

\subsection{Additional Assumptions} 
\label{sec:var_asn}

Our main result requires two further technical assumptions in addition to \cref{asn:u_mds}. Let $\lambda_{\min}(M)$ denote the smallest eigenvalue of a symmetric positive semidefinite matrix $M$.

\begin{asn} \label{asn:var_u_reg}
\leavevmode
\begin{enumerate}[i)] 
\item \label{itm:var_asn_u_bounds} $E(\|u_t\|^8)<\infty$,  and there exists $\delta > 0$ such that $\lambda_{\min}(E[u_tu_t' \mid \lbrace u_s \rbrace_{s<t}]) \geq \delta$ almost surely.
\item \label{itm:var_asn_u2_cum} The process $\lbrace u_t \otimes u_t \rbrace$ has absolutely summable cumulants up to order 4.
\end{enumerate}
\end{asn}

\noindent Part (\ref{itm:var_asn_u_bounds}) of \cref{asn:var_u_reg} is a common requirement for consistent estimation of regression standard errors with possibly heteroskedastic residuals. Part (\ref{itm:var_asn_u2_cum}) is a standard weak dependence restriction on the second moments of $u_t$ \citep[Chapter 2.6]{Brillinger2001}.

We will write $\rho(A)=(\rho_1(A),\dots,\rho_n(A))'$ to represent any of the possible vectors of roots $\rho_1,\dots,\rho_n$ corresponding to a  collection of autoregressive coefficients $A=(A_1, \ldots, A_{p}) \in \mathcal{A}(0,C,\epsilon)$. This is a slight abuse of notation, since the mapping from $A(L)$ to $\rho_i$'s is one-to-many. Define $g(\rho,h)^2 \equiv \min\lbrace \frac{1}{1-|\rho|},h\rbrace$ and $\rho_i^*(A,\epsilon) \equiv  \max\lbrace  |\rho_{i}(A)|, 1-\epsilon/2 \rbrace$. Define also the $np \times np$ diagonal matrix $G(A,h,\epsilon) \equiv I_p \otimes \diag(g(\rho_1^*(A,\epsilon) ,h),\dots,g(\rho_n^*(A,\epsilon),h))$.

\begin{asn} \label{asn:var_XpX}
For any $C>0$ and $\epsilon \in (0,1)$,
\[\lim_{K \to \infty}  \lim_{T\to\infty} \inf_{A \in \mathcal{A}(0,C,\epsilon)} P_A\left( \lambda_{\min}\left(G(A,T,\epsilon)^{-1}\left[\frac{1}{T}\sum_{t=1}^T X_tX_t'\right]G(A,T,\epsilon)^{-1}  \right) \geq 1/K \right) = 1.\]
\end{asn}
\noindent This high-level assumption ensures that the properly scaled (matrix) ``denominator'' in the VAR OLS estimator $\hat{A}(h)$ is uniformly non-singular asymptotically, so the estimator is uniformly well-defined with high probability in the limit. Hence, the assumption is essentially necessary for our result. 

How can \cref{asn:var_XpX} be verified? $G(A_T,T,\epsilon)^{-1}\left[\frac{1}{T}\sum_{t=1}^T X_tX_t'\right]G(A_T,T,\epsilon)^{-1}$ is known to converge in distribution in a \emph{pointwise} sense to an almost surely positive definite (perhaps stochastically degenerate) random matrix under stationary, local-to-unity, or unit root sequences $\lbrace A_T \rbrace$ \citep[e.g.,][]{Phillips1988,Hamilton1994}.\footnote{Note that the diagonal entries of $G(A,T,\epsilon)^{-1}$ are constants for stationary VAR coefficient matrices $A$, whereas these diagonal entries are proportional to $T^{-1/2}$ under local-to-unity or unit root sequences.} \cref{asn:var_XpX} requires that such convergence obtains for \emph{all} possible sequences $\lbrace A_T \rbrace$. In \cref{sec:var_XpX_ar1} we illustrate how the assumption can be verified in the AR(1) model under an additional weak condition on the innovation process.

\subsection{Main Result}

We now state the result that the LP estimator $\nu'\hat{\beta}_1(h)$ is asymptotically normally distributed uniformly over the parameter space in \cref{dfn:param_space}, even at long horizons $h$. Let $P_A$ denote the probability measure of the data $\lbrace y_t \rbrace$ when it is generated by the VAR($p$) model \eqref{eqn:var} with coefficients $A \in \mathcal{A}(a,C,\epsilon)$. The distribution of the innovations $\lbrace u_t \rbrace$ is fixed.

\begin{prop} \label{thm:var_lp_inference}
Let \cref{asn:u_mds,asn:var_u_reg,asn:var_XpX} hold. Let $C>0$ and $\epsilon \in (0,1)$.
\begin{enumerate}[i)]
\item \label{itm:var_lp_inference_stat} Let $a \in (0,1)$. For all $x \in \mathbb{R}$,
\[\sup_{A \in \mathcal{A}(a,C,\epsilon)} \sup_{1 \leq h \leq (1-a)T} \left| P_A\left(\frac{\nu'[\hat{\beta}_1(h)-\beta_1(A,h)]}{\hat{s}_1(h,\nu)} \leq x\right) - \Phi(x) \right| \to 0.\] 
\item \label{itm:var_lp_inference_all} Consider any sequence $\lbrace \bar{h}_T \rbrace$ of nonnegative integers such that $\bar{h}_T < T$ for all $T$ and $\bar{h}_T/T \to 0$. Then for all $x \in \mathbb{R}$,
\[\sup_{A \in \mathcal{A}(0,C,\epsilon)} \sup_{1 \leq h \leq \bar{h}_T} \left| P_A\left(\frac{\nu'[\hat{\beta}_1(h)-\beta_1(A,h)]}{\hat{s}_1(h,\nu)} \leq x\right) - \Phi(x) \right| \to 0.\]
\end{enumerate}
\end{prop}
\begin{proof}
See \cref{sec:var_proof_main}.
\end{proof}
The uniform asymptotic normality established above immediately implies that the confidence interval $\hat{C}_1(h,\nu,\alpha)$ has uniformly valid coverage asymptotically. Part (\ref{itm:var_lp_inference_stat}) considers stationary VAR processes whose largest roots are bounded away from 1; then inference is valid even at long horizons $h=h_T \propto T$. Part (\ref{itm:var_lp_inference_all}) allows all or some of the $n$ roots $\rho_1,\dots,\rho_n$ to be near or equal to 1, but then we require $h_T/T \to 0$.

\paragraph{Remarks.}
\begin{enumerate}
\item The proof of \cref{thm:var_lp_inference} shows that the uniform convergence rate of $\hat{\beta}_1(h_T)$ is $O_p((h_T/T)^{1/2})$ if $h_T/T \to 0$. This rate may be slower than that of the possibly super-consistent non-augmented LP estimator, which is the price to pay for uniformity. If we restrict attention to the stationary parameter space $\mathcal{A}(a,C,\epsilon)$, $a>0$, the convergence rate of $\hat{\beta}_1(h_T)$ is $O_p(T^{-1/2})$ provided that $h_T \leq (1-a)T$.

\item There are three main challenges in establishing the uniform validity of local projection inference.
\begin{enumerate}[a)]
\item The variance of the regression residual $\xi_{1,t}(A,h)$ is increasing in the horizon $h$ and also depends on $A$. Thus, the simplest laws of large numbers and central limit theorems for stationary processes do not apply. We instead apply a central limit theorem for martingale difference sequences and derive uniform bounds on moments of relevant variables. The central limit theorem is delicate, since the regression scores $\xi_{1,t}(A,h)u_t$ are not a martingale difference sequence with respect to the natural filtration generated by past $u_t$'s. However, it is possible to ``reverse time'' in a way that makes the scores a martingale difference sequence with respect to an alternative filtration, see the proof of the auxiliary \cref{thm:var_clt}.
\item To handle both unit roots, stationary processes, and everything in between, we must consider various kinds of sequences of drifting parameters $A=A_T$, following the general logic of \citet{Andrews2019}. This is primarily an issue when showing consistency of the standard error $\hat{s}_1(h,\nu)$, which requires deriving the convergence rates of the various estimators along drifting parameter sequences. We do this by explicit calculation of moment bounds that are uniform in the both the DGP and the horizon.

\item Our proof requires bounds on the rate of decay of impulse response functions that are uniform in both the DGP and the horizon. Though the AR(1) case is trivial due to the monotonically decreasing exponential functional form $\beta(\rho,h)=\rho^h$, the bounds for the general VAR($p$) case require more work, see especially \cref{thm:var_bound_for_IRFs_A} in \cref{sec:var_se_proof}. These results may be of independent interest.
\end{enumerate}

\item \cref{thm:var_lp_inference} does not cover the case where $h \propto T$ and some of the roots $\rho_i$ are local-to-unity or equal to unity. Simulation evidence and analytical calculations along the lines of \citet{Hjalmarsson2020} strongly suggest that even in the AR(1) model the asymptotic normality of lag-augmented local projections does \emph{not} go through when $\rho=1$ and $h = \kappa T$ for $\kappa \in (0,1)$. Indeed, in this case the sample variance of the regression scores $\xi_t(\rho,h)u_t$ appears not to  converge in probability to a constant, thus violating the conclusion of the key auxiliary \cref{thm:var_se_infeas} below. As discussed in \cref{sec:comp}, the behavior of plug-in autoregressive impulse response estimators is also non-standard when $\rho \approx 1$ and $h \propto T$.

\item A corollary of our main result is that we can allow for cointegrating relationships to exist among the control variables $y_{2,t},\dots,y_{n,t}$. This is because both the LP estimator and the reduced-form impulse responses are equivariant with respect to non-singular linear transformations of these $n-1$ variables. For example, consider a 3-dimensional process $(y_{1,t},y_{2,t},y_{3,t})$ that follows a VAR model in the parameter space in \cref{dfn:param_space} with $\rho_2=1,\rho_3=0$. Now consider the transformed process $(y_{1,t},\tilde{y}_{2,t},\tilde{y}_{3,t}) = (y_{1,t}, y_{2,t} + y_{3,t}, -y_{2,t} + y_{3,t})$. The variables $\tilde{y}_{2,t}$ and $\tilde{y}_{3,t}$ are cointegrated with cointegrating vector $(1,1)'$. Since $(\tilde{y}_{2,t},\tilde{y}_{3,t})$ is a non-singular linear transformation of $(y_{2,t},y_{3,t})$, the conclusions of \cref{thm:var_lp_inference} apply also to the transformed data vector.

\item If the vector of innovations $u_t$ were observed, an alternative estimator would regress $y_{1,t+h}$ onto $u_t$ and $y_{t-1},\dots,y_{t-p}$. As discussed in \cref{sec:ar1_intuition}, this estimator is numerically equivalent with $\hat{\beta}_1(h)$, so the uniformity result carries over.

\item It is easily verified in our proofs that, rather than initializing the process at zero, we can allow the initial conditions $y_0,\dots,y_{1-p}$ to be random variables that are independent of the innovations $\lbrace u_t\rbrace_{t \geq 1}$, as long as $E[\|y_\ell\|^4] < \infty$ for $\ell \leq 0$.

\end{enumerate}

\section{Bootstrap Implementation}
\label{sec:boot}
In this section we describe the bootstrap implementation of lag-augmented local projection that we recommend for practical use. We find in simulations that the bootstrap procedure is effective at correcting small-sample coverage distortions. These distortions arise primarily due to the small-sample bias of local projection, which \citet{Herbst2020} show is analogous to the well-known bias of the VAR OLS estimator \citep{Kilian1998}. 

Our baseline algorithm is based on a wild autoregressive bootstrap design, which allows for heteroskedastic VAR innovations \citep{Goncalves2004} as in our theoretical results. Guided by simulation evidence, we construct the bootstrap confidence interval using the equal-tailed percentile-t method, which has a built-in bias correction (\citealp{Kilian1998}; \citealp[Chapter 12.2.6]{Kilian2017}). 

The bootstrap procedure for computing a $1-\alpha$ confidence interval proceeds as follows, assuming a VAR($p$) model:
\begin{enumerate}
    \item Compute the impulse response estimate of interest $\nu'\hat{\beta}_1(h)$ and its standard error $\hat{s}_1(h,\nu)$ by lag-augmented local projection as in \cref{sec:var_model}.
    \item \label{itm:boot_var} Estimate the VAR($p$) model by OLS without lag augmentation. Compute the corresponding VAR residuals $\hat{u}_t$. Bias-adjust the VAR coefficients using the formula in \citet{Pope1990} (this adjustment is optional, but improves finite-sample performance).
    \item Compute the impulse response of interest implied by the VAR model estimated in step \ref{itm:boot_var}. Denote this impulse response by $\nu'\hat{\beta}_\text{1,VAR}(h)$.
    \item For each bootstrap iteration $b=1,\dots,B$:
    \begin{enumerate}[i)]
        \item Generate bootstrap residuals $\hat{u}_t^* \equiv U_t \hat{u}_t$, $t=1,\dots,T$, where $U_t \stackrel{i.i.d.}{\sim} N(0,1)$ are computer-generated random variables that are independent of the data.
        \item Draw a block of $p$ initial observations $(y_1^*,\dots,y_p^*)$ uniformly at random from the $T-p+1$ blocks of $p$ observations in the original data.
        \item Generate bootstrap data $y_t^*$, $t=p+1,\dots,T$, by iterating on the bias-corrected VAR($p$) model estimated in step \ref{itm:boot_var}, using the innovations $\hat{u}_t^*$.
        \item Apply the lag-augmented LP estimator to the bootstrap data $\lbrace y_t^* \rbrace$. Denote the impulse response estimate and its standard error by $\nu'\hat{\beta}(h)^*$ and $\hat{s}_1(h,\nu)^*$, respectively.
        \item Store $\hat{T}_b^* \equiv (\nu'\hat{\beta}_1(h)^*-\nu'\hat{\beta}_\text{1,VAR}(h))/\hat{s}_1(h,\nu)^*$.\footnote{It is critical that the bootstrap t-statistic $\hat{T}_b^*$ is centered at the VAR-implied impulse response $\nu'\hat{\beta}_\text{1,VAR}(h)$ rather than the LP-estimated impulse response $\nu'\hat{\beta}_1(h)$. This is because the former estimate is the pseudo-true parameter in the recursive bootstrap DGP, and the latter estimate differs from the former by an amount that is not asymptotically negligible.}
    \end{enumerate}
    \item Compute the $\alpha/2$ and $1-\alpha/2$ quantiles of the $B$ draws of $\hat{T}_b^*$, $b=1,\dots,B$. Denote these by $\hat{Q}_{\alpha/2}$ and $\hat{Q}_{1-\alpha/2}$, respectively.
    \item Return the percentile-t confidence interval\footnote{It is not valid to use the Efron bootstrap confidence interval based on the bootstrap quantiles of $\hat{\beta}(h)^*$. This is because the bootstrap samples are asymptotically centered around $\hat{\beta}_\text{VAR}(h)$, not $\hat{\beta}(h)$.}
\[[\nu'\hat{\beta}_1(h)-\hat{s}_1(h,\nu)\hat{Q}_{1-\alpha/2},\nu'\hat{\beta}_1(h)-\hat{s}_1(h,\nu)\hat{Q}_{\alpha/2}].\]
\end{enumerate}

Instead of the above recursive VAR design, it is also possible to use the standard fixed-design pairs bootstrap, as in any linear regression with serially uncorrelated scores.\footnote{This is the bootstrap carried out by Stata's {\tt bootstrap} command with standard settings.} In this case, the usual Efron bootstrap confidence interval is valid, like the percentile-t interval. However, simulations suggest that the pairs bootstrap procedure is less accurate in small samples than the above recursive bootstrap design, \alert{mirroring the results in \cite{Goncalves2004} for autoregressive inference}.

Our online code repository implements the above recommended bootstrap procedure, as well as several alternative LP- and VAR-based procedures, see \cref{fn:github}.

\section{Conclusion and Directions for Future Research}
\label{sec:conc}

Local projection inference is already popular in the applied macroeconomics literature. The simple nature of local projections has allowed the methods of causal analysis in macroeconomics to connect with the rich toolkit for program evaluation in applied microeconomics; see for example \citet{Angrist2018}, \citet{Nakamura2018}, \citet{Stock2018}, and \citet{rambachan2019econometric}. We hope the novel results in this paper on the statistical properties of local projections may further this convergence.

\paragraph{Recommendations for Applied Practice.}
The simplicity and statistical robustness of \emph{lag-augmented} local projection inference makes it an attractive option relative to  existing inference procedures. We recommend that applied researchers conduct inference based on lag-augmented local projections with heteroskedasticity-robust (Eicker-Huber-White) standard errors. This procedure can be implemented using  any regression software and has desirable theoretical properties relative to textbook delta method autoregressive inference and to non-augmented local projection methods. In particular, we showed that confidence intervals based on lag-augmented local projections that use robust standard errors with standard normal critical values are uniformly valid over the persistence in the data and for a wide range of horizons. We also suggested a simple bootstrap implementation in \cref{sec:boot}, which seems to achieve even better finite-sample performance.

Conventional VAR-based procedures deliver smaller standard errors than local projections in many cases, but this comes at the cost of fragile coverage, especially at longer horizons. In our opinion, there are only two cases in which the lag-augmented local projection inference method is inferior to competitors: (i) If the data is known to be at most moderately persistent and interest centers on very short impulse response horizons, in which case textbook VAR inference is valid and efficient. (ii) When the data has (near-)unit roots and interest centers on horizons that are a substantial fraction of the sample size, in which case the computationally demanding AR grid bootstrap may be deployed if feasible \citep{Hansen1999,Mikusheva2012}. In all other cases, lag-augmented local projection inference appears to achieve a competitive  trade-off between robustness and efficiency.

How should the VAR lag length $p$ be chosen in practice? Naive pre-testing for $p$ causes uniformity issues for subsequent inference \citep{Leeb2005}. Though we leave the development of a formal procedure for future research (see below), our theoretical analysis yields three insights. First, users of local projection should worry about the choice of $p$ in order to obtain robust inference, just as users of VAR methods do. Second, $p$ should be chosen conservatively, as is conventional in VAR analysis \citep[Chapter 2.6.5]{Kilian2017}. \alert{In our framework there is no asymptotic efficiency cost of controlling for more than $p_0$ lags if the true model is a VAR($p_0$), and the simulation results in \cref{sec:sim_further} confirm that the cost is also small in finite samples.} Third, the logic of \cref{sec:ar1_intuition} suggests that in realistic models where the higher-lag VAR coefficients are relatively small, it is not crucial to get $p$ exactly right: What matters is that we include enough control variables so that the effective regressor of interest approximately satisfies the conditional mean independence condition (\cref{asn:u_mds}).

\paragraph{Directions for Future Research.}
It would be interesting to relax the assumption of a finite lag length $p$ by adopting a VAR($\infty$) framework. We are not aware of existing work on uniform inference in such settings. One possibility would be to base inference on a sieve VAR framework that lets the lag length used for estimation tend to infinity at an appropriate rate as in \citet{Goncalves2007}. A second possibility is to impose \emph{a priori} bounds on the rate of decay of the VAR coefficients, and then take the resulting worst-case bias of finite-$p$ local projection estimators into account when constructing confidence intervals \citep[as in the ``honest inference'' approach of][]{Armstrong_Kolesar:2018}. 

Due to space constraints, we leave a proof of the validity of the suggested bootstrap strategy to future work. It appears straight-forward, albeit tedious, to prove its pointwise validity. Proving uniform validity requires extending the already lengthy proof of \cref{thm:var_lp_inference}.

Several extensions of the results in this paper could be pursued by adopting techniques from the VAR literature. First, the results of \citet{PlagborgMoller2019} suggest straight-forward ways to generalize our results on reduced-form impulse response inference to \emph{structural} inference. Second, our assumption of no deterministic dynamics in the VAR model could presumably be relaxed using standard arguments. Third, by considering linear system estimators rather than single-equation OLS, our results on scalar inference could be extended to simultaneous inference on several impulses \citep{IK2016,MontielOlea2019}. Finally, whereas we adopt a frequentist perspective in this paper, it remains an open question whether local projection inference is relevant from a Bayesian perspective.

\appendix

\numberwithin{asn}{section}

\section{Proof of \texorpdfstring{\cref{thm:var_lp_inference}}{Proposition \ref{thm:var_lp_inference}}} \label{sec:var_proof_main}

\paragraph{Notation.}
We first introduce some additional notation. For $p \geq 1$, the \emph{companion matrix} of the VAR($p$) model \eqref{eqn:var} is the $np \times np$ matrix given by
\begin{equation} \label{eqn:var_companion}
\mathbf{A} =
\begin{bmatrix}
A_1 & A_2 & \ldots & A_{p-1} & A_{p }\\
I_{n} & 0 & \ldots  & 0 & 0 \\
0 & I_n &   & 0& 0 \\
\vdots & & \ddots  & \vdots &\vdots \\
0 &  0 & \ldots & I_n & 0 \\
\end{bmatrix},
\end{equation}
where $A_1, \ldots, A_{p}$ are the slope coefficients of the autoregressive model \citep[p. 25]{Kilian2017}. The companion matrix of a VAR with no lags is defined as a the $n \times n$ matrix of zeros. 

Recall that $\|M\| \equiv \sqrt{\tr(M'M)}$ denotes the Frobenius norm of the matrix $M$. This norm is sub-multiplicative: $\|M_1M_2\| \leq \|M_1\| \times \|M_2\|$. We use $\lambda_{\min}(M)$ to denote the smallest eigenvalue of the symmetric positive semidefinite matrix $M$.

Denote $\Sigma \equiv E(u_tu_t')$, and note that this matrix is positive definite by \cref{asn:var_u_reg}(\ref{itm:var_asn_u_bounds}). Define, for any collection of autoregressive coefficients $A$, for any $h \in \mathbb{N}$, and for an arbitrary vector $w \in \mathbb{R}^n$:
\begin{equation} \label{eqn:var_vph}
v(A,h,w) \equiv \lbrace E[\xi_{1,t}(A,h)^2 (w'u_t)^2] \rbrace^{1/2},
\end{equation}
where
\begin{equation}
\xi_{i,t}(A,h) \equiv  \sum_{\ell = 1}^{h} \beta_{i}(A,h-\ell)'
 u_{t+\ell},\quad i=1,\dots,n.
\end{equation}
The $n \times 1$ vector $\beta_{i}(A,h)$ contains each of variable $i$'s impulse response coefficients at horizon $h \geq 1$:
\begin{equation} \label{eqn:var_IRF_companion}
\beta_{i} (A, h)' \equiv e_i(n)' J \mathbf{A}^{h} J',   
\end{equation}
where $J \equiv [ I_{n}, 0_{n \times n(p-1)} ]$ and $e_{i}(n)$ is the $i$-th column of the identity matrix of dimension $n$.

Finally, recall the notation $\rho_i(A)$, $g(\rho,h)$, $\rho_i^*(A,\epsilon)$, and $G(A,h,\epsilon)$ introduced in \cref{sec:var_asn}.

In the proofs below we simplify notation by omitting the subscript $A$ (which indexes the data generating process) from expectations, variances, covariances, and so on.

\paragraph{Proof.} We have defined the lag-augmented local projection estimator of $\beta_{1}(A,h)$ as the vector of coefficients on $y_{t}$ in the regression of $y_{1,t+h}$ on $y_{t}$ with controls $X_{t} \equiv (y_{t-1}', \ldots, y_{t-p}')$.  By the Frisch-Waugh theorem, we can also obtain the coefficient of interest by regressing $y_{1,t+h}$ on the VAR residuals:

\begin{equation} \label{eqn:var_betahat}
\hat{\beta}_{1}(h) \equiv  \left( \sum_{t=1}^{T-h} \hat{u}_{t}(h) \hat{u}_{t}(h)'   \right)^{-1} \sum_{t=1}^{T-h} \hat{u}_{t}(h) y_{1,t+h},
\end{equation}
where we recall the definitions
\[ \hat{u}_{t}(h) \equiv y_{t} -\hat{A}(h)X_{t}, \quad \hat{A}(h) \equiv \left( \sum_{t=1}^{T-h} y_{t}X_{t}' \right) \left ( \sum_{t=1}^{T-h} X_{t} X_{t}'  \right)^{-1}.\]
Recall also from \eqref{eqn:var_long_regression} that
\begin{eqnarray}
y_{1,t+h} &=&  \beta_1(h, A)' y_{t} + \sum_{\ell=1}^{p-1} \delta_{1,\ell}(A,h)^{\prime}y_{t-\ell}  + \xi_{1,t}(A,h) \nonumber \\   
&=&   \beta_1(h, A)' y_{t} + \gamma_{1}(A,h)'X_{t}  + \xi_{1,t}(A,h) \nonumber \\
&& \textrm{(where the last $n$ entries of $\gamma_{1}(A,h)$ are zero)} \nonumber \\
&=&  \beta_1(h, A)' (y_{t} - A X_{t}  ) + \underbrace{(\beta_1(h, A)' A  + \gamma_{1}(A,h)')}_{\equiv \eta_1(A,h)'}X_{t} + \xi_{1,t}(A,h). \label{eqn:var_long_regression_u} 
\end{eqnarray}
Using the definition \eqref{eqn:var_betahat} of the lag-augmented local projection estimator, we have
\begin{align*}
\hat{\beta}_{1}(h) &=  \left( \sum_{t=1}^{T-h} \hat{u}_{t}(h) \hat{u}_{t}(h)'   \right)^{-1} \sum_{t=1}^{T-h} \hat{u}_{t}(h) y_{1,t+h} \\
&=  \left( \sum_{t=1}^{T-h} \hat{u}_{t}(h) \hat{u}_{t}(h)'   \right)^{-1}  \sum_{t=1}^{T-h} \hat{u}_{t}(h)[u'_t \beta_{1}(A,h) + X'_{t} \eta_1(A,h) + \xi_{1,t}(A,h)] \\
&  \textrm{(by equation \eqref{eqn:var_long_regression_u})} \\
&= \left( \sum_{t=1}^{T-h} \hat{u}_{t}(h) \hat{u}_{t}(h)'   \right)^{-1} \sum_{t=1}^{T-h} \hat{u}_t(h) [u'_t  \beta_1(\rho,h) + \xi_{1,t}(A,h)] \\
&\textrm{(because $\textstyle \sum_{t=1}^{T-h} \hat{u}_t(h)X'_{t}=0$ by definition of $\hat{u}_t(h)$)}\\
&= \beta_1(A,h) + \left( \sum_{t=1}^{T-h} \hat{u}_{t}(h) \hat{u}_{t}(h)'   \right)^{-1}  \sum_{t=1}^{T-h} \hat{u}_t(h) [(u_t-\hat{u}_t(h))'\beta_{1}(A,h)  + \xi_{1,t}(A,h)]  \\
&= \beta_1(A,h) + \left( \sum_{t=1}^{T-h} \hat{u}_{t}(h) \hat{u}_{t}(h)'   \right)^{-1}  \sum_{t=1}^{T-h}  \hat{u}_t(h)\xi_{1,t}(A,h),
\end{align*}
where the last equality uses $u_t - \hat{u}_t(h)=(\hat{A}(h)-A)X_{t}$ and again $\sum_{t=1}^{T-h} \hat{u}_t(h) X'_{t}=0$ by definition of $\hat{u}_t(h)$. Define $\hat{\nu}(h) \equiv \hat{\Sigma}(h)^{-1}\nu$ and $\tilde{\nu} \equiv \Sigma^{-1}\nu$. Then
\begin{align*}
\frac{\nu'[\hat{\beta}_1(h)-\beta_1(A,h)]}{\hat{s}_1(h,\nu)} &=  \frac{\hat{\nu}(h)' \sum_{t=1}^{T-h} \hat{u}_t(h)  \xi_{1,t}(A,h)}{(T-h)\hat{s}_1(h,\nu)} \\
&= \left( \frac{\hat{\nu}(h)'\sum_{t=1}^{T-h} \xi_{1,t}(A,h) u_t }{(T-h)^{1/2}v(A,h,\tilde{\nu})} + \frac{\hat{\nu}(h)'\sum_{t=1}^{T-h} [\hat{u}_t(h)-u_t] \xi_{1,t}(A,h)}{(T-h)^{1/2}v(A,h,\tilde{\nu})}  \right) \\
& \qquad \times \frac{v(A,h,\tilde{\nu})}{(T-h)^{1/2}\hat{s}_1(h,\nu)}.
\end{align*}
Using the drifting parameter sequence approach of \citet{Andrews2019}, both statements (\ref{itm:var_lp_inference_stat}) and (\ref{itm:var_lp_inference_all}) of the proposition follow if we can show the following: For any sequence $\lbrace A_T \rbrace$ of autoregressive coefficients in $\mathcal{A}(0,C,\epsilon)$, and for any sequence $\lbrace h_T \rbrace$ of nonnegative integers satisfying $h_T \leq (1-a)T$ for all $T$ and $g(\max_i \lbrace |\rho_i(A)|\rbrace ,h_T)^2/(T-h_T) \to 0$, we have:

\begin{enumerate}[i)]
\item \label{itm:var_proof_i} $\frac{\sum_{t=1}^{T-h_T} \xi_{1,t}(A_T,h_T) (w'u_t) }{(T-h_T)^{1/2}v(A_T,h_T,w)} \underset{P_{A_T}}{\overset{d}{\to}} N(0,1)$, \quad \textrm{ for any $w \in \mathbb{R}^n \backslash \lbrace 0 \rbrace$}.
\item \label{itm:var_proof_ii} $\frac{(T-h_T)^{1/2}\hat{s}_1(h_T,\nu)}{v(A_T,h_T,\tilde{\nu})} \underset{P_{A_T}}{\overset{p}{\to}} 1$.
\item \label{itm:var_proof_iii} $\frac{\sum_{t=1}^{T-h} [\hat{u}_t(h)-u_t] \xi_{1,t}(A,h)}{(T-h_T)^{1/2}v(A_T,h_T,w)} \underset{P_{A_T}}{\overset{p}{\to}} 0$, \quad \textrm{ for any $w \in \mathbb{R}^n \backslash \lbrace 0 \rbrace$}.
\item \label{itm:var_proof_iv} $\hat{\nu}(h_T) \underset{P_{A_T}}{\overset{p}{\to}}  \nu$.
\end{enumerate}



\noindent Result (\ref{itm:var_proof_i}) follows from \cref{thm:var_clt} below. Result (\ref{itm:var_proof_ii}) follows from \cref{thm:var_se} below. Result (\ref{itm:var_proof_iii}) follows by bounding
\begin{align*}
&\frac{\left\|\sum_{t=1}^{T-h_T} \xi_{1,t}(A_T,h_T)[\hat{u}_t(h_T)-u_t]\right\|}{(T-h_T)^{1/2}v(A_T,h_T,w)} \\
&\leq (T-h_T)^{1/2}\left\|[\hat{A}(h_T)-A_T]G(A_T,T-h_T,\epsilon)\right\| \\
&\qquad \times \left\|\frac{\sum_{t=1}^{T-h_T} G(A_T,T-h_T,\epsilon)^{-1} X_t\xi_{1,t}(A_T,h_T)}{(T-h_T)v(A_T,h_T,w)}\right\| .
\end{align*}
The first factor on the right-hand side above is $O_{P_{A_T}}(1)$ by \cref{thm:var_estim_conv}(\ref{itm:var_estim_conv_iii}) below. The second factor on the right-hand side above tends to zero in probability by \cref{thm:var_estim_conv_numer} below. Thus, result (\ref{itm:var_proof_iii}) follows.

Finally, result (\ref{itm:var_proof_iv}) follows immediately from \cref{thm:var_Sigmahat_conv} below and the fact that $\Sigma$ is positive definite by \cref{asn:var_u_reg}(\ref{itm:var_asn_u_bounds}). \qed

\begin{lem}[Central limit theorem for $\xi_{i,t}(A,h) (w'u_t) $] \label{thm:var_clt}
Let \cref{asn:u_mds,asn:var_u_reg} hold. Let $i=1,\dots,n$. Let $\lbrace A_T \rbrace$ be a sequence of autoregressive coefficients in the parameter space $\mathcal{A}(0,\epsilon,C)$, and let $\lbrace h_T \rbrace$ be a sequence of nonnegative integers satisfying $T-h_T \to \infty$ and $g(\rho_i(A),h_T)^2/(T-h_{T}) \to 0$. Then
\[\frac{\sum_{t=1}^{T-h_T} \xi_{i,t}(A_T,h_T)(w'u_t)}{(T-h_T)^{1/2}v(A_T,h_T,w)} \underset{P_{A_T}}{\overset{d}{\to}} N(0,1),\]
for any $w \in \mathbb{R}^n \backslash \{0\}$.
\end{lem}

\begin{proof}
The definition of the multi-step forecast error implies 
\begin{equation}\label{eqn:var_sumauxiliary}
\sum_{t=1}^{T-h_T} \xi_{i,t}(A_{T},h_{T}) (w'u_t)  = \sum_{t=1}^{T-h_T}  \left( \beta_{i}(A_{T}, h_{T}-1)' u_{t+1} +  \ldots +\beta_{i}(A_{T}, 0)'  u_{t+h_{T}} \right) (w' u_t).
\end{equation}
The summands above do not form a martingale difference sequence with respect to a conventionally defined filtration of the form $\sigma(u_{t+h_T},u_{t+h_T-1},u_{t+h_T-2},\dots)$, even if $\lbrace u_t \rbrace$ is i.i.d. Instead, we will define a process that ``reverses time''. For any $T$ and any time period $1 \leq t \leq T-{h_T}$, define the triangular array and filtration
\begin{align*}
 \chi_{T,t} & = \frac{\xi_{i,T-h_T+1-t}(A_T,h_T) ( w' u_{T-h_{T}+1-t})} {  (T-h)^{1/2} v(A_{T}, h_{T},w),  },  \\
 \mathcal{F}_{T,t} & = \sigma( u_{T-h_{T}+1-t}, u_{T-h_{T}+2-t}, \ldots ).
\end{align*}
We say that we have reversed time because $\chi_{T,1}$ corresponds to the (scaled) last term that appears in the summation \eqref{eqn:var_sumauxiliary}; the term $\chi_{T,2}$ to the second-to-last term, and so on. By reversing time we have achieved three things. First, the sequence of $\sigma$-algebras is a \emph{filtration}:
\[ \mathcal{F}_{T,1} \subseteq \mathcal{F}_{T,2} \subseteq \ldots \subseteq \mathcal{F}_{T,T-h_{T}}. \]
Second, the process $\lbrace \chi_{T,t}\rbrace $ is adapted to the filtration $\lbrace \mathcal{F}_{T,t}\rbrace $, as $\chi_{T,t}$ is measurable with respect to $\mathcal{F}_{T,t}$ for all $t$. Third, the pair $\lbrace \chi_{T,t},\mathcal{F}_{T,t} \rbrace$ form a martingale difference array:
\begin{align*}
E[\chi_{T,t} \mid \mathcal{F}_{T,t-1}] & \propto  E[( \beta_{i}(A_{T},h_{T}-1)' u_{T-h_{T}+2-t}   \ldots + \beta_{i}(A_{T},0)'u_{T+1-t}) (w' u_{T-h_{T}+1-t})  \\
& \qquad \mid u_{T-h_{T}+2-t}, u_{T-h_{T}+3-t}, \dots   ] \\
&=( \beta_{i}(A_{T},h_{T}-1)' u_{T-h_{T}+2-t}   \ldots + \beta_{i}(A_{T},0)'u_{T+1-t}) \\
& \quad \times E[(w' u_{T-h_{T}+1-t}) \mid u_{T-h_{T}+2-t}, u_{T-h_{T}+3-t}, \dots   ] \\
&= 0,
\end{align*}
where the last equality follows from \cref{asn:u_mds}. 

Thus, we can apply the martingale central limit theorem in \citet[Thm. 24.3]{Davidson1994} to show that
\[\sum_{t=1}^{T-h_T} \chi_{T,t} \stackrel{d}{\to} N(0,1),\]
which is the statement of the lemma. We now verify the conditions of this theorem. First, by definition of $v(A,h,w)$,
\[\sum_{t=1}^{T-h_T} E[\chi_{T,t}^2] = 1.\]
Second, in \cref{thm:var_se_infeas} below we show (by means of Chebyshev's inequality)
\[\sum_{t=1}^{T-h_T} \chi_{T,t}^2 = \frac{\sum_{t=1}^{T-h_T}\xi_{i,t}(A_T,h_T)^2 (w' u_t)^2}{(T-h_T) v(A_{T},h_T,w)^2} \stackrel{p}{\to} 1.\]
Finally, we argue that $\max_{1 \leq t \leq T-h_T} |\chi_{T,t}(A_T,h_T)| \stackrel{p}{\to} 0$. By \citet[Thm. 23.16]{Davidson1994}, it is sufficient to prove that, for arbitrary $c>0$, we have
\[(T-h_T)E\left[\chi_{T,t}^2 \mathbbm{1}(|\chi_{T,t}|>c)\right] \to 0.\] Indeed,
\begin{align*}
&(T-h_T)E\left[\chi_{T,t}^2 \mathbbm{1}(|\chi_{T,t}|>c)\right] \\
& \leq (T-h_T)E\left[\chi_{T,t}^2  \mathbbm{1}(|\chi_{T,t}|>c) \times \frac{\chi_{T,t}^2}{c^2}\right] \\
& \leq (T-h_T)\frac{E[\chi_{T,t}^4]}{c^2} \\ 
&= \frac{1}{(T-h_T) c^2}E\left[\left|v(A_T,h_T,w)^{-1}\xi_{i, T-h_{T}+1-t}(A_T,h_T)(w' u_{T-h_{T}+1-t})\right|^4\right] \\
&\leq \frac{6E(\|u_t\| ^8)}{(T-h_T) \times  \delta^2 \times  \lambda_{\min}(\Sigma)^2 \times  c^2},
\end{align*}
where the last inequality uses \cref{thm:var_res_4th_bound} below (recall that $\delta$ is the constant in \cref{asn:var_u_reg}(\ref{itm:var_asn_u_bounds})). The right-hand side tends to zero as $T-h_{T} \to \infty$, as required. 
\end{proof}

\begin{lem}[Consistency of standard errors.] \label{thm:var_se}
Let \cref{asn:u_mds,asn:var_u_reg,asn:var_XpX} hold. Let the sequence $\lbrace A_T \rbrace$ of elements in $\mathcal{A}(0,C,\epsilon)$ and the sequence $\lbrace h_T \rbrace$ of non-negative integers satisfy $T-h_T \to \infty$ and $g(\max_i\lbrace |\rho_i(A_T)| \rbrace,h_T)^2/(T-h_T) \to \infty$. Define $\tilde{\nu} \equiv \Sigma^{-1}\nu$. Then
\[\frac{(T-h_T)^{1/2}\hat{s}(h_T,\nu)}{v(A_T,h_T,\tilde{\nu})} \underset{P_{A_T}}{\overset{p}{\to}} 1.\]
\end{lem}
\begin{proof}
See \cref{sec:var_se_proof}.
\end{proof}

\begin{lem}[Convergence rates of estimators] \label{thm:var_estim_conv}
Let the conditions of \cref{thm:var_se} hold. Let $w \in \mathbb{R}^n \backslash \lbrace 0 \rbrace$. Then the following statements all hold:
\begin{enumerate}[i)]
\item $\frac{\|\hat{\beta}_1(h_T)-\beta_1(A_T,h_T)\|}{v(A_T,h_T,w)} \underset{P_{A_T}}{\overset{p}{\to}} 0$. \label{itm:var_estim_conv_i}
\item $\frac{\left\|G(A_T,T-h_T,\epsilon)[\hat{\eta}_1(A_T,h_T)-\eta_1(A_T,h_T)]\right\|}{v(A_T,h_T,w)} \underset{P_{A_T}}{\overset{p}{\to}} 0$. \label{itm:var_estim_conv_ii}
\item $(T-h_T)^{1/2}\|(\hat{A}(h_T)-A_T)G(A_T,T-h_T,\epsilon)\| = O_{P_{A_T}}(1)$. \label{itm:var_estim_conv_iii}
\end{enumerate}
\end{lem}
\begin{proof}
See \cref{sec:var_estim_conv_proof}.
\end{proof}

\begin{lem}[OLS numerator] \label{thm:var_estim_conv_numer}
Let \cref{asn:u_mds,asn:var_u_reg} hold. Let $\lbrace A_T \rbrace$ be a sequence of autoregressive coefficients in $A_{T} \in \mathcal{A}(0,\epsilon,C)$, and let $\lbrace h_T \rbrace$ be a sequence of nonnegative integers satisfying $T-h_T \to \infty$ and $g(\max_i \lbrace |\rho_i(A)| \rbrace,h_T)^2/T \to 0$. Then, for any $w \in \mathbb{R}^n \backslash \lbrace 0 \rbrace$, $i,j \in \lbrace 1,\ldots, n\rbrace $, and $r \in \lbrace 1,\ldots, p\rbrace $,
\[\frac{\sum_{t=1}^{T-h_T} \xi_{i,t}(A_T,h_T)y_{j,t-r}}{(T-h_T)v(A_T,h_T,w)g(\rho^*_j(A_{T},\epsilon),T-h_T)} \underset{P_{A_T}}{\overset{p}{\to}} 0. \]
\end{lem}
\begin{proof}
See \cref{sec:var_estim_conv_numer_proof}.
\end{proof}

\begin{lem}[Consistency of $\hat{\Sigma}(h)$.] \label{thm:var_Sigmahat_conv}
Let \cref{asn:u_mds,asn:var_u_reg,asn:var_XpX} hold. Let the sequence $\lbrace h_T \rbrace$ of non-negative integers satisfy $T-h_T \to \infty$. Then both the following statements hold:
\begin{enumerate}[i)]
\item $\frac{1}{T-h_T}\sum_{t=1}^{T-h_T}u_tu_t' \stackrel{p}{\to} \Sigma$. \label{itm:var_Sigmahat_conv_i}
\item Assume the sequence $\lbrace A_T \rbrace$ in $\mathcal{A}(0,C,\epsilon)$ and $\lbrace h_T \rbrace$ satisfy $g(\max_i\lbrace |\rho_i(A_T)|\rbrace,h_T)^2/(T-h_T) \to \infty$. Then $\hat{\Sigma}(h_T) - \frac{1}{T-h_T}\sum_{t=1}^{T-h_T}u_tu_t' \underset{P_{A_T}}{\overset{p}{\to}} 0$. \label{itm:var_Sigmahat_conv_ii}
\end{enumerate}
\end{lem}
\begin{proof}
See \cref{sec:var_Sigmahat_conv_proof}.
\end{proof}

\begin{lem}[Consistency of the sample variance of $\xi_{i,t}(A_{T},h)(w'u_t)$] \label{thm:var_se_infeas}
Let the conditions of \cref{thm:var_clt} hold. Then
\[\frac{\sum_{t=1}^{T-h_T}\xi_{i,t}(A_T,h_T)^2 (w'u_t)^2}{(T-h_T) v(A_{T},h_T,w)^2} \underset{P_{A_T}}{\overset{p}{\to}} 1.\]
\end{lem}
\begin{proof}
See \cref{sec:var_se_infeas_proof}.
\end{proof}

\begin{lem}[Bounds on the fourth moments of $\xi_{i,t}(A,h)(w' u_t)$ and $\xi_{i,t}(A,h)$] \label{thm:var_res_4th_bound}
Let \cref{asn:u_mds} and \cref{asn:var_u_reg}(\ref{itm:var_asn_u_bounds}) hold. Then
\[E\left[\left(v(A,h,a)^{-1}\xi_{i,t}(A,h)(w' u_t ) \right)^4 \right] \leq \frac{6 E(\|u_t\|^8) }{\delta^2 \lambda_{\min}(\Sigma)^2}  \]
and
\[E\left[\left(v(A,h,w)^{-1}\xi_{i,t}(A,h)\right)^4 \right] \leq \frac{6 E(\|u_t\|^4) }{\delta^2 \lambda_{\min}(\Sigma)^2 \|w\|^4} \]
for all $h \in \mathbb{N}$, $A \in \mathcal{A}(0,\epsilon,C)$, and $w \in \mathbb{R}^n \backslash \lbrace 0 \rbrace$.
\end{lem}
\begin{proof}
See \cref{sec:var_res_4th_bound_proof}.
\end{proof}

\section{Comparison of Inference Procedures}
\label{sec:appendix}

\subsection{AR(1) Simulation Study: ARCH Innovations}
\label{sec:sim_ar1_arch}
Consider the AR(1) model \eqref{eqn:ar1} with innovations $u_{t}$ that follow an ARCH(1) process
\begin{equation} \label{eqn:ARCH}
u_t = \tau_t \varepsilon_t, \quad \tau^2_t  = \alpha_0 + \alpha_1 u^2_{t-1}, \quad \varepsilon_{t} \stackrel{i.i.d.}{\sim} N(0,1).
\end{equation} 
These innovations satisfy \cref{asn:u_mds}. In our simulations, we set $\alpha_1=.7$ and $\alpha_0=(1-\alpha_1)$.\footnote{This value of $\alpha_0$ ensures $\mathbb{E}[\tau_t^2]=1$.} \cref{tab:TableMC-arch} presents the results, which are qualitatively similar to the i.i.d. case discussed in \cref{sec:sim_ar1}.

\afterpage{
\begin{landscape}
\begin{table}[p]
    \centering
    \caption{Monte Carlo results: ARCH innovations}
    \vspace{0.5\baselineskip}
    \begin{tabular}{r|cccccc|cccccc}
& \multicolumn{6}{c|}{Coverage} & \multicolumn{6}{c}{Median length} \\
$h$ & $\text{LP-LA}_b$ & $\text{LP-LA}$ & $\text{LP}_b$ & $\text{LP}$ & $\text{AR-LA}_b$ & $\text{AR}$ & $\text{LP-LA}_b$ & $\text{LP-LA}$ & $\text{LP}_b$ & $\text{LP}$ & $\text{AR-LA}_b$ & $\text{AR}$ \\
\hline
\multicolumn{13}{c}{$\rho = 0.00$} \\
  1 & 0.892 & 0.861 & 0.916 & 0.804 & 0.831 & 0.868 & 0.386 & 0.356 & 0.406 & 0.316 & 0.337 & 0.360 \\
  6 & 0.913 & 0.903 & 0.910 & 0.865 & 0.000 & 1.000 & 0.211 & 0.207 & 0.218 & 0.195 & 0.000 & 0.000 \\
 12 & 0.901 & 0.895 & 0.896 & 0.874 & 0.000 & 1.000 & 0.209 & 0.205 & 0.211 & 0.204 & 0.000 & 0.000 \\
 36 & 0.903 & 0.894 & 0.899 & 0.890 & 0.000 & 1.000 & 0.222 & 0.217 & 0.221 & 0.215 & 0.000 & 0.000 \\
 60 & 0.899 & 0.889 & 0.904 & 0.887 & 0.000 & 0.991 & 0.236 & 0.229 & 0.237 & 0.231 & 0.000 & 0.000 \\
\multicolumn{13}{c}{$\rho = 0.50$} \\
  1 & 0.891 & 0.865 & 0.908 & 0.806 & 0.836 & 0.874 & 0.387 & 0.357 & 0.330 & 0.257 & 0.336 & 0.294 \\
  6 & 0.900 & 0.892 & 0.908 & 0.843 & 0.837 & 0.776 & 0.246 & 0.238 & 0.272 & 0.232 & 0.090 & 0.048 \\
 12 & 0.904 & 0.895 & 0.896 & 0.879 & 0.837 & 0.689 & 0.240 & 0.233 & 0.265 & 0.250 & 0.008 & 0.001 \\
 36 & 0.897 & 0.887 & 0.894 & 0.869 & 0.837 & 0.579 & 0.254 & 0.246 & 0.277 & 0.265 & 0.000 & 0.000 \\
 60 & 0.901 & 0.885 & 0.902 & 0.879 & 0.837 & 0.540 & 0.273 & 0.262 & 0.300 & 0.283 & 0.000 & 0.000 \\
\multicolumn{13}{c}{$\rho = 0.95$} \\
  1 & 0.897 & 0.859 & 0.823 & 0.824 & 0.838 & 0.856 & 0.392 & 0.359 & 0.084 & 0.079 & 0.335 & 0.086 \\
  6 & 0.896 & 0.819 & 0.854 & 0.788 & 0.838 & 0.806 & 0.621 & 0.519 & 0.381 & 0.327 & 1.746 & 0.355 \\
 12 & 0.880 & 0.785 & 0.850 & 0.747 & 0.838 & 0.758 & 0.724 & 0.560 & 0.604 & 0.489 & 3.942 & 0.467 \\
 36 & 0.869 & 0.788 & 0.859 & 0.667 & 0.838 & 0.643 & 0.717 & 0.596 & 0.816 & 0.596 & 64.319 & 0.291 \\
 60 & 0.881 & 0.825 & 0.885 & 0.692 & 0.838 & 0.579 & 0.711 & 0.615 & 0.900 & 0.625 & 1032.604 & 0.095 \\
\multicolumn{13}{c}{$\rho = 1.00$} \\
  1 & 0.896 & 0.860 & 0.841 & 0.579 & 0.839 & 0.560 & 0.386 & 0.356 & 0.040 & 0.041 & 0.330 & 0.041 \\
  6 & 0.879 & 0.759 & 0.859 & 0.543 & 0.839 & 0.513 & 0.686 & 0.585 & 0.240 & 0.228 & 2.035 & 0.223 \\
 12 & 0.854 & 0.662 & 0.845 & 0.454 & 0.839 & 0.468 & 0.902 & 0.715 & 0.473 & 0.396 & 5.510 & 0.391 \\
 36 & 0.731 & 0.424 & 0.752 & 0.213 & 0.839 & 0.352 & 1.384 & 0.935 & 1.170 & 0.609 & 177.260 & 0.669 \\
 60 & 0.640 & 0.279 & 0.697 & 0.164 & 0.839 & 0.294 & 1.475 & 0.964 & 1.647 & 0.642 & 5593.663 & 0.729 \\
\end{tabular}
    \label{tab:TableMC-arch}
    \\
    \vspace{0.5\baselineskip}
\begin{minipage}{1.1 \textwidth} 
{\footnotesize Coverage probability and median length of nominal 90\% confidence intervals at different horizons. AR(1) model with $\rho \in \lbrace 0,.5,.95,1\rbrace $, $T=240$, innovations as in equation \eqref{eqn:ARCH}. 5,000 Monte Carlo repetitions; 2,000 bootstrap iterations.}
\end{minipage}
\end{table}
\end{landscape}
}

\subsection{Analytical Results}
\label{sec:comp_details}
In this subsection we provide details on the relative efficiency of lag-augmented LP versus other procedures. Throughout, we focus on the tractable AR(1) model in \cref{sec:ar1}.

\subsubsection{Relative Efficiency of Lag-Augmented LP}
\label{sec:comp_details_releff}
Here we compare the efficiency of lag-augmented LP relative to (i) non-augmented AR, (ii) lag-augmented AR, and (iii) non-augmented LP. We restrict attention to a stationary, homoskedastic AR(1) model and to a fixed impulse response horizon $h$.

Specifically, we here assume the AR(1) model \eqref{eqn:ar1} with $\rho \in (-1,1)$ and where the innovations $u_t$ are assumed to be i.i.d. with variance $\sigma^2$. This provides useful intuition, even though the main purpose of this paper is to develop methods that work in empirically realistic settings with several variables/lags, high persistence, and longer horizons.

\paragraph{Comparison with non-augmented AR.}
In a stationary and homoskedastic AR(1) model, the non-augmented AR estimator is the asymptotically efficient estimator among all regular estimators that are consistent also under heteroskedasticity. This follows from standard semiparametric efficiency arguments, since the non-augmented AR estimator simply plugs the semiparametrically efficient OLS estimator of $\rho$ into the smooth impulse response transformation $\rho^h$. In particular, non-augmented AR is weakly more efficient than (i) lag-augmented AR, (ii) non-augmented LP, and (iii) lag-augmented LP. As we have discussed in \cref{sec:comp}, however, standard non-augmented AR inference methods perform poorly in situations outside of the benign stationary, short-horizon case.

To gain intuition about the efficiency loss associated with lag augmentation, consider the first horizon $h=1$. At this horizon, the lag-augmented LP and lag-augmented AR estimators coincide. These estimators regress $y_{t+1}$ on $y_t$, while controlling for $y_{t-1}$. As discussed in \cref{sec:ar1_intuition}, this is the same as regressing $y_{t+1}$ directly on the innovation $u_t$, while controlling for $y_{t-1}$ (which is uncorrelated with $u_t$). In contrast, the non-augmented AR estimator just regresses $y_{t+1}$ on $y_t$ without controls. Note that (i) the regressor $y_t$ has a higher variance than the regressor $u_t$, and (ii) the residual in both the augmented and non-augmented regressions equals $u_{t+1}$. Thus, the usual homoskedastic OLS asymptotic variance formula implies that the non-augmented AR estimator is more efficient than the lag-augmented AR/LP estimator.

\paragraph{Comparison with lag-augmented AR.}
The relative efficiency of the lag-augmented AR and lag-augmented LP impulse response estimators is ambiguous. In the homoskedastic AR(1) model, the proof of \cref{thm:var_lp_inference} implies that the asymptotic variance of the lag-augmented LP estimator $\hat{\beta}(h)$ is
\begin{equation} \label{eqn:asyvar_lalp}
\text{AsyVar}_\rho(\hat{\beta}(h)) = \frac{E[u_t^2\xi_t(\rho,h)^2]}{[E(u_t^2)]^2} = \frac{\sigma^2 E[\xi_t(\rho,h)^2]}{\sigma^4} = \frac{\sigma^2 \sum_{\ell=0}^{h-1}\rho^{2\ell}\sigma^2}{\sigma^4} = \sum_{\ell=0}^{h-1}\rho^{2\ell}.
\end{equation}
We want to compare this to the asymptotic variance of the plug-in AR estimator $\hat{\beta}_\text{ARLA}(h) \equiv \hat{\rho}_\text{LA}^h$, where $\hat{\rho}_\text{LA}$ is the coefficient estimate on the first lag in a regression with \emph{two} lags \citep{Inoue2020}. Note that $\hat{\rho}_\text{LA} = \hat{\beta}(1)$ by definition. By the delta method, the asymptotic variance of $\hat{\beta}_\text{ARLA}(h)$ is given by
\[\text{AsyVar}_\rho(\hat{\beta}_\text{ARLA}(h)) = (h\rho^{h-1})^2 \times \text{AsyVar}_\rho(\hat{\rho}_\text{LA}) = (h\rho^{h-1})^2 \times \text{AsyVar}_\rho(\hat{\beta}(1)) = (h\rho^{h-1})^2.\]
To rank the LP and ARLA estimators in terms of asymptotic variance, note that
\[\text{AsyVar}_\rho(\hat{\beta}(h)) \leq \text{AsyVar}_\rho(\hat{\beta}_\text{ARLA}(h)) \Longleftrightarrow \sum_{\ell=0}^{h-1} \rho^{2(\ell-h+1)} \leq h^2 \Longleftrightarrow \sum_{m=0}^{h-1} \rho^{-2m} \leq h^2.\]
Consider the inequality on the far right of the above display. For $h \geq 2$, the left-hand side is monotonically decreasing from $\infty$ to $h$ as $|\rho|$ goes from 0 to 1. Hence, there exists an indifference function $\underline{\rho} \colon \mathbb{N} \to (0,1)$ such that
\[\text{AsyVar}_\rho(\hat{\beta}(h)) \leq \text{AsyVar}_\rho(\hat{\beta}_\text{ARLA}(h)) \Longleftrightarrow |\rho| \geq \underline{\rho}(h).\]

\noindent \cref{fig:se_indiff} in \cref{sec:comp} plots the indifference curve between lag-augmented LP standard errors and lag-augmented AR standard errors (lower thick line).


\paragraph{Comparison with non-augmented LP.}
The non-augmented LP estimator $\hat{\beta}_\text{LPNA}(h)$ is obtained from a regression of $y_{t+h}$ on $y_t$ without controls. As is clear from the representation \eqref{eqn:ytph_decomp}, the asymptotic variance of this estimator is given by
\begin{align*}
\text{AsyVar}_\rho(\hat{\beta}_\text{LPNA}(h)) &= \frac{\sum_{\ell=-\infty}^\infty E[y_t\xi_t(\rho,h)y_{t-\ell}\xi_{t-\ell}(\rho,h)]}{[E(y_t^2)]^2} \\
&= \frac{\sum_{\ell=-h+1}^{h-1} \rho^{|\ell|} E[y_{t-|\ell|}^2]E[\xi_t(\rho,h)\xi_{t-|\ell|}(\rho,h)]}{[E(y_t^2)]^2} \\
&= \frac{\sum_{\ell=-h+1}^{h-1} \sum_{m=|\ell|}^{h-1}\rho^{2m}}{E(y_t^2)/\sigma^2} = (1-\rho^2)\sum_{\ell=-h+1}^{h-1} \sum_{m=|\ell|}^{h-1}\rho^{2m} \\
&= \sum_{\ell=-h+1}^{h-1} (\rho^{2|\ell|}-\rho^{2h}) = \sum_{\ell=0}^{h-1} \rho^{2\ell} + \sum_{\ell=1}^{h-1} \rho^{2\ell} - (2h-1)\rho^{2h}.
\end{align*}
Thus, using \eqref{eqn:asyvar_lalp}, we find that
\[\text{AsyVar}_\rho(\hat{\beta}(h)) \leq \text{AsyVar}_\rho(\hat{\beta}_\text{LPNA}(h)) \Longleftrightarrow \sum_{\ell=1}^{h-1} \rho^{2\ell} \geq (2h-1)\rho^{2h} \Longleftrightarrow \sum_{\ell=1}^{h-1} \rho^{-2\ell} \geq (2h-1).\]
The last equivalence assumes $\rho \neq 0$, since lag-augmented and non-augmented LP are clearly equally efficient when $\rho=0$. For $h=1$, the last inequality above is never satisfied. This is because at this horizon lag-augmented and non-augmented LP reduce to lag-augmented and non-augmented AR, respectively, and the latter is more efficient, as discussed previously. For $h\geq 2$, the left-hand side of the last inequality above decreases monotonically from $\infty$ to $h-1$ as $|\rho|$ goes from 0 to 1. Thus, there exists an indifference function $\overline{\rho} \colon \mathbb{N} \to (0,1)$ such that
\[\text{AsyVar}_\rho(\hat{\beta}(h)) \leq \text{AsyVar}_\rho(\hat{\beta}_\text{LPNA}(h)) \Longleftrightarrow |\rho| \leq \overline{\rho}(h).\]

\noindent \cref{fig:se_indiff} in \cref{sec:comp} plots the indifference curve between lag-augmented LP and non-augmented LP (upper thick line).


\subsubsection{Length of Lag-Augmented AR Bootstrap Confidence Interval}
\label{sec:comp_details_laarboot}
Here we prove that the lag-augmented AR bootstrap confidence interval of \citet{Inoue2020} is very wide asymptotically when the data is persistent and the horizon is moderately long.

Let $Y^{T}\equiv (y_{1}, \ldots, y_{T})$ denote a sample of size $T$ generated by the AR(1) model \eqref{eqn:ar1}. Let $P_{\rho}$ denote the distribution of the data when the autoregressive parameter equals  $\rho$. Let $\hat{\rho}$ denote the lag-augmented autoregressive estimator of the parameter $\rho$ based on the data $Y^T$ (i.e., the first coefficient in an AR(2) regression). Let $\hat{\rho}^*$ be the corresponding  lag-augmented autoregressive estimator based on a bootstrap sample. We use $\mathbb{P}^*(\cdot  \mid Y^{T})$ to denote the distribution of the bootstrap samples conditional on the data.

By the results in \cite{Inoue2020} we will assume that (i) $\sqrt{T}(\hat{\rho}-\rho)$ converges uniformly to $\mathcal{N}(0, \omega^2)$ for some $\omega>0$, and (ii) the law of  $\sqrt{T}(\hat{\rho}^*-\hat{\rho}) \mid Y^{T} $ also converges to $\mathcal{N}(0, \omega^2)$ (in probability). 

We consider a sequence of autoregressive parameters $\lbrace \rho_{T}\rbrace $ approaching unity as $T \to \infty$, and a sequence of horizons $\lbrace h_{T} \rbrace$ that increases with the sample size. The restrictions on these sequences are as follows:
\begin{equation} \label{eqn:rhoT_boot}
    h_{T}(1-\rho_{T}) \rightarrow a \in [0,\infty),
\end{equation}
\begin{equation} \label{eqn:ht_boot}
    h_{T} \propto T^{\eta}, \quad \eta \in [1/2,1]. 
\end{equation}
For example, these assumptions cover the cases of (i) local-to-unity DGPs $\rho_T = 1-a/T$, $a \geq 0$, at long horizons $h_T \propto T$, and (ii) not-particularly-local-to-unity DGPs $\rho_T = 1-a/\sqrt{T}$, $a>0$, at medium-long horizons $h_T \propto \sqrt{T}$.

We now derive an expression for the quantiles of the bootstrap distribution of the impulse response estimates. For any $c \in \mathbb{R}$,
\begin{eqnarray*}
\mathbb{P}^*( (\hat{\rho}^*)^{h_{T}} \leq c \mid Y^{T}) &=& \mathbb{P}^*( (\hat{\rho}^*)^{h_{T}} \leq c \textrm{ and } \hat{\rho}^* \geq 0  \mid Y_{T} ) + o_{P_{\rho_T}}(1),\\
&=& \mathbb{P}^*( \sqrt{T}(\hat{\rho}^*-\hat{\rho}) \leq \sqrt{T}(c^{1/h_{T}}-\hat{\rho}) \mid Y^{T}) + o_{P_{\rho_T}}(1). 
\end{eqnarray*}
The equation above implies that the $\alpha$ bootstrap quantile of $(\hat{\rho}^*)^{h_{T}}$ is given by $c^*_{\alpha} = \hat{c}_{\alpha} + o_{P_{\rho_T}}$, where
\begin{equation}
    \hat{c}_{\alpha} = \left( \hat{\rho} + \omega z_{\alpha}/\sqrt{T} \right)^{h_{T}},
\end{equation}
and $z_{\alpha}$ is the $\alpha$ quantile of the standard normal distribution. Note that
\begin{eqnarray*}
\log \hat{c}_{\alpha} &=& h_T\left[\log\hat{\rho} + \frac{\omega z_{\alpha}}{\sqrt{T}\hat{\rho}} + o_{P_{\rho_T}}(T^{-1/2}) \right] \\
&& \text{(since $\log(x+y) = \log(x) + y/x + o(y/x)$)} \\
&=& h_T\log\rho_T + \frac{\omega}{\rho_T}\frac{h_T}{\sqrt{T}}\left[\frac{\rho_T}{\omega}\sqrt{T}\log\frac{\hat{\rho}}{\rho_T} + z_{\alpha} + o_{P_{\rho_T}}(1)\right].
\end{eqnarray*}
By \eqref{eqn:rhoT_boot}, we have $\rho_T \to 1$ and $h_T \log\rho_T \to -a$. Also, the delta method implies
\[\frac{\rho_T}{\omega}\sqrt{T}\log\frac{\hat{\rho}}{\rho_T} \stackrel{d}{\to} Z \equiv N(0,1).\]
Since $\sqrt{T}/h_T = O(1)$ by \eqref{eqn:ht_boot}, we then conclude that
\begin{equation} \label{eqn:ik_quant}
\frac{\sqrt{T}}{h_T}(\log \hat{c}_{\alpha} + a) \stackrel{d}{\to} \omega (Z+z_{\alpha}).
\end{equation}
This convergence in distribution is joint if we consider several quantiles $\alpha$ simultaneously.

The \citet{Inoue2020} lag-augmented AR Efron bootstrap confidence interval is given by $[c^*_{\alpha/2},c^*_{1-\alpha/2}]$, so its length equals $\hat{c}_{1-\alpha/2}-\hat{c}_{\alpha/2} + o_{P_{\rho_T}}(1)$. We now argue that this length does not shrink to zero asymptotically in two separate cases.

\paragraph{Case 1:  $h_{T} = \kappa \sqrt{T}, \kappa \in (0,1]$.}
In this case the result \eqref{eqn:ik_quant} immediately implies that the length of the \citet{Inoue2020} bootstrap confidence interval converges to a non-degenerate random variable asymptotically (though the confidence interval has correct asymptotic coverage). This contrasts with the lag-augmented LP confidence interval, whose length shrinks to zero in probability asymptotically.

\paragraph{Case 2: $h_{T} \propto T^{\eta}, \eta \in (1/2,1]$.}
In this case $h_T/\sqrt{T} \to \infty$. The result \eqref{eqn:ik_quant} then implies that, for any $\zeta>0$,
\begin{eqnarray*}
P_{\rho_T}\big([\zeta,1/\zeta] \subset [c^*_{\alpha/2},c^*_{1-\alpha/2}]\big) &=& P_{\rho_T}\big(\log\hat{c}_{\alpha/2} \leq \log\zeta \;\; \text{and} \;\; \log\hat{c}_{1-\alpha/2} \geq \log(1/\zeta) \big) + o(1) \\
&=& P\big(Z+z_{\alpha/2} < 0 \;\; \text{and} \;\; Z+z_{1-\alpha/2} > 0 \big) + o(1) \\
&=& 1-\alpha + o(1).
\end{eqnarray*}
This means that, though the Efron bootstrap confidence interval of \citet{Inoue2020} has correct coverage, it achieves this at the expense of reporting---with probability $(1-\alpha)$--- intervals that asymptotically contain any compact subset of the positive real line $(0,\infty)$. A similar argument shows that if we intersect the \citet{Inoue2020} confidence interval with the parameter space $[-1,1]$ for the impulse response, the confidence interval almost equals $[0,1]$ with probability $1-\alpha$. In contrast, as long as $\eta<1$, the lag-augmented LP confidence interval has valid coverage and length that tends to zero in probability.

\section{Verification of \texorpdfstring{\cref{asn:var_XpX}}{Assumption \ref{asn:var_XpX}}: AR(1) Case} \label{sec:var_XpX_ar1}
In the notation of \cref{sec:ar1}, and setting $\epsilon=0$ without loss of generality, it suffices to show: Any sequence $\lbrace \rho_T \rbrace \in [-1,1]$ has a subsequence (which we will also denote by $\lbrace \rho_T \rbrace$ for simplicity) such that the random variable $\max\lbrace T(1-|\rho_T|), 1\rbrace \frac{1}{T^2}\sum_{t=1}^T y_{t-1}^2$ converges in distribution along $\lbrace P_{\rho_T} \rbrace$ to a random variable that is strictly positive almost surely. By passing to a further subsequence if necessary, we may assume that $\lim_{T\to\infty} T(1-|\rho_T|)$ exists.

\paragraph{Case 1: $T(1-|\rho_T|) \to \infty$.}
We will argue that $\frac{1-|\rho_T|}{T}\sum_{t=1}^T y_{t-1}^2$ converges in probability to a nonzero constant along some subsequence. This follows from three facts. First, $\rho_T \to \tilde{c} \in [-1,1]$, at least along some subsequence. Second, direct calculation using \cref{asn:u_mds} shows that $E[\frac{1-\rho_T^2}{T}\sum_{t=1}^T y_{t-1}^2] \to \sigma^2=E(u_t^2)>0$. Third, tedious calculations similar to the proof of \cref{thm:var_estim_conv_numer} show that $\var[\frac{1-\rho_T^2}{T}\sum_{t=1}^T y_{t-1}^2] \to 0$.

\paragraph{Case 2: $T(1-|\rho_T|) \to c \in [0,\infty)$.}
By passing to a further subsequence, we may assume $\rho_T \to 1$ (the case $\rho_T \to -1$ can be handled similarly). We impose the additional assumption that, for ``local-to-unity'' sequences $\lbrace \rho_T \rbrace$ satisfying $T(1-\rho_T) \to c \in [0,\infty)$, the sequence of probability measures $\lbrace P_{\rho_T} \rbrace$ is contiguous to the measure $P_1$ (i.e., with $\rho=1$). This is known to hold for i.i.d. innovations $\lbrace u_t \rbrace$ whose density satisfy a smoothness condition \citep{Jansson2008}, and it also allows for certain types of conditional heteroskedasticity \citep[Section 4]{Jeganathan1995}. Under this extra assumption, we now just need to argue that, when $\rho=1$ is \emph{fixed}, $\frac{1}{T^2}\sum_{t=1}^T y_{t-1}^2$ converges in distribution to a continuously distributed random variable concentrated on $(0,\infty)$. But this is a well-known result from the unit root literature \citep[e.g.,][Chapter 17.4]{Hamilton1994}, since $\lbrace u_t \rbrace$ satisfies a Functional Central Limit Theorem under \cref{asn:u_mds,asn:var_u_reg} \citep[Theorem 27.14]{Davidson1994}.

\phantomsection
\addcontentsline{toc}{section}{References}
\bibliography{ref}

\begin{thebibliography}{48}
\newcommand{\enquote}[1]{``#1''}
\expandafter\ifx\csname natexlab\endcsname\relax\def\natexlab#1{#1}\fi

\bibitem[\protect\citeauthoryear{Andrews, Cheng, and Guggenberger}{Andrews
  et~al.}{2019}]{Andrews2019}
\textsc{Andrews, D. W.~K., X.~Cheng, and P.~Guggenberger} (2019):
  \enquote{{Generic Results for Establishing the Asymptotic Size of Confidence
  Sets and Tests},} \emph{Journal of Econometrics}, forthcoming.

\bibitem[\protect\citeauthoryear{Angrist, Jord\`{a}, and Kuersteiner}{Angrist
  et~al.}{2018}]{Angrist2018}
\textsc{Angrist, J.~D., {\`{O}}.~Jord\`{a}, and G.~M. Kuersteiner} (2018):
  \enquote{{Semiparametric Estimates of Monetary Policy Effects: String Theory
  Revisited},} \emph{Journal of Business \& Economic Statistics}, 36, 371--387.

\bibitem[\protect\citeauthoryear{Armstrong and Kolesar}{Armstrong and
  Kolesar}{2018}]{Armstrong_Kolesar:2018}
\textsc{Armstrong, T.~B. and M.~Kolesar} (2018): \enquote{{Optimal Inference in
  a Class of Regression Models},} \emph{Econometrica}, 86, 655--683.

\bibitem[\protect\citeauthoryear{Benkwitz, Neumann, and L\"{u}tkepohl}{Benkwitz
  et~al.}{2000}]{Benkwitz2000}
\textsc{Benkwitz, A., M.~H. Neumann, and H.~L\"{u}tkepohl} (2000):
  \enquote{{Problems related to confidence intervals for impulse responses of
  autoregressive processes},} \emph{Econometric Reviews}, 19, 69--103.

\bibitem[\protect\citeauthoryear{Breitung and Br\"{u}ggemann}{Breitung and
  Br\"{u}ggemann}{2019}]{Breitung2019}
\textsc{Breitung, J. and R.~Br\"{u}ggemann} (2019): \enquote{{Projection
  estimators for structural impulse responses},} University of Konstanz
  Department of Economics Working Paper Series 2019-05.

\bibitem[\protect\citeauthoryear{Brillinger}{Brillinger}{2001}]{Brillinger2001}
\textsc{Brillinger, D.~R.} (2001): \emph{{Time Series: Data Analysis and
  Theory}}, Classics in Applied Mathematics, SIAM.

\bibitem[\protect\citeauthoryear{Brugnolini}{Brugnolini}{2018}]{Brugnolini2018}
\textsc{Brugnolini, L.} (2018): \enquote{{About Local Projection Impulse
  Response Function Reliability},} Manuscript, University of Rome ``Tor
  Vergata''.

\bibitem[\protect\citeauthoryear{Chevillon}{Chevillon}{2017}]{Chevillon2017}
\textsc{Chevillon, G.} (2017): \enquote{{Robustness of Multistep Forecasts and
  Predictive Regressions at Intermediate and Long Horizons},} ESSEC Working
  Paper 1710.

\bibitem[\protect\citeauthoryear{Davidson}{Davidson}{1994}]{Davidson1994}
\textsc{Davidson, J.} (1994): \emph{{Stochastic Limit Theory: An Introduction
  for Econometricians}}, Advanced Texts in Econometrics, Oxford University
  Press.

\bibitem[\protect\citeauthoryear{Dolado and L{\"u}tkepohl}{Dolado and
  L{\"u}tkepohl}{1996}]{Dolado1996}
\textsc{Dolado, J.~J. and H.~L{\"u}tkepohl} (1996): \enquote{{Making Wald tests
  work for cointegrated VAR systems},} \emph{Econometric Reviews}, 15,
  369--386.

\bibitem[\protect\citeauthoryear{Dufour, Pelletier, and \'{E}ric
  Renault}{Dufour et~al.}{2006}]{Dufour2006}
\textsc{Dufour, J.-M., D.~Pelletier, and \'{E}ric Renault} (2006):
  \enquote{{Short run and long run causality in time series: inference},}
  \emph{Journal of Econometrics}, 132, 337--362.

\bibitem[\protect\citeauthoryear{Gon\c{c}alves and Kilian}{Gon\c{c}alves and
  Kilian}{2004}]{Goncalves2004}
\textsc{Gon\c{c}alves, S. and L.~Kilian} (2004): \enquote{Bootstrapping
  autoregressions with conditional heteroskedasticity of unknown form,}
  \emph{Journal of Econometrics}, 123, 89--120.

\bibitem[\protect\citeauthoryear{Gon\c{c}alves and Kilian}{Gon\c{c}alves and
  Kilian}{2007}]{Goncalves2007}
---\hspace{-.1pt}---\hspace{-.1pt}--- (2007): \enquote{{Asymptotic and
  Bootstrap Inference for AR($\infty$) Processes with Conditional
  Heteroskedasticity},} \emph{Econometric Reviews}, 26, 609--641.

\bibitem[\protect\citeauthoryear{Gospodinov}{Gospodinov}{2004}]{Gospodinov2004}
\textsc{Gospodinov, N.} (2004): \enquote{{Asymptotic confidence intervals for
  impulse responses of near-integrated processes},} \emph{Econometrics
  Journal}, 7, 505--527.

\bibitem[\protect\citeauthoryear{Hamilton}{Hamilton}{1994}]{Hamilton1994}
\textsc{Hamilton, J.~D.} (1994): \emph{{Time Series Analysis}}, Princeton
  University Press.

\bibitem[\protect\citeauthoryear{Hansen}{Hansen}{1999}]{Hansen1999}
\textsc{Hansen, B.~E.} (1999): \enquote{The {Grid} {Bootstrap} and the
  {Autoregressive} {Model},} \emph{Review of Economics and Statistics}, 81,
  594--607.

\bibitem[\protect\citeauthoryear{Herbst and Johannsen}{Herbst and
  Johannsen}{2020}]{Herbst2020}
\textsc{Herbst, E.~P. and B.~K. Johannsen} (2020): \enquote{{Bias in Local
  Projections},} Board of Governors of the Federal Reserve System Finance and
  Economics Discussion Series 2020-010.

\bibitem[\protect\citeauthoryear{Hjalmarsson and Kiss}{Hjalmarsson and
  Kiss}{2020}]{Hjalmarsson2020}
\textsc{Hjalmarsson, E. and T.~Kiss} (2020): \enquote{{Long-Run Predictability
  Tests Are Even Worse Than You Thought},} Manuscript.

\bibitem[\protect\citeauthoryear{Inoue and Kilian}{Inoue and
  Kilian}{2002}]{Inoue2002}
\textsc{Inoue, A. and L.~Kilian} (2002): \enquote{{Bootstrapping Autoregressive
  Processes with Possible Unit Roots},} \emph{Econometrica}, 70, 377--391.

\bibitem[\protect\citeauthoryear{Inoue and Kilian}{Inoue and
  Kilian}{2016}]{IK2016}
---\hspace{-.1pt}---\hspace{-.1pt}--- (2016): \enquote{Joint confidence sets
  for structural impulse responses,} \emph{Journal of Econometrics}, 192,
  421--432.

\bibitem[\protect\citeauthoryear{Inoue and Kilian}{Inoue and
  Kilian}{2020}]{Inoue2020}
---\hspace{-.1pt}---\hspace{-.1pt}--- (2020): \enquote{{The uniform validity of
  impulse response inference in autoregressions},} \emph{Journal of
  Econometrics}, 215, 450--472.

\bibitem[\protect\citeauthoryear{Jansson}{Jansson}{2008}]{Jansson2008}
\textsc{Jansson, M.} (2008): \enquote{{Semiparametric Power Envelopes for Tests
  of the Unit Root Hypothesis},} \emph{Econometrica}, 76, 1103--1142.

\bibitem[\protect\citeauthoryear{Jeganathan}{Jeganathan}{1995}]{Jeganathan1995}
\textsc{Jeganathan, P.} (1995): \enquote{{Some Aspects of Asymptotic Theory
  with Applications to Time Series Models},} \emph{Econometric Theory}, 11,
  818--887.

\bibitem[\protect\citeauthoryear{Jord\`{a}}{Jord\`{a}}{2005}]{Jorda2005}
\textsc{Jord\`{a}, {\`{O}}.} (2005): \enquote{{Estimation and Inference of
  Impulse Responses by Local Projections},} \emph{American Economic Review},
  95, 161--182.

\bibitem[\protect\citeauthoryear{Kilian}{Kilian}{1998}]{Kilian1998}
\textsc{Kilian, L.} (1998): \enquote{Small-sample {Confidence} {Intervals} for
  {Impulse} {Response} {Functions},} \emph{Review of Economics and Statistics},
  80, 218--230.

\bibitem[\protect\citeauthoryear{Kilian and Kim}{Kilian and
  Kim}{2011}]{Kilian2011}
\textsc{Kilian, L. and Y.~J. Kim} (2011): \enquote{{How Reliable Are Local
  Projection Estimators of Impulse Responses?}} \emph{Review of Economics and
  Statistics}, 93, 1460--1466.

\bibitem[\protect\citeauthoryear{Kilian and L\"{u}tkepohl}{Kilian and
  L\"{u}tkepohl}{2017}]{Kilian2017}
\textsc{Kilian, L. and H.~L\"{u}tkepohl} (2017): \emph{{Structural Vector
  Autoregressive Analysis}}, Cambridge University Press.

\bibitem[\protect\citeauthoryear{Lazarus, Lewis, Stock, and Watson}{Lazarus
  et~al.}{2018}]{Lazarus2018}
\textsc{Lazarus, E., D.~J. Lewis, J.~H. Stock, and M.~W. Watson} (2018):
  \enquote{{HAR Inference: Recommendations for Practice},} \emph{Journal of
  Business \& Economic Statistics}, 36, 541--559.

\bibitem[\protect\citeauthoryear{Leeb and P\"{o}tscher}{Leeb and
  P\"{o}tscher}{2005}]{Leeb2005}
\textsc{Leeb, H. and B.~M. P\"{o}tscher} (2005): \enquote{{Model Selection and
  Inference: Facts and Fiction},} \emph{Econometric Theory}, 21, 21--59.

\bibitem[\protect\citeauthoryear{Mikusheva}{Mikusheva}{2007}]{Mikusheva2007}
\textsc{Mikusheva, A.} (2007): \enquote{{Uniform Inference in Autoregressive
  Models},} \emph{Econometrica}, 75, 1411--1452.

\bibitem[\protect\citeauthoryear{Mikusheva}{Mikusheva}{2012}]{Mikusheva2012}
---\hspace{-.1pt}---\hspace{-.1pt}--- (2012): \enquote{{One-Dimensional
  Inference in Autoregressive Models With the Potential Presence of a Unit
  Root},} \emph{Econometrica}, 80, 173--212.

\bibitem[\protect\citeauthoryear{Montiel~Olea and
  Plagborg-M{\o}ller}{Montiel~Olea and
  Plagborg-M{\o}ller}{2019}]{MontielOlea2019}
\textsc{Montiel~Olea, J.~L. and M.~Plagborg-M{\o}ller} (2019):
  \enquote{{Simultaneous confidence bands: Theory, implementation, and an
  application to SVARs},} \emph{Journal of Applied Econometrics}, 34, 1--17.

\bibitem[\protect\citeauthoryear{Nakamura and Steinsson}{Nakamura and
  Steinsson}{2018}]{Nakamura2018}
\textsc{Nakamura, E. and J.~Steinsson} (2018): \enquote{{Identification in
  Macroeconomics},} \emph{Journal of Economic Perspectives}, 32, 59--86.

\bibitem[\protect\citeauthoryear{Pesavento and Rossi}{Pesavento and
  Rossi}{2006}]{Pesavento2006}
\textsc{Pesavento, E. and B.~Rossi} (2006): \enquote{{Small-sample confidence
  intervals for multivariate impulse response functions at long horizons},}
  \emph{Journal of Applied Econometrics}, 21, 1135--1155.

\bibitem[\protect\citeauthoryear{Pesavento and Rossi}{Pesavento and
  Rossi}{2007}]{Pesavento2007}
---\hspace{-.1pt}---\hspace{-.1pt}--- (2007): \enquote{{Impulse response
  confidence intervals for persistent data: What have we learned?}}
  \emph{Journal of Economic Dynamics and Control}, 31, 2398--2412.

\bibitem[\protect\citeauthoryear{Phillips}{Phillips}{1988}]{Phillips1988}
\textsc{Phillips, P. C.~B.} (1988): \enquote{{Regression Theory for
  Near-Integrated Time Series},} \emph{Econometrica}, 56, 1021--1043.

\bibitem[\protect\citeauthoryear{Phillips}{Phillips}{1998}]{Phillips1998}
---\hspace{-.1pt}---\hspace{-.1pt}--- (1998): \enquote{{Impulse response and
  forecast error variance asymptotics in nonstationary VARs},} \emph{Journal of
  Econometrics}, 83, 21--56.

\bibitem[\protect\citeauthoryear{Phillips and Lee}{Phillips and
  Lee}{2013}]{Phillips2013}
\textsc{Phillips, P. C.~B. and J.~H. Lee} (2013): \enquote{{Predictive
  regression under various degrees of persistence and robust long-horizon
  regression},} \emph{Journal of Econometrics}, 177, 250--264, special issue on
  ``Dynamic Econometric Modeling and Forecasting''.

\bibitem[\protect\citeauthoryear{Plagborg-M{\o}ller and
  Wolf}{Plagborg-M{\o}ller and Wolf}{2020}]{PlagborgMoller2019}
\textsc{Plagborg-M{\o}ller, M. and C.~K. Wolf} (2020): \enquote{{Local
  Projections and VARs Estimate the Same Impulse Responses},}
  \emph{Econometrica}, forthcoming.

\bibitem[\protect\citeauthoryear{Pope}{Pope}{1990}]{Pope1990}
\textsc{Pope, A.~L.} (1990): \enquote{Biases of {Estimators} in {Multivariate}
  {Non}-{Gaussian} {Autoregressions},} \emph{Journal of Time Series Analysis},
  11, 249--258.

\bibitem[\protect\citeauthoryear{Rambachan and Shephard}{Rambachan and
  Shephard}{2019}]{rambachan2019econometric}
\textsc{Rambachan, A. and N.~Shephard} (2019): \enquote{Econometric analysis of
  potential outcomes time series: instruments, shocks, linearity and the causal
  response function,} ArXiv: 1903.01637.

\bibitem[\protect\citeauthoryear{Ramey}{Ramey}{2016}]{Ramey2016}
\textsc{Ramey, V.~A.} (2016): \enquote{{Macroeconomic Shocks and Their
  Propagation},} in \emph{Handbook of Macroeconomics}, ed. by J.~B. Taylor and
  H.~Uhlig, Elsevier, vol.~2, chap.~2, 71--162.

\bibitem[\protect\citeauthoryear{Richardson and Stock}{Richardson and
  Stock}{1989}]{Richardson1989}
\textsc{Richardson, M. and J.~H. Stock} (1989): \enquote{{Drawing inferences
  from statistics based on multiyear asset returns},} \emph{Journal of
  Financial Economics}, 25, 323--348.

\bibitem[\protect\citeauthoryear{Sims, Stock, and Watson}{Sims
  et~al.}{1990}]{Sims1990}
\textsc{Sims, C.~A., J.~H. Stock, and M.~W. Watson} (1990): \enquote{{Inference
  in Linear Time Series Models with Some Unit Roots},} \emph{Econometrica}, 58,
  113--144.

\bibitem[\protect\citeauthoryear{Stock and Watson}{Stock and
  Watson}{2018}]{Stock2018}
\textsc{Stock, J.~H. and M.~W. Watson} (2018): \enquote{{Identification and
  Estimation of Dynamic Causal Effects in Macroeconomics Using External
  Instruments},} \emph{Economic Journal}, 128, 917--948.

\bibitem[\protect\citeauthoryear{Toda and Yamamoto}{Toda and
  Yamamoto}{1995}]{Toda1995}
\textsc{Toda, H.~Y. and T.~Yamamoto} (1995): \enquote{{Statistical inference in
  vector autoregressions with possibly integrated processes},} \emph{Journal of
  Econometrics}, 66, 225--250.

\bibitem[\protect\citeauthoryear{Valkanov}{Valkanov}{2003}]{Valkanov2003}
\textsc{Valkanov, R.} (2003): \enquote{{Long-horizon regressions: theoretical
  results and applications},} \emph{Journal of Financial Economics}, 68,
  201--232.

\bibitem[\protect\citeauthoryear{Wright}{Wright}{2000}]{Wright2000}
\textsc{Wright, J.~H.} (2000): \enquote{{Confidence Intervals for Univariate
  Impulse Responses With a Near Unit Root},} \emph{Journal of Business \&
  Economic Statistics}, 18, 368--373.

\end{thebibliography}


\begin{thebibliography}{5}
\newcommand{\enquote}[1]{``#1''}
\expandafter\ifx\csname natexlab\endcsname\relax\def\natexlab#1{#1}\fi

\bibitem[\protect\citeauthoryear{Brillinger}{Brillinger}{2001}]{Brillinger2001}
\textsc{Brillinger, D.~R.} (2001): \emph{{Time Series: Data Analysis and
  Theory}}, Classics in Applied Mathematics, SIAM.

\bibitem[\protect\citeauthoryear{Davidson}{Davidson}{1994}]{Davidson1994}
\textsc{Davidson, J.} (1994): \emph{{Stochastic Limit Theory: An Introduction
  for Econometricians}}, Advanced Texts in Econometrics, Oxford University
  Press.

\bibitem[\protect\citeauthoryear{Gertler and Karadi}{Gertler and
  Karadi}{2015}]{Gertler2015}
\textsc{Gertler, M. and P.~Karadi} (2015): \enquote{Monetary Policy Surprises,
  Credit Costs, and Economic Activity,} \emph{American Economic Journal:
  Macroeconomics}, 7, 44--76.

\bibitem[\protect\citeauthoryear{Inoue and Kilian}{Inoue and
  Kilian}{2020}]{Inoue2020}
\textsc{Inoue, A. and L.~Kilian} (2020): \enquote{{The uniform validity of
  impulse response inference in autoregressions},} \emph{Journal of
  Econometrics}, 215, 450--472.

\bibitem[\protect\citeauthoryear{Kilian and Kim}{Kilian and
  Kim}{2011}]{Kilian2011}
\textsc{Kilian, L. and Y.~J. Kim} (2011): \enquote{{How Reliable Are Local
  Projection Estimators of Impulse Responses?}} \emph{Review of Economics and
  Statistics}, 93, 1460--1466.

\end{thebibliography}

\end{document}


\title{\texorpdfstring{\vspace{-2em}}{}Online Appendix: Local Projection Inference is Simpler and More Robust Than You Think}
\author{Jos\'{e} Luis Montiel Olea \and Mikkel Plagborg-M{\o}ller}
\date{December 4, 2020}
\maketitle

\begin{appendices}
\crefalias{section}{sappsec}
\crefalias{subsection}{sappsubsec}
\crefalias{subsubsection}{sappsubsubsec}
\setcounter{section}{3}

\section{Further Simulation Results}
\label{sec:sim_further}
\paragraph{Bivariate VAR(4) model.}
We first consider the bivariate VAR($p_0$) model
\[\textstyle y_{1,t} = \rho y_{1,t-1} + u_{1,t},\quad (1-\frac{1}{2}L)^{p_0} y_{2,t} = \frac{1}{2}y_{1,t-1} + u_{2,t},\quad (u_{1,t},u_{2,t})' \stackrel{i.i.d.}{\sim} N\left(0, \begin{psmallmatrix}
1 & 0.3 \\
0.3 & 1
\end{psmallmatrix}\right),\]
where $L$ is the lag operator, and the parameter $\rho$ indexes the persistence. For $p_0=1$, this model reduces to the one considered by \citet[section III]{Kilian2011}; we instead set $p_0=4$ to generate richer dynamics. The parameters of interest are the reduced-form impulse responses of $y_{2,t}$ with respect to the innovation $u_{1,t}$.

\cref{tab:TableMC_var_biv} shows that the qualitative conclusions from the AR(1) simulation study in \cref{sec:sim_ar1} carry over to the present bivariate DGP with $p_0=4$. We employ four different inference methods that use the correct estimation lag length $p=p_0$: non-augmented VAR, delta method confidence interval (``AR''); lag-augmented VAR \citep{Inoue2020}, Efron bootstrap interval (``AR-LA$_b$''); local projection with HAR standard errors as in \cref{sec:sim_ar1}, percentile-t bootstrap interval (``LP$_b$''); and our preferred method, lag-augmented local projection with heteroskedasticity-robust standard errors, percentile-t bootstrap interval (``LP-LA$_b$''). As a fifth method, we consider our preferred procedure with a larger estimation lag length $p=8$ (``LP-LA$_b^8$''). The bootstrap is a wild recursive residual VAR bootstrap. We set $T=240$. The nominal confidence level is 90\%.

Consistent with the theory in \cref{sec:var}, lag-augmented local projection achieves good coverage in all cases, except at long horizons $h \geq 36$ when there is a unit root ($\rho=1$). Over-specifying the lag length to be 8 instead of 4 barely affects the coverage of lag-augmented local projection confidence intervals and only widens them by 3--5\% (see columns 2 and 7). Non-augmented delta method VAR inference suffers from poor coverage at long horizons when $\rho \geq 0.95$, while lag-augmented VAR confidence intervals can be very wide.

\afterpage{
\begin{landscape}
\begin{table}[p]
    \centering
    \caption{Monte Carlo results: bivariate VAR(4) model}
    \vspace{0.5\baselineskip}
    \begin{tabular}{r|ccccc|ccccc}
& \multicolumn{5}{c|}{Coverage} & \multicolumn{5}{c}{Median length} \\
$h$ & $\text{LP-LA}_b$ & $\text{LP-LA}_b^8$ & $\text{LP}_b$ & $\text{AR-LA}_b$ & $\text{AR}$ & $\text{LP-LA}_b$ & $\text{LP-LA}_b^8$ & $\text{LP}_b$ & $\text{AR-LA}_b$ & $\text{AR}$ \\
\hline
\multicolumn{11}{c}{$\rho = 0.00$} \\
  1 & 0.910 & 0.906 & 0.906 & 0.901 & 0.902 & 0.234 & 0.241 & 0.245 & 0.229 & 0.226 \\
  6 & 0.892 & 0.892 & 0.899 & 0.894 & 0.895 & 1.481 & 1.518 & 1.517 & 1.310 & 1.278 \\
 12 & 0.895 & 0.889 & 0.895 & 0.903 & 0.901 & 1.605 & 1.661 & 1.627 & 3.813 & 0.660 \\
 36 & 0.906 & 0.901 & 0.905 & 0.924 & 1.000 & 1.694 & 1.754 & 1.709 & 30.081 & 0.015 \\
 60 & 0.913 & 0.912 & 0.911 & 0.927 & 1.000 & 1.825 & 1.901 & 1.832 & 301.439 & 0.000 \\
\multicolumn{11}{c}{$\rho = 0.50$} \\
  1 & 0.908 & 0.906 & 0.907 & 0.900 & 0.900 & 0.235 & 0.240 & 0.244 & 0.228 & 0.226 \\
  6 & 0.896 & 0.890 & 0.894 & 0.892 & 0.889 & 1.731 & 1.774 & 1.776 & 1.706 & 1.624 \\
 12 & 0.891 & 0.880 & 0.889 & 0.889 & 0.897 & 2.006 & 2.065 & 2.037 & 7.186 & 1.264 \\
 36 & 0.902 & 0.897 & 0.902 & 0.922 & 1.000 & 2.079 & 2.148 & 2.101 & 89.302 & 0.066 \\
 60 & 0.913 & 0.909 & 0.906 & 0.922 & 1.000 & 2.239 & 2.322 & 2.262 & 1517.269 & 0.001 \\
\multicolumn{11}{c}{$\rho = 0.95$} \\
  1 & 0.904 & 0.902 & 0.907 & 0.895 & 0.893 & 0.235 & 0.241 & 0.245 & 0.230 & 0.227 \\
  6 & 0.891 & 0.890 & 0.888 & 0.887 & 0.889 & 2.296 & 2.361 & 2.362 & 2.407 & 2.136 \\
 12 & 0.890 & 0.884 & 0.891 & 0.902 & 0.881 & 4.542 & 4.665 & 4.641 & 16.838 & 4.014 \\
 36 & 0.830 & 0.809 & 0.832 & 0.933 & 0.841 & 6.295 & 6.421 & 6.407 & 1113.555 & 5.413 \\
 60 & 0.876 & 0.859 & 0.872 & 0.931 & 0.763 & 6.146 & 6.297 & 6.343 & 73988.007 & 3.253 \\
\multicolumn{11}{c}{$\rho = 1.00$} \\
  1 & 0.904 & 0.897 & 0.900 & 0.893 & 0.890 & 0.236 & 0.242 & 0.245 & 0.230 & 0.227 \\
  6 & 0.894 & 0.892 & 0.890 & 0.859 & 0.874 & 2.381 & 2.445 & 2.472 & 2.450 & 2.181 \\
 12 & 0.877 & 0.873 & 0.872 & 0.879 & 0.828 & 5.278 & 5.407 & 5.364 & 17.862 & 4.491 \\
 36 & 0.767 & 0.760 & 0.769 & 0.965 & 0.775 & 11.346 & 11.558 & 11.509 & 1311.475 & 8.200 \\
 60 & 0.659 & 0.654 & 0.677 & 0.961 & 0.751 & 12.436 & 12.355 & 12.750 & 95033.410 & 11.423 \\
\end{tabular}
    \label{tab:TableMC_var_biv}
    \\
    \vspace{0.5\baselineskip}
\begin{minipage}{1 \textwidth} 
{\footnotesize Coverage probability and median length of nominal 90\% confidence intervals at different horizons. Bivariate VAR(4) model with $\rho \in \lbrace 0,.5,.95,1\rbrace $, $T=240$. 5,000 Monte Carlo repetitions; 2,000 bootstrap iterations.} 
\end{minipage}
\end{table}
\end{landscape}
}

\paragraph{Empirically calibrated VAR(12) models.}
We additionally consider two empirically calibrated VAR(12) models with four or five observables. The first DGP broadly follows \citet[section IV]{Kilian2011} and is given by the empirical least-squares estimate of a workhorse monetary VAR model estimated on monthly U.S. data for 1984--2018 ($T=419$). The four variables in the empirical VAR are the Federal Funds Rate, the Chicago Fed National Activity Index, CPI inflation, and real commodity price inflation (CRB Raw Industrials deflated by CPI).\footnote{St. Louis FRED codes: CFNAI, CPIAUCSL, FEDFUNDS. Global Financial Data code: CMCRBIND.} The second DGP is based on the main specification in \citet{Gertler2015} estimated on their monthly data set for 1990--2012 ($T=270$).\footnote{The data was downloaded from: \url{https://www.aeaweb.org/articles?id=10.1257/mac.20130329}} The five variables are industrial production (log levels), CPI (log levels), the 1-year Treasury rate, the Excess Bond Premium, and a monetary shock series given by high-frequency changes in 3-month Federal Funds Futures prices around FOMC announcements. For both DGPs, we simulate data from a Gaussian VAR(12) model with true parameters given by the empirically estimated coefficients and innovation covariance matrix (but no intercept). The sample sizes are the same as in the real data, mentioned earlier.

\cref{fig:sim_kk_gk} shows that lag-augmented local projection achieves acceptable coverage in these empirically calibrated DGPs. The figure shows the coverage and median length of 90\% confidence intervals for reduced-form impulse responses of selected response variables with respect to an innovation in the Federal Funds Rate (first DGP) or the monetary shock series (second DGP). Our preferred lag-augmented local projection procedure (solid black line) exhibits coverage distortions below 5 percentage points at all horizons for four of the six impulse response functions shown. The distortions only approach 10 percentage points for two response variables at long horizons in the second DGP. This second DGP is very challenging: Four of the eigenvalues of the VAR companion matrix exceed 0.98 in magnitude, while the sample size (270) is small relative to the number of covariates in each equation (60 plus the intercept). The \citet{Inoue2020} procedure (dashed blue line) exhibits near-uniform coverage in both DGPs, but this comes at the expense of extremely large confidence interval length at medium and long horizons.

\begin{figure}[p]
\centering
\textsc{Monte Carlo results: Kilian-Kim VAR(12) specification} \\[1ex]
\includegraphics[width=0.85\linewidth,clip=true,trim=5em 2.5em 5em 1em]{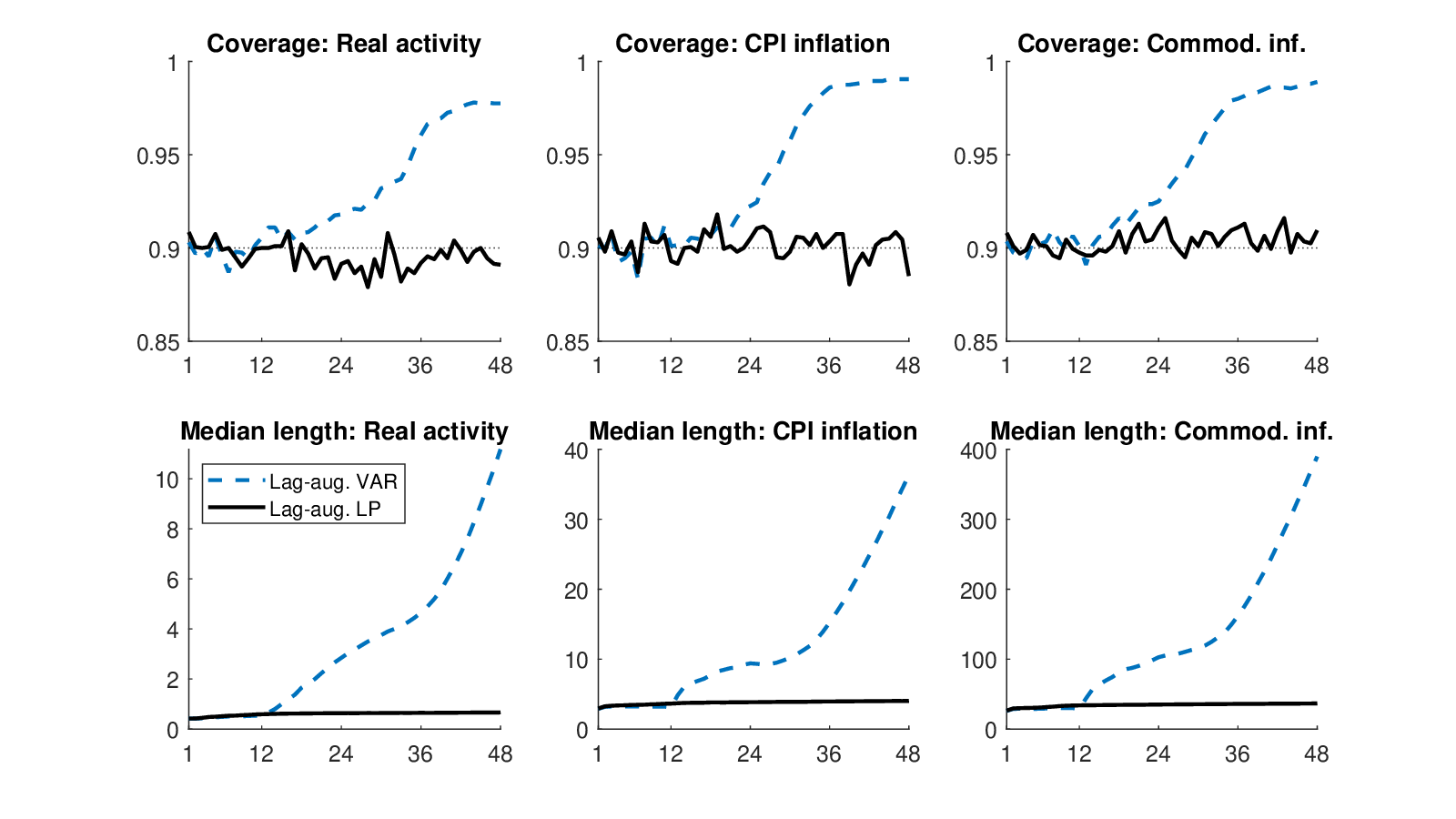} \\[\baselineskip]
\textsc{Monte Carlo results: Gertler-Karadi VAR(12) specification} \\[1ex]
\includegraphics[width=0.85\linewidth,clip=true,trim=5em 2.5em 5em 1em]{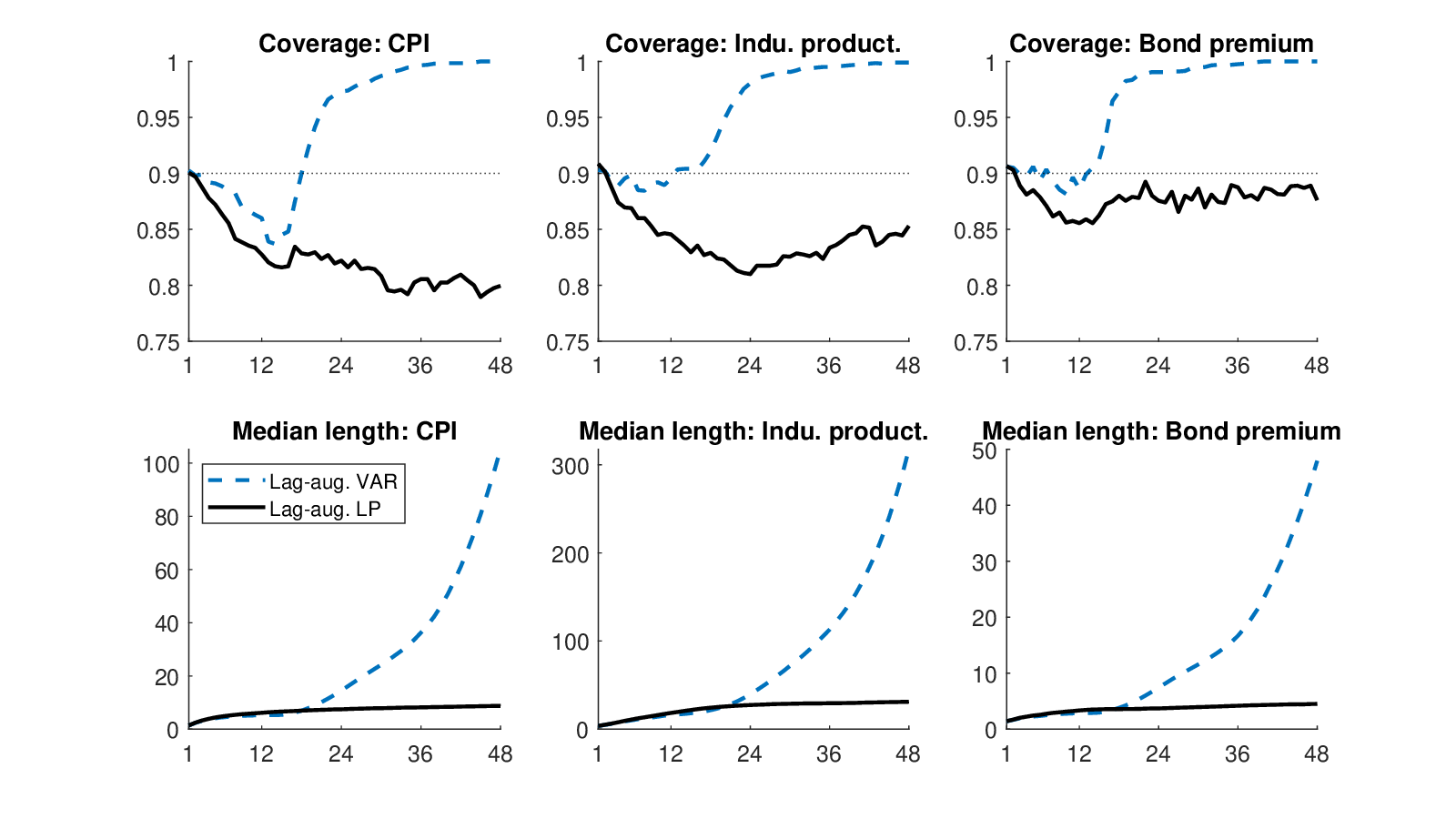}

\caption{Coverage rate and median length of 90\% confidence intervals for reduced-form impulse responses at horizons up to 48 (horizontal axis). Black solid line: lag-augmented local projection, percentile-t bootstrap interval. Blue dashed line: \citet{Inoue2020} Efron bootstrap interval. 2,000 Monte Carlo repetitions; 2,000 bootstrap iterations.} \label{fig:sim_kk_gk}
\end{figure}

\section{Additional Proofs}

\subsection{Notation}
\label{sec:var_notation}

%
%
%
%
%

Geometric series of the form $\sum_{\ell=0}^{h-1} (\rho_i^*(A,\epsilon))^{2\ell}$ will show up repeatedly in the proofs below. Observe that, for any $A \in \mathcal{A}(0,C,\epsilon)$  and $h \in \mathbb{N}$,
\[1 \leq \sum_{\ell=0}^{h-1}\rho^*_i(A,\epsilon)^{2\ell} \leq \min\left\lbrace \frac{1}{1-\rho^*_i(A,\epsilon)^2},h \right\rbrace \leq \min\left\lbrace \frac{1}{1-\rho^*_i(A,\epsilon)},h \right\rbrace = g(\rho_i^*(A,\epsilon),h)^2.\]
Recall also the definition of the lag-augmented LP residuals $\hat{\xi}_{1,t}(h) = y_{1,t+h} - \hat{\beta}_1(h)' y_{t} - \hat{\gamma}_1(h)' X_{t}$. We can write
\begin{align} 
\hat{\xi}_{1,t}(h) - \xi_{1,t}(\rho,h) &= (y_{1,t+h}- \hat{\beta}_{1}(h)'y_{t} - \hat{\gamma}_1(h)'X_{t})  -( y_{1,t+h}-\beta_{1}(A,h)'u_{t}-\eta_{1}(A,h)'X_{t}  ) \nonumber\\
&= - \hat{\beta}_1(h)'\underbrace{(y_t-A X_{t})}_{=u_t} - (\underbrace{\hat{\beta}_1(h)'A +\hat{\gamma}_1(h)'}_{\equiv \hat{\eta}_{1}(A,h)'})X_{t} +\beta_{1}(A,h)u_{t}+\eta_{1}(A,h)X_{t} \nonumber \\
&=
[\beta_1(A,h)-\hat{\beta}_1(h)]'u_t + [\eta_{1}(A,h)-\hat{\eta}_{1}(A,h)]'X_{t}. \label{eqn:var_xihat_error}
\end{align}

\subsection{Proof of \texorpdfstring{\cref{thm:var_se}}{Lemma \ref{thm:var_se}}} \label{sec:var_se_proof}

Define $\hat{\nu}(h_T) \equiv \hat{\Sigma}(h_T)^{-1}\nu$, where $\nu \in \mathbb{R} \backslash \{0 \}$ is a user-specified vector. The result follows from \cref{thm:var_se_infeas} if we can show that
\[\frac{\sum_{t=1}^{T-h_T} \hat{\xi}_{1,t}(h_T)^2 (\hat{\nu}(h_T)'\hat{u}_t(h_T))^2-\sum_{t=1}^{T-h_T} \xi_{1,t}(h_T)^2 (\tilde{\nu}'u_t)^2}{(T-h_T)v(A_T,h_T,\tilde{\nu})^2} \underset{P_{A_T}}{\overset{p}{\to}} 0,\]
where we have defined $\tilde{\nu} \equiv \Sigma^{-1} \nu$. Algebra shows that 
\begin{align*}
&\frac{\left|\sum_{t=1}^{T-h_T}\left[\hat{\xi}_{1,t}(h_T)^2(\hat{\nu}(h_T)'\hat{u}_t(h_T))^2 - \xi_{1,t}(A_T,h_T)^2(\tilde{\nu}'u_t)^2\right]\right|}{(T-h_T)v(A_T,h_T,\tilde{\nu})^2} \\
&\leq \frac{\sum_{t=1}^{T-h_T} \left|\hat{\xi}_{1,t}(h_T)^2(\hat{\nu}(h_T)'\hat{u}_t(h_T))^2 - \xi_{1,t}(A_T,h_T)^2(\tilde{\nu}'u_t)^2\right|}{(T-h_T)v(A_T,h_T,\tilde{\nu})^2} \\
&= \frac{1}{{(T-h_T)v(A_T,h_T,\tilde{\nu})^2}}\sum_{t=1}^{T-h_T} \left|\hat{\xi}_{1,t}(h_T)(\hat{\nu}(h_T)'\hat{u}_t(h_T)) - \xi_{1,t}(A_T,h_T)(\tilde{\nu}'u_t)\right| \\
&\qquad\qquad\qquad\qquad\qquad \times \left|\hat{\xi}_{1,t}(h_T)(\hat{\nu}(h_T)'\hat{u}_t(h_T)) - \xi_{1,t}(A_T,h_T)(\tilde{\nu}'u_t) + 2\xi_{1,t}(A_T,h_T)(\tilde{\nu}'u_t)\right| \\
& \textrm{(as $(a+b)(a-b)=a^2 -b^2)$} \\
&\leq \left(\frac{\sum_{t=1}^{T-h_T} \left[\hat{\xi}_{1,t}(h_T)(\hat{\nu}(h_T)'\hat{u}_t(h_T)) - \xi_{1,t}(A_T,h_T)(\tilde{\nu}'u_t)\right]^2}{(T-h_T)v(A_T,h_T,\tilde{\nu})^2}\right)^{1/2} \\
&\qquad \times \left(\frac{\sum_{t=1}^{T-h_T} \left[\hat{\xi}_{1,t}(h_T)(\hat{\nu}(h_T)'\hat{u}_t(h_T)) - \xi_{1,t}(A_T,h_T)(\tilde{\nu}'u_t) + 2\xi_{1,t}(A_T,h_T)(\tilde{\nu}'u_t) \right]^2}{(T-h_T)v(A_T,h_T,\tilde{\nu})^2}\right)^{1/2}.
\end{align*}
Consider the expression in the last line above. By Lo\`{e}ve's inequality \citep[Thm. 9.28]{Davidson1994}, this expression is bounded above by
\[ \left(2\frac{\sum_{t=1}^{T-h_T} \left[\hat{\xi}_{1,t}(h_T)(\hat{\nu}(h_T)'\hat{u}_t(h_T)) - \xi_{1,t}(A_T,h_T)(\tilde{\nu}'u_t)\right]^2}{(T-h_T)v(A_T,h_T,\tilde{\nu})^2} + 8\frac{\sum_{t=1}^{T-h_T} \xi_{1,t}(A_T,h_T)^2(\tilde{\nu}'u_t)^2}{(T-h_T)v(A_T,h_T,\tilde{\nu})^2}\right)^{1/2}.\]
The last fraction above is bounded in probability by \cref{thm:var_se_infeas}. Thus, it is sufficient to show that 
\[ \frac{\sum_{t=1}^{T-h_T} \left[\hat{\xi}_{1,t}(h_T)(\hat{\nu}(h_T)'\hat{u}_t(h_T)) - \xi_{1,t}(A_T,h_T)(\tilde{\nu}'u_t)\right]^2}{(T-h_T)v(A_T,h_T,\tilde{\nu})^2}\]
converges in probability to zero. To that end, decompose
\begin{align*}
& \hat{\xi}_{1,t}(h_T)(\hat{\nu}(h_T)'\hat{u}_t(h_T)) - \xi_{1,t}(A_T,h_T)(\tilde{\nu}'u_t) \\
&=  (\hat{\xi}_{1,t}(h_T)-\xi_{1,t}(A_T,h_T))(\tilde{\nu}'u_t) + (\hat{\nu}(h_T)'\hat{u}_t(h_T)-\tilde{\nu}'u_{T})\xi_{1,t}(A_T,h_T) \\
&\qquad + (\hat{\xi}_{1,t}(h_T)-\xi_{1,t}(A_T,h_T))(\hat{\nu}(h_T)'\hat{u}_t(h_T)-\tilde{\nu}'u_{T}).  
\end{align*}
Hence, by another application of Lo\`{e}ve's inequality,
\begin{align*}
&\frac{\sum_{t=1}^{T-h_T} \left[\hat{\xi}_{1,t}(h_T)(\hat{\nu}(h_T)'\hat{u}_t(h_T)) - \xi_{1,t}(A_T,h_T)(\tilde{\nu}'u_t)\right]^2}{(T-h_T)v(A_T,h_T,\tilde{\nu})^2} \\
&\leq 3\frac{\sum_{t=1}^{T-h_T} [\hat{\xi}_{1,t}(h_T)-\xi_{1,t}(A_T,h_T)]^2(\tilde{\nu}'u_t)^2}{(T-h_T)v(A_T,h_T,\tilde{\nu})^2} \\
&\qquad + 3 \frac{\sum_{t=1}^{T-h_T}[\hat{\nu}(h_T)'\hat{u}_t(h_T)-\tilde{\nu}'u_t]^2\xi_{1,t}(A_T,h_T)^2}{(T-h_T)v(A_T,h_T,\tilde{\nu})^2} \\
&\qquad + 3\frac{\sum_{t=1}^{T-h_T}[\hat{\xi}_{1,t}(h_T)-\xi_{1,t}(A_T,h_T)]^2[\hat{\nu}(h_T)'\hat{u}_t(h_T)-\tilde{\nu}'u_t]^2}{(T-h_T)v(A_T,h_T,\tilde{\nu})^2} \\
&\leq 3\left(\frac{\sum_{t=1}^{T-h_T} [\hat{\xi}_{1,t}(h_T)-\xi_{1,t}(A_T,h_T)]^4}{(T-h_T)v(A_T,h_T,\tilde{\nu})^4}\right)^{1/2} \times \left(\|\tilde{\nu}\|^4 \frac{\sum_{t=1}^{T-h_T} \|u_t\|^4}{T-h_T}\right)^{1/2} \\
&\qquad + 3\left(\frac{\sum_{t=1}^{T-h_T}[\hat{\nu}(h_T)'\hat{u}_t(h_T)-\tilde{\nu}'u_t]^4}{T-h_T}\right)^{1/2} \times \left(\frac{\sum_{t=1}^{T-h_T}\xi_{1,t}(A_T,h_T)^4}{(T-h_T)v(A_T,h_T,\tilde{\nu})^4}\right)^{1/2} \\
&\qquad + 3\left(\frac{\sum_{t=1}^{T-h_T} [\hat{\xi}_{1,t}(h_T)-\xi_t(A_T,h_T)]^4}{(T-h_T)v(A_T,h_T,\tilde{\nu})^4}\right)^{1/2} \times \left(\frac{\sum_{t=1}^{T-h_T}[\hat{\nu}(h_T)'\hat{u}_t(h_T)-\tilde{\nu}'u_t]^4}{T-h_T}\right)^{1/2} \\
& \textrm{(by Cauchy-Schwarz)}\\
&\equiv 3\left[ (\hat{R}_1)^{1/2} \times (\hat{R}_2)^{1/2} \right] + 3\left[(\hat{R}_3)^{1/2} \times (\hat{R}_4)^{1/2} \right] + 3 \left[(\hat{R}_1)^{1/2} \times (\hat{R}_3)^{1/2} \right].
\end{align*}
It follows from \cref{thm:var_xihat_conv} below that $\hat{R}_1$ tends to zero in probability. $\hat{R}_2$ is bounded in probability by \cref{asn:var_u_reg}(\ref{itm:var_asn_u_bounds}) and a standard application of Markov's inequality. We show below that $\hat{R}_3$ tends to zero in probability. Another standard application of Markov's inequality combined with \cref{thm:var_res_4th_bound} implies that $\hat{R}_4$ is also uniformly bounded in probability. Hence, the entire expression tends to zero in probability, as needed.

To finish the proof, we prove the claim that $\hat{R}_3$ tends to zero in probability. Note that
\[\hat{R}_3 \leq \|\hat{\nu}(h_T)\|^4 \frac{\sum_{t=1}^{T-h_T}\|\hat{u}_t(h_T)-u_t\|^4}{T-h_T} + \|\hat{\nu}(h_T)-\tilde{\nu}\|^4 \frac{\sum_{t=1}^{T-h_T}\|u_t\|^4}{T-h_T}.\]
Since $\|\hat{\nu}(h_T)-\tilde{\nu}\| \leq \|\hat{\Sigma}(h_T)^{-1}-\Sigma^{-1}\| \times \|\nu\|$, it follows from \cref{thm:var_Sigmahat_conv}(\ref{itm:var_Sigmahat_conv_ii}), \cref{thm:var_uhat_conv} below, \cref{asn:var_u_reg}(\ref{itm:var_asn_u_bounds}), and an application of Markov's inequality that the above display tends to zero in probability. \qed

\begin{lem}[Negligibility of estimation error in $\hat{\xi}_{1,t}(h)$] \label{thm:var_xihat_conv}
Let the conditions of \cref{thm:var_se} hold. Let $w \in \mathbb{R}^n \backslash \lbrace 0 \rbrace$. Then
\[\frac{\sum_{t=1}^{T-h_T}[\hat{\xi}_{1,t}(h_T)-\xi_{1,t}(A_T,h_T)]^4}{(T-h_T)v(A_T,h_T,w)^4} \underset{P_{A_T}}{\overset{p}{\to}} 0.\]
\end{lem}

\begin{proof}
Recall equation \eqref{eqn:var_xihat_error}:
\[\hat{\xi}_{1,t}(h) - \xi_{1,t}(A,h)=  [\beta_1(A,h)-\hat{\beta}_1(h)]'u_t + [\eta_1(A,h)-\hat{\eta}_1(A,h)]'X_t. \]
By Lo\`{e}ve's inequality \citep[Thm. 9.28]{Davidson1994},
\begin{align*}
&\frac{\sum_{t=1}^{T-h_T}[\hat{\xi}_{1,t}(h) - \xi_{1,t}(\rho,h)]^4}{(T-h_T)v(A_T,h_T,w)^4} \\
&\leq 8\frac{\|\hat{\beta}_1(h)-\beta_1(A_T,h_T)\|^4}{v(A_T,h_T,w)^4} \frac{\sum_{t=1}^{T-h_T}\|u_t\|^4}{T-h_T} \\
&\quad + 8\frac{\|G(A_T,T-h_T,\epsilon)[\hat{\eta}_1(A_T,h_T)-\eta_1(A_T,h_T)]\|^4}{v(A_T,h_T,w)^4} \frac{\sum_{t=1}^{T-h_T}\|G(A_T,T-h_T,\epsilon)^{-1}X_t\|^4}{T-h_T}.
\end{align*}
By \cref{asn:var_u_reg}(\ref{itm:var_asn_u_bounds}) and Markov's inequality, we have $(T-h_T)^{-1}\sum_{t=1}^{T-h_T}\|u_t\|^4 = O_{P_{A_T}}(1)$. \cref{thm:var_estim_conv}(\ref{itm:var_estim_conv_i}) then implies that the first term on the right-hand side in the above display tends to zero in probability. Similarly, the second term on the right-hand side of the above display tends to zero in probability by \cref{thm:var_y_4th_bound} below, \cref{thm:var_estim_conv}(\ref{itm:var_estim_conv_ii}), and Markov's inequality. 
\end{proof}

\begin{lem}[Negligibility of estimation error in $\hat{u}_t(h)$] \label{thm:var_uhat_conv}
Let the conditions of \cref{thm:var_se} hold. Then
\[\frac{\sum_{t=1}^{T-h_T}\|\hat{u}_t(h_T)-u_t\|^4}{T-h_T} \underset{P_{A_T}}{\overset{p}{\to}} 0.\]
\end{lem}

\begin{proof}
Since $\hat{u}_t(h_T)-u_t = [A-\hat{A}(h_T)]X_t$, we have
\[\frac{\sum_{t=1}^{T-h_T}\|\hat{u}_t(h_T)-u_t\|^4}{T-h_T} \leq \|G(A_T,T-h_T,\epsilon)(\hat{A}(h_T)-A_T)\|^4\frac{\sum_{t=1}^{T-h_T}\|G(A_T,T-h_T,\epsilon)^{-1}X_t\|^4}{T-h_T}.\]
\cref{thm:var_estim_conv}(\ref{itm:var_estim_conv_iii}) shows that the first factor after the inequality is $o_{P_{A_T}}(1)$. \Cref{thm:var_y_4th_bound} below and Markov's inequality show that the second factor is $O_{P_{A_T}}(1)$.
\end{proof}

\begin{lem}[Moment bound for $y_{i,t}^4$] \label{thm:var_y_4th_bound}
Let \cref{asn:u_mds} and \cref{asn:var_u_reg}(\ref{itm:var_asn_u_bounds}) hold. Then, for all $T \in \mathbb{N}$, $A \in \mathcal{A}(0,C,\epsilon)$, and $i=1,\dots,n$,
\[\max_{1 \leq t \leq T} E(y_{i,t}^4) \leq \frac{6C_1(E(\|u_t\|^4))^3}{\delta^2 \lambda_{\min}(\Sigma)^2} \times g(\rho_i^*(A,\epsilon),T)^4\]
where the expectations are taken with respect to the measure $P_A$, and $C_1$ is the constant defined in \cref{thm:var_bound_for_IRFs_A} below.
\end{lem}

\begin{proof}
We have defined 
\[\xi_{i,t}(A,h) \equiv \sum_{\ell=1}^h \beta_i(A,\ell)' u_{t+\ell}.\]
Since we have set the initial conditions $y_0=\ldots=y_{-p+1}=0$, we have 
\[ y_{i,t} = \sum_{\ell=1}^{t} \beta_i(A,\ell)' u_{\ell} = \xi_{i,0}(A,t).  \] 
Consider any $w \in \mathbb{R}^n$ such that $\|w\|=1$. Then \cref{thm:var_res_4th_bound} gives
\begin{align*}
\max_{1 \leq t \leq T} E(y_{i,t}^4) &= \max_{1 \leq t \leq T} E[\xi_{i,0}(A,t)^4] \\
&\leq \frac{6E(\|u_0\|^4)}{\delta^2 \lambda_{\min}(\Sigma)^2} \times \max_{1 \leq t \leq T} v(A,t,w)^4. 
\end{align*}
\cref{thm:var_v_bounds,thm:var_bound_for_IRFs_A} below then imply that
\begin{align*}
\pushQED{\qed}
\max_{1\leq t \leq T} E(y_{i,t}^4) &\leq \frac{6E(\|u_0\|^4)}{\delta^2 \lambda_{\min}(\Sigma)^2} \times (E[\|u_0\|^4])^2\|w\|^4 \times \max_{1 \leq t \leq T} \left(\sum_{\ell=0}^{t-1} \|\beta_i(A,\ell)\|^2\right)^2 \\
&= \frac{6(E(\|u_0\|^4))^3}{\delta^2 \lambda_{\min}(\Sigma)^2} \times \left(\sum_{\ell=0}^{T-1} \|\beta_i(A,\ell)\|^2\right)^2 \\
&\leq \frac{6(E(\|u_0\|^4))^3}{\delta^2 \lambda_{\min}(\Sigma)^2} \times \left(\sum_{\ell=0}^{T-1} C_1 \rho_i^*(A,\epsilon)^{2\ell}\right)^2 \\
&\leq \frac{6C_1^2(E(\|u_0\|^4))^3}{\delta^2 \lambda_{\min}(\Sigma)^2} \times g(\rho_i^*(A,\epsilon),T)^4. \qedhere
\end{align*}
\end{proof}

\begin{lem}  \label{thm:var_bound_for_IRFs_A}
Let $A(L)$ be a lag polynomial such that $A=(A_1,\dots,A_p) \in \mathcal{A}(0,C,\epsilon)$ for constants $C>0$ and $0<\epsilon<1$. Then, for any $i=1,\ldots, n$, the following statements hold.
\begin{enumerate}[i)]
\item $\| \beta_{i}(A,h) \| \leq  C_1 \rho_{i}^*(A,\epsilon)^h$, where $C_1 \equiv  1 + 2C \times \frac{1-\epsilon}{\epsilon}  $. \label{itm:var_bound_for_IRFs_A_i}
\item $\| \beta_{i}(A,h+m) \| \leq  \rho_{i}^*(A,\epsilon) ^m \times C_2 \sum_{b=0}^{p-1} \| \beta_{i}(A,h-b)\|$, where $C_2 \equiv  1+ 4 \tilde{C} \left( \frac{1-\epsilon}{\epsilon} \right)$, and $\tilde{C} \equiv C \left( 1 + C(p-1) \right)$. \label{itm:var_bound_for_IRFs_A_ii}
\end{enumerate}
\end{lem}

\begin{proof}
Since $A$ is in the parameter space $\mathcal{A}(0,C,\epsilon)$ in \cref{dfn:param_space},

\begin{equation}\label{eqn:var_aux_recurv}
\beta_{i}(A,h) = \rho_{i}\beta_{i} (A, h-1) + \beta_{i}(B,h).
\end{equation}
Thus, applying the equation above recursively,
\[\beta_{i}(A,h+m) = \rho_{i}^m \beta_{i} (A, h) +  \sum_{\ell=1}^{m} \rho_{i}^{m-\ell} \beta_{i}(B,h+\ell ).\]
We now use the above equation to prove each of the two statements of the lemma.

\paragraph{Part (\ref{itm:var_bound_for_IRFs_A_i}).}
We have
\begin{eqnarray*} 
\| \beta_{i}(A,h) \| &\leq& |\rho_{i}|^h \| \beta_{i}(A,0) \| + \sum_{\ell=1}^{h}  |\rho_{i}|^{h-\ell} \| \beta_{i}(B,\ell) \| \\
&\leq &  |\rho_{i}|^h + \sum_{\ell=1}^{h}  |\rho_{i}|^{h-\ell} C (1-\epsilon)^{\ell}\\
&& \textrm{(where we have used \cref{thm:var_bound_for_IRFs_B} below and $\beta(A,0)=I_n$}) \\ 
&\leq& \max\lbrace    |\rho_{i}| , 1-\epsilon/2 \rbrace ^h + \sum_{\ell=1}^{h}  \max\lbrace    |\rho_{i}| , 1-\epsilon/2 \rbrace ^{h-\ell} C (1-\epsilon)^{\ell}\\
&=&   \rho_{i}^*(A,\epsilon) ^h  \left(1+ C \sum_{\ell=1}^{h} \left( \frac{1-\epsilon}{\max\lbrace    |\rho_{i}| , 1-\epsilon/2 \rbrace  } \right)^\ell \right) \\
&\leq& \rho_{i}^*(A,\epsilon) ^h  \left(1+ C \sum_{\ell=1}^{\infty} \left( \frac{1-\epsilon}{1-\epsilon/2 } \right)^\ell \right) \\
&=&  \rho_{i}^*(A,\epsilon) ^h \left( 1 + C \left( \frac{1-\epsilon}{\epsilon/2}  \right) \right).
\end{eqnarray*}

\paragraph{Part (\ref{itm:var_bound_for_IRFs_A_ii}).}
To establish the remaining inequality, note that
\begin{align*}
\pushQED{\qed}
&\|\beta_{i}(A,h+m) \| \\
&\leq |\rho_{i}|^m \| \beta_{i}(A,h) \| + \sum_{\ell=1}^{m} |\rho_{i}|^{m-\ell} \| \beta_{i}(B,h+\ell) \| \\
&\leq   |\rho_{i}|^m \| \beta_{i}(A,h) \| + \sum_{\ell=1}^{m} |\rho_{i}|^{m-\ell} \left( \tilde{C} (1-\epsilon)^{\ell} \sum_{b=0}^{p-2} \| \beta_{i}(B,h-b) \| \right)  \\
& \textrm{(by \cref{thm:var_bound_for_IRFs_B}(\ref{itm:var_bound_for_IRFs_B_ii}) below)} \\
&\leq \max\lbrace  |\rho_{i}|, 1-\epsilon/2 \rbrace ^m \\
& \quad \times \left(  \| \beta_{i}(A,h) \| + \tilde{C}  \left( \sum_{\ell=1}^{m} \left( \frac{1-\epsilon}{\max\lbrace  |\rho_{i}|, 1-\epsilon/2 \rbrace  } \right)^{\ell} \right) \left(  \sum_{b=0}^{p-2} \| \beta_{i}(B,h-b) \| \right)  \right) \\
&\leq \rho_{i}^*(A,\epsilon) ^m \times \left(  \| \beta_{i}(A,h) \| + 2\tilde{C} \left( \frac{1-\epsilon}{\epsilon} \right)   \left(  \sum_{b=0}^{p-2} \| \beta_{i}(A,h-b) \|+\| \beta_{i}(A,h-b-1) \|  \right)  \right) \\
& \textrm{(where we have used equation \eqref{eqn:var_aux_recurv})} \\
&\leq \rho_{i}^*(A,\epsilon) ^m \times  \left(1+4\tilde{C} \left( \frac{1-\epsilon}{\epsilon} \right)  \right) \sum_{b=0}^{p-1} \| \beta_{i}(A,h-b) \|. \qedhere
\end{align*}
\end{proof}

\begin{lem}[Bounds on $v(A,h,w)$] \label{thm:var_v_bounds}
Let \cref{asn:u_mds} and \cref{asn:var_u_reg}(\ref{itm:var_asn_u_bounds}) hold. Then for any $i=1,\dots,n$ and for any matrix of autoregressive parameters $A$, and any  $h \in \mathbb{N}$
\[    \delta \times  \lambda_{\min} (\Sigma) \leq  \frac{1}{ \left \| a \right \|^2} \frac{v_i(A,h,w)^2}{\sum_{\ell=0}^{h-1} \left \| \beta_{i}(A,\ell) \right \|^2} \leq  E \left(  \left \| u_{t}  \right \| ^4 \right), \]
where $v_i(A,h,w) \equiv E[\xi_{i,t}(A,h)^2 (w'u_t)^2] $
\end{lem}

\begin{proof}
Algebra shows
\begin{align*}
v(A,h,w)^2 &=  E[\xi_{i,t}(A,h)^2 (w'u_t)^2]  \\
&=  E \left[ (\beta_{i}(A,h-1)' u_{t+1}+ \ldots + \beta_{i}(A,0)' u_{t+h})^2 u_{t}^2   \right] \\
&= E \left[ \left(\sum_{\ell=1}^{h}\sum_{m=1}^{h} \left( \beta_{i}(A,h-\ell)'u_{t+\ell} u'_{t+m} \beta_{i}(A,h-m)  \right) \right) (w' u_t)^{2}  \right].
\end{align*}
\cref{asn:u_mds} implies that the last expression above equals
\begin{equation} \label{eqn:var_aux1}
\sum_{\ell=1}^h E \left(  \left( \beta_{i}(A,h-\ell )'u_{t+\ell}  \right)^2 (w' u_t)^{2} \right) .
\end{equation}
An application of Cauchy-Schwarz gives the upper bound
\begin{eqnarray*}
v(A,h,w)^2 &\leq&  \sum_{\ell=1}^{h} E \left(  \left( \beta_{i}(A,h-\ell )'u_{t+\ell}  \right)^4 \right)^{1/2} E \left(  \left( w'u_{t}  \right)^4 \right)^{1/2} .\\
& \leq & \sum_{\ell=1}^{h} \left \| \beta_{i}(A,h-\ell) \right \|^2 E \left(  \left \| u_{t+\ell}  \right \| ^4 \right)^{1/2} \left \| w \right \|^2 E \left(  \left \| u_{t}  \right \| ^4 \right)^{1/2} \\
&=& E \left(  \left \| u_{t}  \right \| ^4 \right) \left \| w \right \|^2  \left( \sum_{\ell=0}^{h-1} \left \| \beta_{i}(A,\ell) \right \|^2 \right),
\end{eqnarray*}
where the last line follows from stationarity. 

For the lower bound, re-write expression \eqref{eqn:var_aux1} as

\[ \left \| w \right \|^2 \sum_{\ell=1}^h    \left \| \beta_{i}(A,h-\ell) \right \|^2 E \left(  \left( \omega_1 'u_{t+\ell}  \right)^2 (\omega_2' u_t)^{2} \right) . \]
where $\omega_1, \omega_2$ are vectors of unit norm. 

By \cref{asn:var_u_reg}(\ref{itm:var_asn_u_bounds}),
\begin{eqnarray*}
E \left(  \left( \omega_1 'u_{t+\ell}  \right)^2 (\omega_2' u_t)^{2} \right) &=& E\left[ E\big( \left( \omega_1 'u_{t+\ell}  \right)^2   \,\big|\, \lbrace u_s\rbrace _{s < t+\ell}  \big) (\omega_2' u_t)^{2} \right] \\
&\geq& \delta E[(\omega_2' u_t)^{2}] \\
&=& \delta  \omega_2'  E[u_t u_{t}']  \omega_2 \\
&\geq& \delta \lambda_{\min} (\Sigma).
\end{eqnarray*}
This gives the lower bound
\[ v(A,h,w)^2   \geq \left \| w \right \|^2   \delta \lambda_{\min} (\Sigma) \sum_{\ell=0}^{h-1} \left \| \beta_{i}(A,\ell) \right \|^2,\]
which concludes the proof.
\end{proof}

\begin{lem} \label{thm:var_Bound_h_m}
Partition the identity matrix $I_{np}$ of dimension $np \times np$ into $p$ column blocks of size $n$:
\[ I_{np} =   ( J_1', \ldots, J_p' ).   \]
Let $A(L)$ be a lag polynomial of order $p$ with autoregressive coefficients $A=(A_1, \ldots, A_{p})$.  Then, for any $h,m=0,1,\dots$,
\begin{align*}
\beta_{i}(A,h+m)' &=  \beta_{i}(A, h)' \: (J_1 \mathbf{A}^{m} J_1') \\
&\quad + \sum_{j=2}^{p} \left( \sum_{k=0}^{p-j} \beta_{i}(A, h-1-k)' A_{j+k}  \right)  \left( J_{j-1}  \mathbf{A}^{m-1} J_1' \right) ,
\end{align*}
where we define $\beta_i(A,\ell)=0$ for $\ell < 0$.
\end{lem}

\begin{proof}
Define $\beta(A,\ell) \equiv (\beta_1(A,\ell),\dots,\beta_n(A,\ell))'$. Then
\begin{eqnarray*}
\beta(A,h+m)& \equiv & J_1 \mathbf{A}^{h+m} J_1' \\
&=& J_1 \mathbf{A}^{h} \mathbf{A}^m J_1' \\
&=& J_1 \mathbf{A}^h I_{np} I_{np}' \mathbf{A}^m J_1' \\
&=& J_1 \mathbf{A}^h [ J_1', \ldots, J_{p}' ] \begin{bmatrix} J_1 \\ \vdots \\ J_{p} \end{bmatrix} \mathbf{A}^m J_1' \\
&=&  \left( J_1 \mathbf{A}^{h} J_1' \right) (J_1 \mathbf{A}^m J_1')  +   \sum_{j=2}^{p}  J_1 \mathbf{A}^{h} J_{j}' J_{j} \mathbf{A}^m J_1'  \\
&=& \beta(A,h) \beta(A,m) +  \sum_{j=2}^{p}  J_1 \mathbf{A}^{h} J_{j}'J_{j} \mathbf{A}^m J_1'.
\end{eqnarray*}
The definition of the companion matrix implies
\[J_j\mathbf{A} = J_{j-1},\quad j=2,\dots,p,\]
and
\[\mathbf{A} J_j'  =  J_1' A_j + J_{j+1}',\quad j=1,\dots,p-1,\quad \mathbf{A} J_p'  =  J_1' A_p.\]
Therefore, for $j \leq p$,
\begin{eqnarray*}
J_1 \mathbf{A}^h J_j' &=&  \sum_{k=0}^{p-j}  \beta(A,h-1-k) A_{j+k}.
\end{eqnarray*}
Thus, we have shown that
\begin{align*}
 \beta(A,h+m) &=  \beta(A,h) \beta(A,m)  \\
 &\qquad + \sum_{j=2}^{p} \left( \left( \sum_{k=0}^{p-j}  \beta(A,h-1-k) A_{j+k}  \right)  \left( J_{j-1}  \mathbf{A}^{m-1} J_1' \right) \right).
\end{align*}
The lemma follows by selecting the $i$-th equation of the above system of equations.
\end{proof}

\begin{lem} \label{thm:var_bound_for_IRFs_B}
Let $B(L)$ be a lag polynomial of order $p-1$ satisfying $\| \mathbf{B}^{\ell} \| \leq C(1-\epsilon)^{\ell}$ for every $\ell = 1,2, \dots$. Then the following two statements hold.
\begin{enumerate}[i)]
\item Define the $n \times n$ matrix $\beta(B,\ell) \equiv (\beta_1(B,\ell),\dots,\beta_n(B,\ell))'$. Then $\|\beta(B,\ell)\| \leq C (1-\epsilon)^{\ell}$ for all $\ell \geq 0$. \label{itm:var_bound_for_IRFs_B_i}
\item $\| \beta_{i}(B, h+m) \| \leq   \tilde{C} \times   (1-\epsilon)^m \times   \sum_{\ell=0}^{p-2} \| \beta_{i}(B, h-\ell ) \|$ for all $h,m \geq 0$, where $\tilde{C} \equiv C \left( 1 + C(p-1) \right)$. \label{itm:var_bound_for_IRFs_B_ii}
\end{enumerate}
\end{lem}

\begin{proof}
Let the selector matrix $J_j$ be defined as in \cref{thm:var_Bound_h_m}. Part (\ref{itm:var_bound_for_IRFs_B_i}) follows immediately from the fact
\[ \beta(B,\ell) = J_1 \mathbf{B}^{m} J_1' \]
and the assumed bound on $\|\mathbf{B}^m\|$.

We now turn to part (\ref{itm:var_bound_for_IRFs_B_ii}). \cref{thm:var_Bound_h_m} implies
\begin{align*}
\|\beta_{i}(B,h+m)\| &\leq  \|\beta_{i}(B, h)\| \times \|J_1 \mathbf{B}^{m} J_1'\|   \\
 &\qquad + \sum_{j=2}^{p-1} \left( \left( \sum_{k=0}^{p-1-j} \|\beta_{i}(B, h-1-k)\| \times 
\|B_{j+k}\|  \right)  \| J_{j-1}  \mathbf{B}^{m-1} J_1' \| \right) \\
 & \leq  \|  \beta_{i}(B, h)\|  \: C(1-\epsilon)^{m}  \\
 &\qquad + \sum_{j=2}^{p-1} \left( \left( \sum_{k=0}^{p-1-j} \| \beta_{i}(B, h-1-k) \| \times  \| B_{j+k} \|  \right)  C(1-\epsilon)^{m-1} \right)   \\
&  \textrm{(since $\|J_1 \mathbf{B}^{m} J_1'\| \leq C(1-\epsilon)^m $ and $\| J_{j-1}  \mathbf{B}^{m-1} J_1' \| \leq C(1-\epsilon)^{m-1} $ ) }\\
&\leq  C(1-\epsilon)^m \left(  \|  \beta_{i}(B, h)\| + \sum_{j=2}^{p-1} \left( \left( \sum_{k=0}^{p-1-j} \| \beta_{i}(B, h-1-k) \| \times  C  \right)  \right) \right)   \\
& \textrm{(since $\|B_{j+k}\| = \|J_1 \mathbf{B} J_{j+k}'\| \leq \|\mathbf{B}\| $)} \\
&\leq  C(1-\epsilon)^m \left(  \|  \beta_{i}(B, h)\| + C(p-2) \left( \sum_{k=0}^{p-3} \| \beta_{i}(B, h-1-k) \|    \right)   \right) \\
&\leq  (1-\epsilon)^m C \left( 1 + C(p-2) \right) \left( \sum_{\ell=0}^{p-2} \| \beta_{i}(B, h-\ell) \|    \right),\\
&\leq  (1-\epsilon)^m C \left( 1 + C(p-1) \right) \left( \sum_{\ell=0}^{p-2} \| \beta_{i}(B, h-\ell) \|    \right).
\end{align*}
The last step merely ensures that the constant is positive for all $p \geq 1$. Note that, in the case $p=1$, the sum in the last expression is zero.
\end{proof}

\subsection{Proof of \texorpdfstring{\cref{thm:var_estim_conv}}{Lemma \ref{thm:var_estim_conv}}} \label{sec:var_estim_conv_proof}
We first prove the statements (\ref{itm:var_estim_conv_i})--(\ref{itm:var_estim_conv_ii}), and then turn to statement (\ref{itm:var_estim_conv_iii}). For brevity, denote $G_T \equiv G(A_T,T-h_T,\epsilon)$.

\paragraph{Parts (\ref{itm:var_estim_conv_i})--(\ref{itm:var_estim_conv_ii}).}
Recall the definition $\hat{\eta}_1(A,h) \equiv A'\hat{\beta}_1(h)+\hat{\gamma}_1(h)$ in equation \eqref{eqn:var_xihat_error}. Since the OLS coefficients $(\hat{\beta}_1(h)',\hat{\eta}_1(A,h)')'$ are a non-singular linear transformation of the OLS coefficients $(\hat{\beta}_1(h)',\hat{\gamma}_1(h)')'$, the former vector equals the OLS coefficients in a regression of $y_{1,t+h}$ on $(u_t',X_t')'$, due to the relationship $u_t = y_t - A X_t$. By the representation 
\[y_{1,t+h} = \beta_1(A,h)'u_t + \eta_1(A,h)' X_t + \xi_{1,t}(A,h)\]
in equation \eqref{eqn:var_long_regression_u}, we can therefore write
\begin{align}
&\begin{pmatrix}
\frac{1}{v(A_T,h_T,w)}[\hat{\beta}_1(h_T)-\beta_1(A_T,h_T)] \\
\frac{1}{v(A_T,h_T,w)}G_T[\hat{\eta}(A_T,h_T)-\eta(A_T,h_T)]
\end{pmatrix} \nonumber \\
&=
\begin{pmatrix}
\frac{1}{T-h_T}\sum_{t=1}^{T-h_T} u_tu_t' & \frac{1}{T-h}\sum_{t=1}^{T-h_T} u_tX_t'G_T^{-1} \\
\frac{1}{T-h_T}\sum_{t=1}^{T-h_T} G_T^{-1}X_tu_t' & \frac{1}{T-h_T}\sum_{t=1}^{T-h_T} G_T^{-1}X_tX_t'G_T^{-1}
\end{pmatrix}
^{-1} \label{eqn:var_Mhat} \\
&\qquad \times 
\begin{pmatrix}
\frac{1}{(T-h_T)v(A_T,h_T,w)}\sum_{t=1}^{T-h_T} u_t \xi_{1,t}(A_T,h_T) \\
\frac{1}{(T-h_T)v(A_T,h_T,w)}\sum_{t=1}^{T-h_T} G_T^{-1} X_t\xi_{1,t}(A_T,h_T)
\end{pmatrix} \nonumber \\
&\equiv \hat{M}^{-1}
\begin{pmatrix}
\hat{m}_1 \\
\hat{m}_2
\end{pmatrix}. \nonumber
\end{align}
We must prove that the above display tends to zero in probability. $\hat{m}_1$ tends to zero in probability by \cref{thm:var_clt} and the fact that \cref{thm:var_v_bounds} implies that $v(A_T,h_T,w)/v(A_T,h_T,\tilde{w})$ is uniformly bounded from below and from above for any $\tilde{w} \in \mathbb{R}^n \backslash \lbrace 0 \rbrace$. $\hat{m}_2$ also tends to zero in probability by \cref{thm:var_estim_conv_numer}. Hence, it just remains to show that the $n(p+1) \times n(p+1)$ symmetric positive semidefinite matrix $\hat{M}^{-1}$ is bounded in probability. It suffices to show that $1/\lambda_{\min}(\hat{M})$ is uniformly asymptotically tight. Consider the $2 \times 2$ block partition of $\hat{M}$ in \eqref{eqn:var_Mhat}. The off-diagonal blocks of $\hat{M}$ tend to zero in probability by \cref{thm:var_estim_conv_denom} below. Moreover, the upper left block of $\hat{M}$ tends in probability to the positive definite matrix $\Sigma$ by \cref{thm:var_Sigmahat_conv}(\ref{itm:var_Sigmahat_conv_i}) and \cref{asn:var_u_reg}. Thus, the tightness of 1/$\lambda_{\min}(\hat{M})$ follows from \cref{asn:var_XpX}, which pertains to the lower right block of $\hat{M}$. This concludes the proof of the first two statements.

\paragraph{Part (\ref{itm:var_estim_conv_iii}).}
Write
\begin{align*}
& (T-h_T)^{1/2}[\hat{A}(h_T)-A_T]G(A_T,T-h_T,\epsilon) \\
&= \left(\frac{1}{(T-h_T)^{1/2}}\sum_{t=1}^{T-h_T} u_tX_t'G_T^{-1} \right) \times \left(\frac{1}{T-h_T}\sum_{t=1}^{T-h_T} G_T^{-1}X_tX_t'G_T^{-1} \right)^{-1}.
\end{align*}
The first factor on the right-hand side above is $O_{P_{A_T}}(1)$ by \cref{thm:var_estim_conv_denom} below, while the second factor is also $O_{P_{A_T}}(1)$ by the same argument as in parts (\ref{itm:var_estim_conv_i})--(\ref{itm:var_estim_conv_ii}) above. \qed

\begin{lem}[OLS denominator] \label{thm:var_estim_conv_denom}
Let \cref{asn:u_mds} and \cref{asn:var_u_reg}(\ref{itm:var_asn_u_bounds}) hold. Let there be given a sequence $\lbrace A_T \rbrace$ in $\mathcal{A}(0,C,\epsilon)$ and a sequence $\lbrace h_T \rbrace$ of nonnegative integers satisfying $T-h_T \to \infty$. Then for any $i,j=1,\dots,n$ and $r = 1,\dots,p$,
\[\frac{\sum_{t=1}^{T-h_T} u_{i,t}y_{j,t-r}}{(T-h_T)^{1/2}g(\rho_j^*(A,\epsilon),T-h_T) } = O_{P_{A_T}}(1).\]
\end{lem}

\begin{proof}
Write $g_{j,T} \equiv g(\rho_j^*(A,\epsilon),T-h_T)$ for brevity. Note that $\lbrace u_{i,t} y_{j,t-r} \rbrace_t$ is a martingale difference array with respect to the natural filtration $\tilde{\mathcal{F}}_t = \sigma(u_t,u_{t-1},\dots)$ under \cref{asn:u_mds}. Thus, the sequence is serially uncorrelated, implying that
\begin{align*}
E\left[\left(\frac{\sum_{t=1}^{T-h_T} u_{i,t}y_{j,t-r}}{(T-h_T)^{1/2}g_{j,T} }\right)^2\right] &= \frac{1}{(T-h_T)g_{j,T}^2} \sum_{t=1}^{T-h_T} E[u_{i,t}^2y_{j,t-r}^2] \\
&\leq \frac{1}{g_{j,T}^2} \times [E(u_{i,t}^4)]^{1/2} \times \max_{1 \leq t \leq T-h_T} E(y_{j,t-1}^4)^{1/2} \\
&= [E(u_{i,t}^4)]^{1/2} \times \left(\frac{\max_{1 \leq t \leq T-h_T} E(y_{j,t-1}^4)}{g_{j,T}^4}\right)^{1/2} \\
&\leq \frac{\sqrt{6}C_1(E(\|u_t\|^4))^2}{\delta \lambda_{\min}(\Sigma)},
\end{align*}
where the last inequality uses \cref{thm:var_y_4th_bound}. The lemma follows from Markov's inequality.
\end{proof}

\subsection{Proof of \texorpdfstring{\cref{thm:var_estim_conv_numer}}{Lemma \ref{thm:var_estim_conv_numer}}} \label{sec:var_estim_conv_numer_proof}

We will show that
\[E\left[\left(\frac{\sum_{t=1}^{T-h_T} \xi_{i,t}(A_T,h_T) y_{j,t-r}}{(T-h_T)v(A_T,h_T,w)g(\rho_j^*(A_{T},\epsilon),T-h_T)}\right)^2\right] \to 0.\]
To that end, observe that if $t \geq s+h_T$, then
\begin{align*}
&E[\xi_{i,t}(A_T,h_T) y_{j,t-r}\xi_{i,s}(\rho_T,h_T)y_{j,s-r}] \\
&= E\left[E(\xi_{i,t}(A_T,h_T) \mid u_t,u_{t-1},\dots)y_{j,t-r}\xi_s(A_T,h_T)y_{j,s-r}\right] \\
&= 0,
\end{align*}
by \cref{asn:u_mds}. By symmetry, the far left-hand side above equals 0 also if $s \geq t+h_T$. Thus,
\begin{align}
& E\left[\left(\frac{\sum_{t=1}^{T-h_T} \xi_{i,t}(A_T,h_T)y_{j,t-r}}{(T-h_T)v(A_T,h_T,w)g(\rho_j^*(A_{T},\epsilon),T-h_T)}\right)^2\right] \nonumber \\
&\leq \sum_{t=1}^{T-h_T}\sum_{s=1}^{T-h_T} \mathbbm{1}(|s-t|<h_T) \frac{|E[\xi_{i,t}(A_T,h_T)y_{j,t-r}\xi_{i,s}(A_T,h_T)y_{j,s-r}]|}{(T-h_T)^2v(A_T,h_T,w)^2g(\rho_j^*(A_{T},\epsilon),T-h_T)^2}. \label{eqn:var_estim_conv_numer_ii_bound}
\end{align}
We now bound the summands on the right-hand side above. Consider first the case $s \in (t-h_T, t]$ (we will handle the case $t \in (s-h_T,s]$ by symmetry). Since the initial conditions for the VAR are zero, we have
\[y_{j,t-r}= \xi_{j,0}(A_T,t-r).\]
Thus,
\begin{align*}
&E[\xi_{i,t}(A_T,h_T)y_{j,t-r}\xi_{i,s}(A_T,h_T)y_{j,s-r}] \\
&= E[\xi_{i,t}(A_T,h_T)  \xi_{j,0}(A_T,t-r) \xi_{i,s}(A_T,h_T) \xi_{0,j}(A_T,t-r) ]  \\
&= \sum_{\ell_1=1}^{h_T}\sum_{\ell_2=1}^{h_T} \sum_{m_1=r}^{t-1} \sum_{m_2=r}^{s-1} E\Big[ \left( \beta_{i}(A_{T},h_{T}-\ell_1)'u_{t+\ell_1} \right) ( \beta_{j}(A_{T}, m_1-r )' u_{t-m_1}) \\
&  \qquad\qquad\qquad\qquad\qquad \times ( \beta_{i}(A_{T},h_{T}-\ell_2)'u_{s+\ell_2} ) ( \beta_{j}(A_{T}, m_2-r)'u_{s-m_2} ) \Big].
\end{align*}
Consider any summand above defined by its indices $(\ell_1,\ell_2,m_1,m_2)$. Since $t+\ell_1 > \max\lbrace t-m_1,s-m_2\rbrace$, \cref{asn:u_mds} implies that the summand can only be nonzero if $s+\ell_2=t+\ell_1$, which requires $\ell_1 \leq h_T+s-t$. Moreover, when $s+\ell_2=t+\ell_1$, we also need $t-m_1=s-m_2$ for the summand to be nonzero, which in turn requires $m_1 \geq t-s+1$. Thus,
\begin{align*}
& \lvert E[\xi_{i,t}(A_T,h_T)y_{j,t-r}\xi_{i,s}(A_T,h_T)y_{j,s-r}] \rvert \\
&\leq \sum_{\ell_1=1}^{h_T+s-t} \sum_{m_1=t-s+r}^{t-r}  E \left [ \lvert ( \beta_{i}(A_{T},h_{T}-\ell_1)^{\prime }u_{t+\ell_1} ) ( \beta_{i}(A_{T},h_{T}-\ell_1-(t-s))^{\prime }u_{t+\ell_1} )   \right. \\
& \qquad\qquad\qquad\qquad\qquad \times \left.  ( \beta_{j}(A_{T},m_1-r)^{\prime }u_{m_1-r} )   ( \beta_{j}(A_{T},m_1-r-(t-s))^{\prime }u_{m_1-r} ) \rvert \right] \\
&= \sum_{\ell_1=1}^{h_T+s-t} \sum_{m_1=t-s+r}^{t-r} \left \| \beta_{i}(A_{T},h_{T}-\ell_1)  \right \| \times  \left \|  \beta_{i}(A_{T},h_{T}-\ell_1-(t-s)) \right \| \times \left \|  \beta_{j}(A_{T},m_1-r) \right \|  \\
&  \qquad\qquad\qquad\qquad \times \left \| \beta_{j}(A_{T},m_1-r-(t-s))\right \| \times E \left [ \| u_{t+\ell_1} \|^2 \times \| u_{m_1-r}  \|^2 \right] \\ 
& \textrm{(by Cauchy-Schwarz)}\\
& \leq C_1^2 E(\|u_0\|^4) \sum_{\ell_1=1}^{h_T+s-t} \sum_{m_1=t-s+r}^{t-r} \left \| \beta_{i}(A_{T},h_{T}-\ell_1)  \right \| \times \left \|  \beta_{i}(A_{T},h_{T}-\ell_1-(t-s)) \right \| \\
& \qquad\qquad\qquad\qquad\qquad\qquad\qquad \times  \rho^*_j(A_{T},\epsilon)^{2(m_1-r)-(t-s)}   \\
& \textrm{(since $\| \beta_{j} (A_{T},h) \| \leq C_1 \rho^*_j(A_{T},\epsilon)^{h}$ for any $j, h$ by \cref{thm:var_bound_for_IRFs_A})} \\
&\leq C_1^2 \times  E(\|u_0\|^4) \times  \rho^*_j(A_{T},\epsilon)^{(t-s)} \left( \sum_{\ell_1=1}^{h_T+s-t} \left \| \beta_{i}(A_{T},h_{T}-\ell_1)  \right \| \times  \left \|  \beta_{i}(A_{T},h_{T}-\ell_1-(t-s)) \right \|  \right) \\
& \qquad \times \left(  \sum_{m_1=t-s+r}^{t-r} \rho^*_j(A_T,\epsilon)^{2[m_1-r-(t-s)]} \right) \\
&\leq  E(\|u_0\|^4) \times  \rho^*_j(A,\epsilon)^{(t-s)} \left( \sum_{\ell_1=1}^{h_T+s-t} B^p_{i}(A_{T},h_{T}-\ell_1-(t-s)) \rho^*_i(A,\epsilon)^{(t-s)} \right) \\
& \qquad \times \left(  \sum_{m_1=t-s+r}^{t-r} \rho^*_j(A_T,\epsilon)^{2[m_1-r-(t-s)]} \right) \\
& \textrm{(using \cref{thm:var_bound_for_IRFs_A} and the definition of $B^p_{i}(A_{T},h_{T}-\ell_1-(t-s))$ in \cref{thm:var_bound_for_consistency} below)} \\
&= E(\|u_0\|^4) \times  \rho^*_j(A_T,\epsilon)^{(t-s)} \rho^*_i(A_T,\epsilon)^{(t-s)} \left( \sum_{\ell=0}^{h_T-1-(t-s)} B_i^{p}  (A_{T}, \ell)  \right) \left(  \sum_{m=0}^{s-2r} \rho^*_j(A_{T},\epsilon)^{2m} \right) \\
&\leq E( \|u_0\|^4)\times  \rho^*_j(A_T,\epsilon)^{(t-s)} \rho^*_i(A_T,\epsilon)^{(t-s)} \left( \sum_{\ell=0}^{h_T-1} B_i^{p}  (A_{T}, \ell)  \right) \left(  \sum_{m=0}^{T-h_{T}} \rho^*_j(A,\epsilon)^{2m} \right) \\
&\leq E( \|u_0\|^4)\times  \rho^*_j(A_T,\epsilon)^{(t-s)} \rho^*_i(A_{T},\epsilon)^{(t-s)}  \left( \sum_{\ell=0}^{h_T-1} B_i^{p}  (A_{T}, \ell)   \right) g(\rho_j^*(A_{T},\epsilon),T-h_T)^2  \\
&\leq E( \|u_0\|^4)\times  \rho^*_j(A_T,\epsilon)^{(t-s)} \rho^*_i(A_T,\epsilon)^{(t-s)} \\
&\qquad \times  C_2 p \left( \sum_{\ell=0}^{h_T-1} \| \beta_{i}(A_{T},\ell) \|^2   \right) g(\rho_j^*(A_{T},\epsilon),T-h_T)^2  \\
& \textrm{(by \cref{thm:var_bound_for_consistency} below).} 
\end{align*}
We have derived the bound in the above display under the assumption $s \in (t-h_T,t]$, but by symmetry, it also applies when $t \in (s-h_T,s]$ if we replace $(t-s)$ with $|t-s|$. Inserting into \eqref{eqn:var_estim_conv_numer_ii_bound}, we get
\begin{align*}
\pushQED{\qed}
&E\left[\left(\frac{\sum_{t=1}^{T-h_T} \xi_{i,t}(A_T,h_T)y_{j,t-r}}{(T-h_T)v(A_T,h_T,w)g(\rho_j^*(A_{T},\epsilon),T-h_T)}\right)^2\right] \\
&\leq   C_2p \times  \frac{E( \|u_0\|^4)}{(T-h_T)^2} \\
&\qquad \times  \frac{  \sum_{\ell=0}^{h_T-1} \| \beta_{i}(A_{T},\ell) \|^2   }{v(A_{T},h_{T},w)^2} \sum_{t=1}^{T-h_T}\sum_{s=1}^{T-h_T} \mathbbm{1}(|s-t|<h_T) \left( \rho^*_j(A_{T},\epsilon) \rho^*_i(A_{T},\epsilon) \right)^{|t-s|} \\
&\leq  \frac{C_2p}{\|w\|^2 \times  \delta \times  \lambda_{\min}(\Sigma)} \times  \frac{E( \|u_0\|^4)}{(T-h_T)^2}  \times    \sum_{t=1}^{T-h_T}\sum_{s=1}^{T-h_T} \mathbbm{1}(|s-t|<h_T) \left( \rho^*_j(A_{T},\epsilon) \rho^*_i(A_{T},\epsilon) \right)^{|t-s|} \\
& \textrm{(where we have used the lower bound of \cref{thm:var_v_bounds})} \\
&= \frac{C_2p}{\|w\|^2 \times  \delta \times  \lambda_{\min}(\Sigma)} \times  \frac{E( \|u_0\|^4)}{(T-h_T)}  \times  \sum_{|m|<h_T}\left(1 - \frac{|m|}{T-h_T}\right) \left( \rho^*_j(A_{T},\epsilon) \rho^*_i(A_{T},\epsilon) \right)^{|m|} \\
&\leq \frac{C_2p}{\|w\|^2 \times  \delta \times  \lambda_{\min}(\Sigma)} \times  \frac{E( \|u_0\|^4)}{(T-h_T)}  \times   \sum_{m=0}^{h_T-1} \left( \rho^*_j(A_{T}) \rho^*_i(A_{T},\epsilon) \right)^{m} \\
&\leq \frac{C_2p}{\|w\|^2 \times  \delta \times  \lambda_{\min}(\Sigma)} \times  \frac{E( \|u_0\|^4)}{(T-h_T)}  \times   \left( \sum_{m=0}^{h_T-1}  \rho^*_j(A_{T},\epsilon)^{2m} \right)^{1/2}  \left( \sum_{m=0}^{h_T-1} \rho^*_i(A_{T},\epsilon)^{2m} \right)^{1/2} \\
& \textrm{(by Cauchy-Schwarz)} \\
&\leq \frac{C_2p \times  E( \|u_0\|^4) }{\|w\|^2 \times  \delta \times  \lambda_{\min}(\Sigma)}  \times   \left( \frac{g(\rho_i^*(A_{T},\epsilon),T-h_T)}{T-h_T} \right)^{1/2} \left( \frac{g(\rho_j^*(A_{T},\epsilon),T-h_T)}{T-h_T} \right)^{1/2} \\
&\to 0. \qedhere
\end{align*}

\begin{lem} \label{thm:var_bound_for_consistency} 
Consider any lag polynomial $A(L)$ of order $p$ with autoregressive coefficients $A=(A_1, \ldots, A_{p})$. Then for any $h=1,2,\dots$,
\[ \frac{   \sum_{\ell=0}^{h-1}  B_{i}^p(A,\ell)    }{ \sum_{\ell=0}^{h-1}  \left \| \beta_{i}(A,\ell) \right \|^2  }  \leq C_2 p,  \]
where
\[B_{i}^p(A,\ell) \equiv C_2 \sum_{b=0}^{p-1} \left( \| \beta_{i}(A,\ell)\| \times  \left \|  \beta_i(A, \ell-b) \right \| \right),\]
and we define $\beta_i(A,\ell)=0$ whenever $\ell<0$. Here $C_2$ is the constant defined in \cref{thm:var_bound_for_IRFs_A}.
\end{lem}

\begin{proof}
Changing the order of summation, we have
\begin{align*}
&\sum_{\ell=0}^{h-1} \left( \sum_{b=0}^{p-1} \| \beta_{i}(A,\ell)\| \times  \left \|  \beta_i(A, \ell -b) \right \|   \right) \\
&= \sum_{b=0}^{p-1}\left( \sum_{\ell=0}^{h-1} \| \beta_{i}(A,\ell)\| \times  \left \|  \beta_i(A, \ell -b) \right \|   \right) \\
&\leq  \sum_{b=0}^{p-1}\left( \sum_{\ell=0}^{h-1} \| \beta_{i}(A,\ell)\|^2 \right)^{1/2}  \times \left(  \sum_{\ell=0}^{h-1} \left \|  \beta_i(A, \ell -b) \right \|^2   \right)^{1/2}\\
& \leq \sum_{b=0}^{p-1}\left( \sum_{\ell=0}^{h-1} \| \beta_{i}(A,\ell)\|^2 \right) \\
& \textrm{(since $\| \beta_{i}(A,\ell-b) \| = 0$ for $\ell-b < 0$)}\\
&= p \left( \sum_{\ell=0}^{h-1} \| \beta_{i}(A,\ell)\|^2 \right).
\end{align*}
Therefore, 
\[\pushQED{\qed}
\sum_{\ell=0}^{h-1}  B_{i}^p(A,\ell) \leq  C_2p \left( \sum_{\ell=0}^{h-1} \| \beta_{i}(A,\ell)\|^2 \right). \qedhere\]
\end{proof}

\subsection{Proof of \texorpdfstring{\cref{thm:var_Sigmahat_conv}}{Lemma \ref{thm:var_Sigmahat_conv}}} \label{sec:var_Sigmahat_conv_proof}

We consider each statement separately.

\paragraph{Part (\ref{itm:var_Sigmahat_conv_i}).}
Since $E(u_tu_t')=\Sigma$ by definition, this statement follows from a standard application of Chebyshev's inequality, exploiting the summability of the autocovariances of $\lbrace u_t \otimes u_t \rbrace$, cf. Assumption \ref{asn:var_u_reg}(\ref{itm:var_asn_u2_cum}). See for example \citet[Thm. 19.2]{Davidson1994}.

\paragraph{Part (\ref{itm:var_Sigmahat_conv_ii}).}
Using $\hat{u}_t(h)-u_t = (A-\hat{A}(h))X_t$, we get
\begin{align*}
& \left\|\hat{\Sigma}(h_T)-\frac{1}{T-h_T}\sum_{t=1}^{T-h_T}u_tu_t' \right\| \\
&\leq \frac{1}{T-h_T}\sum_{t=1}^{T-h_T} \|\hat{u}_t(h_T)\hat{u}_t(h_T)'-u_tu_t'\| \\
&\leq \frac{1}{T-h_T}\sum_{t=1}^{T-h_T} \|\hat{u}_t(h_T)-u_t\|^2 + \frac{2}{T-h_T}\sum_{t=1}^{T-h_T} \|(\hat{u}_t(h_T)-u_t)u_t'\| \\
&\leq \|G(A_T,T-h_T,\epsilon)(\hat{A}(h_T)-A_T)\|^2 \times \frac{1}{T-h_T}\sum_{t=1}^{T-h_T} \|G(A_T,T-h_T,\epsilon)^{-1}X_t\|^2 \\
&\qquad + 2 \times \|G(A_T,T-h_T,\epsilon)(\hat{A}(h_T)-A_T)\| \times \frac{1}{T-h_T}\sum_{t=1}^{T-h_T} \|G(A_T,T-h_T,\epsilon)^{-1}X_tu_t'\|.
\end{align*}
\cref{thm:var_y_4th_bound}, \cref{thm:var_estim_conv}(\ref{itm:var_estim_conv_iii}), \cref{thm:var_estim_conv_denom}, and an application of Markov's inequality imply that the last expression above is
\[\pushQED{\qed}
o_{P_{A_T}}(1) \times O_{P_{A_T}}(1) + 2 \times o_{P_{A_T}}(1) \times o_{P_{A_T}}(1) = o_{P_{A_T}}(1). \qedhere\]

\subsection{Proof of \texorpdfstring{\cref{thm:var_se_infeas}}{Lemma \ref{thm:var_se_infeas}}} \label{sec:var_se_infeas_proof}
We would like to show $\hat{\varsigma} \underset{P_{A_T}}{\overset{p}{\to}} 1$, where
\[\hat{\varsigma} \equiv \frac{1}{T-h_T}\sum_{t=1}^{T-h_T}\frac{\xi_{i,t}(A_T,h_T)^2(w'u_t)^2}{v(A_T,h_T,w)^2}.\]
Note that the summands could be serially correlated under our assumptions. We establish the desired convergence in probability by showing that the variance of $\hat{\varsigma}$ tends to 0 (since its mean is 1). Observe that
\begin{align}
\var(\hat{\varsigma}) &= \frac{1}{(T-h_T)^2v(A_T,h_T,w)^4}\sum_{t=1}^{T-h_T}\sum_{s=1}^{T-h_T} \cov\left(\xi_{i,t}(A_T,h_T)^2(w' u_t)^2,\xi_{i,s}(A_T,h_T)^2(w' u_s)^2\right) \nonumber \\
&= \frac{1}{(T-h_T)v(A_T,h_T,w)^4} \nonumber \\
& \qquad \times \sum_{|m|<T-h_T} \left(1-\frac{|m|}{T-h_T}\right) \cov\left(\xi_{i,0}(A_T,h_T)^2(w' u_0)^2,\xi_{i,m}(A_T,h_T)^2(w' u_m)^2\right) \nonumber \\
&\leq \frac{2}{(T-h_T)v(A_T,h_T,w)^4}\sum_{m=0}^{T-h_T} |\Gamma_T(m)|, \label{eqn:var_variance_varsigmahat}
\end{align}
where we define
\[\Gamma_T(m) \equiv \cov\left(\xi_{i,0}(A_T,h_T)^2(w' u_{i,0})^2,\xi_{i,m}(A_T,h_T)^2(w' u_m)^2\right),\quad m =0,1,2,\dots\]
By expanding the squares $\xi_0(\rho,h)^2$ and $\xi_m(\rho,h)^2$, we obtain
\begin{align*}
\Gamma_T(m) = \sum_{\ell_1=1}^{h_T} \sum_{\ell_2=1}^{h_T} \sum_{\ell_3=1}^{h_T} \sum_{\ell_4=1}^{h_T}  & \cov\Big((\beta_{i}(A_{T},h_{T}-\ell_1)' u_{\ell_1})   (\beta_{i}(A_{T},h_{T}-\ell_2)' u_{\ell_2})   (w' u_0)^2, \\
& \qquad\quad (\beta_{i}(A_{T},h_{T}-\ell_3)' u_{m+\ell_3}) (\beta_{i}(A_{T},h_{T}-\ell_4)' u_{m+\ell_4}) (w' u_m)^2 \Big).
\end{align*}
Consider any summand on the right-hand side above defined by indices $(\ell_1,\ell_2,\ell_3,\ell_4)$. If $\ell_1=\ell_2$, then \cref{asn:u_mds} implies that the covariance in the summand equals zero whenever $\ell_3 \neq \ell_4$, since in this case at most one of the subscripts $m+\ell_3$ or $m+\ell_4$ can equal $\ell_1$ ($=\ell_2$). Thus, if $\ell_1=\ell_2$, then the summand can only be nonzero when $\ell_3=\ell_4$. If instead $\ell_1 \neq \ell_2$, then \cref{asn:u_mds} implies that the summand can only be nonzero when $\lbrace \ell_1,\ell_2 \rbrace = \lbrace m+\ell_3,m+\ell_4 \rbrace$, which in turn requires that $m < h_T$. Putting these facts together, we obtain
\begin{align}
&|\Gamma_T(m)|  \nonumber \\
&\leq \sum_{\ell_1=1}^{h_T}  \sum_{\ell_3=1}^{h_T} \left|\cov\left(( \beta_{i}(A_{T},h_{T}-\ell_{1})' u_{m+\ell_1})^2   (w'u_m)^2, (\beta_{i}(A_{T},h_{T}-\ell_{3})'u_{\ell_3})^2  (w' u_0)^2 \right)\right| \label{eqn:var_variance_varsigmahat_term1}  \\
&\quad + \mathbbm{1}(m < h_T)2\sum_{\ell_1=1}^{h_T} \sum_{\ell_2 \neq \ell_1}  \left|\cov\left((\beta_{i}(A_{T},h_{T}-\ell_1)' u_{\ell_1})  (\beta_{i}(A_{T},h_{T}-\ell_2)' u_{\ell_2})   (w' u_m)^2, \right . \right . \nonumber \\
& \qquad\qquad\qquad\qquad \left . \left .  (\beta_{i}(A_{T},h_{T}-(\ell_1-m))' u_{\ell_1})  (\beta_{i}(A_{T},h_{T}-(\ell_2-m))' u_{\ell_2})   (w' u_0)^2 \right) \right|.  \label{eqn:var_variance_varsigmahat_term2}
\end{align}
Let $\tilde{\Gamma}_{1,T}(m)$ and $\tilde{\Gamma}_{2,T}(m)$ denote expressions \eqref{eqn:var_variance_varsigmahat_term1} and \eqref{eqn:var_variance_varsigmahat_term2}, respectively. We will now bound $\sum_{m=0}^{T-h_T} \tilde{\Gamma}_{1,T}(m)$ and $\sum_{m=0}^{T-h_T} \tilde{\Gamma}_{2,T}(m)$, so that we can ultimately insert these bounds into \eqref{eqn:var_variance_varsigmahat}.

\paragraph{Bound on $\sum_{m=0}^{T-h_{T}} \tilde{\Gamma}_{1,T}(m)$.}
We first bound the term in expression  \eqref{eqn:var_variance_varsigmahat_term1}. To do this, we define the unit-norm vectors
\[ \omega_{A_{T}, h_{T}, \ell} \equiv \beta_{i}(A_{T},h_{T}-\ell)/ \| \beta_{i}(A_{T},h_{T}-\ell)   \|, \quad \omega_{w} \equiv w / \| w \|.   \]
By \cref{thm:var_bound_for_IRFs_A}, the term
\[ \left|\cov\left(( \beta_{i}(A_{T},h_{T}-\ell_{1})' u_{m+\ell_1})^2  (w'u_m)^2, (\beta_{i}(A_{T},h_{T}-\ell_{3})'u_{\ell_3})^2   (w' u_0)^2 \right)\right| \]
is bounded above by
\[ \| w \|^4  C_1^4 \rho^*_i(A_{T},\epsilon) ^{2(h_{T}-\ell_1)+2(h_{T}-\ell_{3})} \left|\cov\left(( \omega_{A_{T}, h_{T}, \ell_1} ' u_{m+\ell_1})^2  (\omega_w'u_m)^2, (\omega_{A_{T}, h_{T}, \ell_3} 'u_{\ell_3})^2  (\omega_w' u_0)^2 \right)\right| .\]
Since $A_{T} \in \mathcal{A}(0,\epsilon,C)$, we have $\rho^*_i(A_{T},\epsilon) \leq 1$, so
\begin{align}
&\sum_{m=0}^{T-h_T} \tilde{\Gamma}_{1,T}(m) \nonumber \\
&\leq \|w\|^4 C_1^4 \sum_{m=0}^{T-h_T}\sum_{\ell_1=1}^{h_T}  \sum_{\ell_3=1}^{h_T}\rho^*_i(A_{T},\epsilon)^{2(h_T-\ell_3)} \nonumber \\
& \qquad\qquad\qquad\qquad\quad \times \left|\cov\left(( \omega_{A_{T}, h_{T},\ell_1} ' u_{m+\ell_1})^2  (\omega_w'u_m)^2, (\omega_{A_{T}, h_{T},\ell_3} 'u_{\ell_3})^2  (\omega_w' u_0)^2 \right)\right|  \nonumber \\
&\leq \|w\|^4 C_1^4 \sum_{b_1=1}^{h_T} \rho^*_i(A_{T},\epsilon)^{2(h_T-b_1)} \nonumber \\
& \qquad\qquad \times \left ( \sum_{b_2= -\infty}^{\infty} \sum_{b_3=-\infty}^{\infty}  \sup_{\|\omega_{j}\|=1} \left|\cov\left( ({\omega_1}' u_{b_1})^2(\omega_2' u_0)^2,(\omega_3' u_{b_3+b_{2}})^2(\omega_{4}' u_{b_3})^2 \right)\right|\right). \label{eqn:var_variance_varsigmahat_Gamma1} 
\end{align}
Consider the double sum in large parentheses above. If we expand the various squares of the form $(\omega_j'u_t)^2$, then the double sum can be bounded above by at most $4n^2$ terms of the form
\begin{equation} \label{eqn:var_varsigmahat_Gamma1_supp}
\sum_{b_2= -\infty}^{\infty} \sum_{b_3=-\infty}^{\infty} \left|\cov\left( \tilde{u}_{j_1,b_1}\tilde{u}_{j_2,0},\tilde{u}_{j_3,b_3+b_2}\tilde{u}_{j_4,b_3}\right)\right|,
\end{equation}
where $\tilde{u}_t = (\tilde{u}_{1,t},\dots,\tilde{u}_{n^2,t})' \equiv u_t \otimes u_t$, and $j_1,j_2,j_3,j_4 \in \lbrace 1,2,\dots,n^2 \rbrace$ are summation indices. By \cref{asn:var_u_reg}(\ref{itm:var_asn_u2_cum}), the process $\lbrace \tilde{u}_t \rbrace$ has absolutely summable cumulants up to order four. We can therefore show there exists a constant $K \in (0,\infty)$ such that the large parenthesis \eqref{eqn:var_variance_varsigmahat_Gamma1} is bounded above by $K$.\footnote{According to \citet[Thm. 2.3.2]{Brillinger2001},
\begin{align*}
\cov\left( \tilde{u}_{j_1,b_1} \tilde{u}_{j_2,0}  ,\tilde{u}_{j_3,b_2} \tilde{u}_{j_4,b_3} \right) &= \cov\left(\tilde{u}_{j_2,0},\tilde{u}_{j_3,b_2} \right)\cov\left( \tilde{u}_{j_1,b_1},\tilde{u}_{j_4,b_3}\right) + \cov\left(\tilde{u}_{j_2,0},\tilde{u}_{j_4,b_3} \right)\cov\left(\tilde{u}_{j_1,b_1},\tilde{u}_{j_3,b_2} \right) \\
&\quad + \text{Cum}\left( \tilde{u}_{j_2,0},\tilde{u}_{j_1,b_1},\tilde{u}_{j_3,b_2} ,\tilde{u}_{j_4,b_3}\right),
\end{align*}
where ``Cum'' denotes the joint fourth-order cumulant. Thus, the expression \eqref{eqn:var_varsigmahat_Gamma1_supp} is bounded above by
\begin{align*}
& \left(  \sum_{b_2=-\infty}^{\infty}  \left| \cov\left(\tilde{u}_{j_2,0},\tilde{u}_{j_3,b_2} \right) \right| \right) \left( \sum_{b_3=-\infty}^{\infty}\left| \cov\left( \tilde{u}_{j_1,b_1},\tilde{u}_{j_4,b_3}\right)  \right| \right) \\
&+ \left( \sum_{b_2=-\infty}^{\infty}   \left|  \cov\left(\tilde{u}_{j_1,b_1},\tilde{u}_{j_3,b_2} \right) \right| \right)  \left( \sum_{b_3=-\infty}^{\infty}   \left|  \cov\left(\tilde{u}_{j_2,0},\tilde{u}_{j_4,b_3} \right) \right| \right) \\
& + \sum_{b_1=-\infty}^{\infty} \sum_{b_2=-\infty}^{\infty}\sum_{b_3=-\infty}^{\infty}\left|  \text{Cum}\left( \tilde{u}_{j_2,0},\tilde{u}_{j_1,b_1},\tilde{u}_{j_3,b_2} ,\tilde{u}_{j_4,b_3}\right) \right|.
\end{align*}
The third term above is finite by \cref{asn:var_u_reg}(\ref{itm:var_asn_u2_cum}), since $\tilde{u_t} \equiv u_t \otimes u_t$ has absolutely summable cumulants up to order 4. Consider the first term above (the second term is handled similarly). The stationarity of $\tilde{u}_t$ implies that this term equals $\left(  \sum_{b_2=-\infty}^{\infty}  \left| \cov\left(\tilde{u}_{j_2,0},\tilde{u}_{j_3,b_2} \right) \right| \right) \left( \sum_{\ell=-\infty}^{\infty}\left| \cov\left( \tilde{u}_{j_1,0},\tilde{u}_{j_4,\ell}\right)  \right| \right)$. By \cref{asn:var_u_reg}(\ref{itm:var_asn_u2_cum}), the autocovariances of $  \{ \tilde{u_t} \} $ are absolutely summable. This implies the above display is bounded. Thus, we have shown that there exists a constant $K(j_1,j_2,j_3,j_4)$ (which only depends on the fixed data generating process for $\lbrace u_t \rbrace$) that bounds the expression \eqref{eqn:var_varsigmahat_Gamma1_supp}. Picking the largest constant over all summation indices gives the desired result. 
} Consequently, 

\[\sum_{m=0}^{T-h_T} \tilde{\Gamma}_{1,T}(m)  \leq \|w\|^4 C_1^4 K \sum_{b_1=1}^{h_T} \rho^*_i(A_{T},\epsilon)^{2(h_T-b_1)} = \|w\|^4 C_1^4 K \sum_{\ell=0}^{h_T-1} \rho^*_i(A_{T},\epsilon)^{2\ell }    .\]

\paragraph{Bound on $\sum_{m=0}^{T-h_{T}} \tilde{\Gamma}_{2,T}(m)$.}
Expression \eqref{eqn:var_variance_varsigmahat_term2} can be bounded above by
\begin{align*}
\mathbbm{1}(m < h_T)2\sum_{\ell_1=1}^{h_T} \sum_{\ell_2 \neq \ell_1}  & E \left[ \:  | \beta_{i}(A_{T},h_{T}-\ell_1)' u_{\ell_1} |  \times  | \beta_{i}(A_{T},h_{T}-\ell_2)' u_{\ell_2} |  \times  (w' u_m)^2 \right.  \\
& \left . | \beta_{i}(A_{T},h_{T}-(\ell_1-m))' u_{\ell_1} | \times   | \beta_{i}(A_{T},h_{T}-(\ell_2-m) )' u_{\ell_2} |  \times    (w'u_0)^2  \right] .
\end{align*}
Applying Cauchy-Schwarz, we get the upper bound
\begin{align*}
\mathbbm{1}(m < h_T)2\sum_{\ell_1=1}^{h_T} \sum_{\ell_2 \neq \ell_1} &   \Big( \: \|w\|^4 \times   \| \beta_{i}(A_{T},h_{T}-\ell_1)\| \times  \|\beta_{i}(A_{T},h_{T}-\ell_2) \|  \\
& \quad \times \left . \| \beta_{i}(A_{T},h_{T}-(\ell_1-m))\| \times  \| \beta_{i}(A_{T},h_{T}-(\ell_2-m))\|     \right .  \\
& \quad \times E \left[ \|u_{\ell_1}\|^2 \times  \| u_{\ell_2} \|^2 \times  \|u_m\|^2 \times  \| u_0 \| ^2 \right]  \: \Big).
\end{align*}
Another application of the Cauchy-Schwarz inequality gives 
\[E \left[ \|u_{\ell_1}\|^2 \times  \| u_{\ell_2} \|^2 \times  \|u_m\|^2 \times  \| u_0 \| ^2 \right] \leq E[ \: \| u_{t}^8 \| \: ]. \]
Thus,
\begin{align*}
&\sum_{m=0}^{T-h_T} \tilde{\Gamma}_{2,T}(m)  \\
& \leq  2 \times  E[ \: \| u_{t}^8 \| \: ] \times  \|w\|^4  \times   \sum_{m=0}^{h_T-1} \sum_{\ell_1=1}^{h_T} \sum_{\ell_2= 1}^{h_{T}}  \left(  \| \beta_{i}(A_{T},h_{T}-\ell_1)\| \times  \|\beta_{i}(A_{T},h_{T}-\ell_2) \| \right.   \\
&  \qquad\qquad\qquad\qquad\qquad\qquad \times \left.   \| \beta_{i}(A_{T},h_{T}-(\ell_1-m))\| \times  \| \beta_{i}(A_{T},h_{T}-(\ell_2-m))\| \right). 
\end{align*}
The bound in \cref{thm:var_bound_for_IRFs_A} implies that
\[\| \beta_{i}(A_{T},h_{T}-\ell_1)\| \times  \| \beta_{i}(A_{T},h_{T}-(\ell_1-m))\| \]
is less than or equal to
\begin{equation} \label{eqn:var_B_i_p}
\underbrace{ C_2 \sum_{b=0}^{p-1} \| \beta_{i}(A_{T},h_{T}-\ell_1)\| \times  \left \|  \beta_i(A_{T}, h_{T}-\ell_{1}-b) \right \| }_{ \equiv B^p_{i}(A_{T}, h_{T}-\ell_1) }  \times  \rho^*_i(A_{T},\epsilon)^{m},
\end{equation}
for a positive constant $C_2$ that depends on $p$ and $\epsilon$. Thus,
\begin{align}
&\sum_{m=0}^{T-h_T} \tilde{\Gamma}_{2,T}(m) \nonumber \\
& \leq  2 \times  E[ \: \| u_{t}^8 \| \: ] \times  \|w\|^4  \times   \sum_{m=0}^{h_T-1} \sum_{\ell_1=1}^{h_T} \sum_{\ell_2= 1}^{h_{T}}  \left(  B^p_i(A_{T}, h_{T}-\ell_1) \times  B^p_i(A_{T}, h_{T}-\ell_2) \times  \rho^*_i(A_T,\epsilon)^{2m} \right )   \nonumber \\
&= 2 \times  E[ \: \| u_{t}^8 \| \: ] \times  \|w\|^4  \left(\sum_{\ell=0}^{h_T-1} \rho^*_i(A_{T},\epsilon)^{2\ell} \right) \left( \sum_{\ell=0}^{h_{T}-1}  B_{i}^p(A_{T},\ell) \right)^2.  \label{eqn:var_variance_varsigmahat_Gamma2sum}
\end{align}

\paragraph{Conclusion of proof.}
Putting together \eqref{eqn:var_variance_varsigmahat}, \eqref{eqn:var_variance_varsigmahat_term1}, \eqref{eqn:var_variance_varsigmahat_term2},  and \eqref{eqn:var_variance_varsigmahat_Gamma2sum}, we get
\begin{align*}
\var(\hat{\varsigma}) &\leq \frac{2 \| w \|^4 }{(T-h_T)v(A_T,h_T,w)^4 }\left\lbrace  C_1^4 K \sum_{\ell=0}^{h_T-1} \rho^*_i(A_{T},\epsilon)^{2\ell }  \right. \\
& \qquad + \left.  2 \times  E[ \: \| u_{t}^8 \| \: ] \times   \left(\sum_{\ell=0}^{h_T-1} \rho^*_i(A_{T},\epsilon)^{2\ell} \right) \left( \sum_{\ell=0}^{h_{T}-1}  B_{i}^p(A_{T},\ell) \right)^2  \right\rbrace \\
&\leq \left \lbrace  \frac{2C_1^4 K \times \sum_{\ell=0}^{h_T-1} \rho^*_i(A_{T},\epsilon)^{2\ell } }{ (T-h_T) \left( \sum_{\ell=0}^{h_{T}-1}  \left \| \beta_{i}(A,\ell) \right \|^2   \right)^2 \delta^2 \lambda_{\min}(\Sigma)^2 }  \right .    \\
&\qquad  + \left .\frac{2 \times  E[ \: \| u_{t}^8 \| \: ] \times   \sum_{\ell=0}^{h_T-1} \rho^*_i(A_{T},\epsilon)^{2\ell}  }{(T-h_T)  \delta^2 \lambda_{\min}(\Sigma)^2  } \times   \frac{   \left( \sum_{\ell=0}^{h_{T}-1}  B_{i}^p(A_{T},\ell) \right)^2   }{ \left( \sum_{\ell=0}^{h_{T}-1}  \left \| \beta_{i}(A,\ell) \right \|^2   \right)^2 } \right \rbrace  \\
& \textrm{(by the lower bound for $v(A_T,h_T,w)^2$ derived in \cref{thm:var_v_bounds}) } \\
&\leq \frac{2 \left\lbrace \left( C_1^4 K \times \sum_{\ell=0}^{h_T-1} \rho^*_i(A_{T},\epsilon)^{2\ell } \right) + \left(  2 \times  E[ \: \| u_{t}^8 \| \: ] \times  C_2p \times   \sum_{\ell=0}^{h_T-1} \rho^*_i(A_{T},\epsilon)^{2\ell} \right)  \right\rbrace }{(T-h_T) \delta^2 \lambda_{\min} (\Sigma)^2  } \\
& \textrm{(where we have used $\textstyle \sum_{\ell=0}^{h-1}  \left \| \beta_{i}(A,\ell) \right \|^2  \geq \|\beta_{i}(A,0)\|=1$ and \cref{thm:var_bound_for_consistency})  } \\
&= \frac{ \left(2 \times C_1^4 K \right) +  \left( 4\times  E[ \: \| u_{t}^8 \| \: ] \times  C_2p  \right) }{ \delta^2 \lambda_{\min}(\Sigma)^2 } \times   \frac{\sum_{\ell=0}^{h_T-1} \rho^*_i(A_{T},\epsilon)^{2\ell }}{T-h_{T}} .
\end{align*}
The final expression above tends to zero as $T \to \infty$, since
\[\frac{\sum_{\ell=0}^{h_T-1} \rho^*_i(A_{T},\epsilon)^{2\ell}}{T-h_T} \leq  \frac{g(\rho^*_i(A_{T},\epsilon),h_T)^2}{T-h_T} \to 0.\]
Thus, $\var(\hat{\varsigma}) \to 0$.  \qed

\subsection{Proof of \texorpdfstring{\cref{thm:var_res_4th_bound}}{Lemma \ref{thm:var_res_4th_bound}}} \label{sec:var_res_4th_bound_proof}

We prove only the first statement of the lemma, as the proof is completely analogous for the second part. Define the unit-norm vectors
\[ \omega_{A, h, \ell} \equiv \beta_{i}(A,h-\ell)/ \| \beta_{i}(A,h-\ell)   \|, \quad \omega_{w} \equiv w / \| w \|.   \]
In a slight abuse notation, throughout the proof of this lemma we will sometimes write $\beta_{i}(h-\ell)$ instead of $\beta_{i}(A,h-\ell)$. Expanding the four-fold product $\xi_{i,t}(A,h)^4$, we obtain
\begin{align} 
& E[\xi_{i,t}(A,h)^4(a' u_t)^4] \nonumber \\
&= \sum_{\ell_1=1}^h \sum_{\ell_2=1}^h \sum_{\ell_3=1}^h \sum_{\ell_4=1}^h \| \beta_{i}(h-\ell_1)\| \times  \| \beta_{i}(h-\ell_2) \| \times  \| \beta_{i}(h-\ell_3) \| \times  \| \beta_{i}(h-\ell_4) \|  \nonumber  \\
&\qquad \times E \left[  (\omega_{A, h, \ell_1}' u_{t+\ell_1}) \times  (\omega_{A, h, \ell_2}'  u_{t+\ell_2}) \times  (   \omega_{A, h, \ell_3}'  u_{t+\ell_3} ) \times  ( \omega_{A, h, \ell_4}' u_{t+\ell_4}) \times  (w' u_t)^4 \right]. \label{eqn:var_xi4}
\end{align}
By \cref{asn:u_mds}, the summands above equal zero if one of the indices $\ell_j$ is different from the three other indices. Hence, the only possibly nonzero summands are those for which the four indices appear in two pairs, e.g., $\ell_1=\ell_3$ and $\ell_2=\ell_4$. The typical nonzero summand can thus be written in the form 
\[ \| \beta_{i}(h-\ell)\|^2 \| \beta_{i}(h-m)\|^2   E \left [  ( \omega_{A,h,\ell}' u_{t+\ell} )^2\times  (\omega_{A,h,m}'u_{t+m})^2\times  (w' u_t)^4 \right] \] 
where $\ell,m \in \lbrace 1,\dots,h\rbrace$. For given $\ell$ and $m$, this specific type of summand is obtained precisely when either (i) $\ell_1=\ell_2=\ell$ and $\ell_3=\ell_4=m$, or (ii) $\ell_1=\ell_3=\ell$ and $\ell_2=\ell_4=m$, or (iii) $\ell_1=\ell_4=\ell$ and $\ell_2=\ell_3=m$, or (iv) $\ell_1=\ell_2=m$ and $\ell_3=\ell_4=\ell$, or (v) $\ell_1=\ell_3=m$ and $\ell_2=\ell_4=\ell$, or (vi) $\ell_1=\ell_4=m$ and $\ell_2=\ell_3=\ell$. That is, there are six summands in \eqref{eqn:var_xi4} of this form. Thus,
\begin{align*}
E[\xi_{i,t}(A,h)^4(w' u_t)^4] &= 6 \sum_{\ell=1}^h\sum_{m=1}^h  \left(  \| \beta_{i}(h-\ell)\|^2 \| \beta_{i}(h-m)\|^2 \right. \\
&   \qquad\qquad\qquad \times \left . E \left [  ( \omega_{A,h,\ell}' u_{t+\ell} )^2\times  (\omega_{A,h,m}'u_{t+m})^2\times  (w' u_t)^4 \right]  \right)  \\
&\leq  6 \| w \|^4 E(\|u_t\|^8) \sum_{\ell=1}^h\sum_{m=1}^h  \| \beta_{i}(h-\ell)\|^2 \| \beta_{i}(h-m)\|^2 \\
& \textrm{(by applying Cauchy-Schwarz twice)}\\
&= 6 \|w\|^4 E(\|u_t\|^8)  \left(\sum_{\ell=0}^{h-1} \| \beta_{i}(A,h-\ell)\|^2  \right)^2.
\end{align*}
It follows from \cref{thm:var_v_bounds} that
\[\pushQED{\qed}
E\left[\left(v(A,h,w)^{-1}\xi_t(A,h)u_t\right)^4\right] \leq \frac{6 E(\|u_t\|^8) }{\delta^2 \lambda_{\min}(\Sigma)^2}. \qedhere\]

\phantomsection
\addcontentsline{toc}{section}{References}
\bibliography{ref}

\end{appendices}


\title{\texorpdfstring{\vspace{-2em}}{}Corrigendum: Local Projection Inference is Simpler and More Robust Than You Think\thanks{Email: {\tt montiel.olea@gmail.com}, {\tt mikkelpm@uchicago.edu}. We are grateful to Amilcar Velez and Richard Crump for extremely helpful discussions and suggestions, and to Mark Watson and Ke-Li Xu for useful comments. Any remaining errors are our own.}}
\author{Jos\'{e} Luis Montiel Olea \\ Cornell University \and Mikkel Plagborg-M{\o}ller \\ University of Chicago}

\maketitle

\noindent This document corrects two errors in \citet[henceforth ``MOPM'']{MOPM2021}, related to the formulation of Assumption 3 and to the proof of Lemma A.6.

\section*{Assumption 3}
The high-level Assumption 3 on p.\ 1805 of MOPM is more restrictive than intended. As stated, the assumption allows for any VAR(1) process, both stationary and nonstationary, as well as VAR($p$) processes with roots bounded away from the nonstationary part of the parameter space. However, several nonstationary VAR($p$) models with $p>1$ are ruled out. 

In this note, we therefore propose a modification of Assumption 3 that can be verified for a wide range of VAR($p$) models whose autoregressive parameters are contained in the parameter space defined on p.\ 1804 in MOPM. Our modified assumption is similar to Assumption 3 in the recent paper by \citet{Xu2022}, who applies the appropriate ``Dickey-Fuller'' transformation to the regressors.

If our modified Assumption 3 replaces the one in MOPM, all theoretical conclusions in our paper go through as originally stated. All econometric procedures, simulation results, efficiency calculations, and verbal discussions in our original paper are unaffected by the modification of Assumption 3.

\subsection*{Modified assumption}
We first state and discuss the modified assumption, and then we indicate the corresponding minor changes to the proofs.

Our modification simply amounts to redefining the $np \times np$ matrix $G(A,h,\epsilon)$ introduced on p.\ 1805 in MOPM. In place of the old definition, consider
\begin{align*}
G(A,h,\epsilon) = &\begin{pmatrix}
	I_n & 0_{n \times n} & 0_{n \times n} & 0_{n \times n} & \dots & 0_{n \times n} \\
	I_n & -\diag(\tilde{\rho}(A,\epsilon)) & 0_{n \times n} & 0_{n \times n} & \dots & 0_{n \times n} \\
	0_{n \times n} & I_n & -\diag(\tilde{\rho}(A,\epsilon))  & 0_{n \times n} & \dots & 0_{n \times n} \\
	& & \ddots & \ddots & \\
	0_{n \times n} & \dots & \dots & 0_{n \times n} & I_n & -\diag(\tilde{\rho}(A,\epsilon))
\end{pmatrix}^{-1} \\
&\times \begin{pmatrix}
\diag(g(\rho_1^*(A,\epsilon),h),\dots,g(\rho_n^*(A,\epsilon),h)) & 0_{n \times n(p-1)} \\
0_{n(p-1) \times n} & I_{n(p-1)}
\end{pmatrix},
\end{align*}
where both matrices above are $np \times np$, and we define $\tilde{\rho}(A,\epsilon) \equiv (\tilde{\rho}_1(A,\epsilon),\dots,\tilde{\rho}_n(A,\epsilon))'$ and $\tilde{\rho}_i(A,\epsilon) \equiv \text{sign}(\rho_i(A))\max\lbrace |\rho_i(A)|, \epsilon \rbrace$ for $i=1,\dots,n$.\footnote{Since $|\tilde{\rho}_i(A,\epsilon)|\geq \epsilon$, the first matrix in the above display is nonsingular.} Here and in the following we use several objects defined on p.\ 1805 in MOPM; in particular, $g(\rho,h)^2 \equiv \min\lbrace \frac{1}{1-|\rho|},h\rbrace $, $\rho_i^*(A,\epsilon) \equiv \max\lbrace |\rho_i(A)|,1-\epsilon/2\rbrace $, and $\rho_i(A)$ is the $i$-th potentially near-unity root for any VAR coefficients $A$ contained in the parameter space $\mathcal{A}(a,C,\epsilon)$ defined on p.\ 1804 in MOPM.
\setcounter{asn}{2}
\begin{asn}[modified]
For any $C>0$ and $\epsilon \in (0,1)$,
\[\lim_{K \to \infty}  \lim_{T\to\infty} \inf_{A \in \mathcal{A}(0,C,\epsilon)} P_A\left( \lambda_{\min}\left(G(A,T,\epsilon)^{-1}\left[\frac{1}{T}\sum_{t=1}^T X_tX_t'\right]G(A,T,\epsilon)^{-1\prime}  \right) \geq 1/K \right) = 1.\]
\end{asn}
This modified assumption is identical to the old one, except for the definition of the matrix $G(A,T,\epsilon)$. Define the quasi-differenced process $\tilde{y}_t(A,\epsilon) \equiv (\tilde{y}_{1,t}(A,\epsilon),\dots,\tilde{y}_{n,t}(A,\epsilon))'$ by $\tilde{y}_{i,t}(A,\epsilon) \equiv y_{i,t} - \tilde{\rho}_i(A,\epsilon)y_{i,t-1}$ for all $i$ and $t$. Our modified Assumption 3 then requires the sample second moment matrix of the scaled and transformed $np$-dimensional process
\[G(A,T,\epsilon)^{-1}X_t = \left(\frac{y_{1,t-1}}{g(\rho_1^*(A,\epsilon),T)},\dots,\frac{y_{n,t-1}}{g(\rho_n^*(A,\epsilon),T)}, \tilde{y}_{t-1}(A,\epsilon)', \dots, \tilde{y}_{t-p+1}(A,\epsilon)'\right)'\]
to be asymptotically uniformly nonsingular. Note that if $p=1$, then only the first $n$ elements appear in the above vector. It is standard to verify the asymptotic nonsingularity of this sample second moment matrix for stationary VAR($p$) parameter sequences $A=A_T$, as well as for parameter sequences that have a single local-to-unity or unit root per series $y_{i,t}$, $i=1,\dots,n$, as assumed in the parameter space in Definition 1 of MOPM (p.\ 1804).\footnote{Note that $g(\rho_i^*(A_T,\epsilon),T)^{-1} \propto T^{-1/2}$ for local-to-unity or unit root sequences $\rho_i(A_T)$.} See for example \citet[pp.\ 551--552]{Hamilton1994} for the unit root case and \citet[pp.\ 2754--2755]{Stock1994} for the local-to-unity case. Verifying \emph{uniform} nonsingularity requires additional steps, as in Appendix C of MOPM.

\subsection*{Modified proofs}
The only propositions or lemmas in MOPM (and the Supplemental Material) that rely on the specific definition of the matrix $G(A,h,\epsilon)$ are Proposition 1 and Lemmas A.3, A.5, E.1, and E.2.\footnote{See specifically the proof steps on p.\ 1812 and Supplemental Material pp.\ 7--8, 13--14, and 19. Note that certain expressions involving $G(A,h,\epsilon)$ should be transposed, as should be clear from context.} The proofs of these results go through unchanged, except that we must additionally show that the following three statements hold for any $j \in \lbrace 1,\dots,n \rbrace$ and $r \in \lbrace 1,\dots,p-1 \rbrace$, under the assumptions of Lemma A.3 in MOPM:
\begin{enumerate}[i)]
\item \label{itm:i1} $\frac{1}{(T-h_T)^{1/2}}\sum_{t=1}^{T-h_T} u_t \tilde{y}_{j,t-r}(A_T,\epsilon) = O_{P_{A_T}}(1)$.
\item \label{itm:i2} $\frac{1}{(T-h_T)v(A_T,h_T,w)}\sum_{t=1}^{T-h_T} \xi_{1,t}(A_T,h_T) \tilde{y}_{j,t-r}(A_T,\epsilon) = o_{P_{A_T}}(1)$.
\item \label{itm:i3} $\frac{1}{T-h_T}\sum_{t=1}^{T-h_T} (\tilde{y}_{j,t-r}(A_T,\epsilon))^4 = O_{P_{A_T}}(1)$.
\end{enumerate}
Note that if $|\rho_j(A_T)|\geq \epsilon$, then by Definition 1 in MOPM (p.\ 1804), $\tilde{y}_{j,t}(A_T,\epsilon)=y_{j,t} - \rho_j(A_T)y_{j,t-1} $ can be viewed as a component of a VAR($p-1$) process with coefficients $\tilde{A}_T$ contained in the uniformly stationary parameter space $\mathcal{A}(1,\epsilon,C)$. The proof of Lemma E.8 in MOPM (Supplemental Material pp.\ 14--15) then immediately implies that the expression on the left-hand side of (\ref{itm:i1}) has uniformly bounded second moment (simply substitute $\tilde{y}_{j,t}(A_T,\epsilon)$ for $y_{j,t}$ in the proof). If on the other hand $|\rho_j(A_T)| \leq \epsilon$, then $\frac{1}{(T-h_T)^{1/2}}\sum_{t=1}^{T-h_T} u_t \tilde{y}_{j,t-r}(A_T,\epsilon) = \frac{1}{(T-h_T)^{1/2}}\sum_{t=1}^{T-h_T} u_t y_{j,t-r} \pm \epsilon \frac{1}{(T-h_T)^{1/2}}\sum_{t=1}^{T-h_T} u_t y_{j,t-r-1}$, and the proof of Lemma E.8 in MOPM shows that both these terms have uniformly bounded second moments. Statement (\ref{itm:i1}) above follows.

Similarly, statement (\ref{itm:i3}) above follows directly from Markov's inequality and Lemma E.3 in MOPM (Supplemental Material p.\ 8). Again, in the case $|\rho_j(A_T)| \leq \epsilon$, we use convexity to derive the bound $(\tilde{y}_{j,t-r}(A_T,\epsilon))^4 \leq 8(y_{j,t-r}^4 + \epsilon^4 y_{j,t-r-1}^4)$ and treat the two terms separately.

Finally, to show statement (\ref{itm:i2}) above, it suffices by Chebyshev's inequality to show that the variance of the left-hand side is uniformly $o(1)$. For the case, $|\rho_j(A_T)| \leq \epsilon$ we can write $\frac{1}{(T-h_T)v(A_T,h_T,w)}\sum_{t=1}^{T-h_T} \xi_{1,t}(A_T,h_T) \tilde{y}_{j,t-r}(A_T,\epsilon) = \frac{1}{(T-h_T)v(A_T,h_T,w)}\sum_{t=1}^{T-h_T} \xi_{1,t}(A_T,h_T) y_{j,t-r} \pm \epsilon \frac{1}{(T-h_T)v(A_T,h_T,w)}\sum_{t=1}^{T-h_T} \xi_{1,t}(A_T,h_T) y_{j,t-r-1}$ and directly apply the proof of Lemma A.4 in MOPM (Supplemental Material pp.\ 15--18) to the two terms separately. For the case $|\rho_j(A_T)| \geq \epsilon$, the variance of the left-hand side in (\ref{itm:i2}) is given by
\[\frac{1}{(T-h_T)^2}\sum_{t=1}^{T-h_T}\sum_{s=1}^{T-h_T} \frac{E[\xi_{1,t}(A_T,h_T) \tilde{y}_{j,t-r}(A_T,\epsilon)\xi_{1,s}(A_T,h_T) \tilde{y}_{j,s-r}(A_T,\epsilon)]}{v(A_T,h_T,w)^2}.\]
We now argue that the summands in the double sum are bounded by a constant times $(1-\epsilon)^{|t-s|}$, which yields the desired conclusion. Consider the case $t \geq s$ (the case $t < s$ follows by symmetry). As argued above, we can write $\tilde{y}_{j,t}(A_T,\epsilon) = \sum_{\ell=1}^t \beta_j(\tilde{A}_T,t-\ell)'u_\ell$, where $\|\beta_j(\tilde{A}_T,\ell)\| \leq C(1-\epsilon)^\ell$ (recall that initial conditions are zero). Assumption 1 in MOPM (p.\ 1793) implies
\begin{align*}
&E[\tilde{y}_{j,s-r}(A_T,\epsilon)\tilde{y}_{j,t-r}(A_T,\epsilon) \mid u_{s+1},u_{s+2},\dots] \\
&= E\Big[\tilde{y}_{j,s-r}(A_T,\epsilon)\sum_{\ell=1}^s \beta_j(\tilde{A}_T,t-r-\ell)'u_\ell \;\Big|\; u_{s+1},u_{s+2},\dots \Big],
\end{align*}
where we define $\beta_j(\tilde{A}_T,\ell)=0_{n \times 1}$ for all $\ell < 0$. Thus,
\begin{align*}
&\left|E[\xi_{1,t}(A_T,h_T) \tilde{y}_{j,t-r}(A_T,\epsilon)\xi_{1,s}(A_T,h_T) \tilde{y}_{j,s-r}(A_T,\epsilon)]\right| \\
&= \left|E\left[\xi_{1,t}(A_T,h_T)\xi_{1,s}(A_T,h_T) E\lbrace \tilde{y}_{j,t-r}(A_T,\epsilon)\tilde{y}_{j,s-r}(A_T,\epsilon) \mid u_{s+1},u_{s+2},\dots \rbrace \right]\right| \\
&= \left|E\left[\xi_{1,t}(A_T,h_T)\xi_{1,s}(A_T,h_T)\tilde{y}_{j,s-r}(A_T,\epsilon)\sum_{\ell=1}^s \beta_j(\tilde{A}_T,t-r-\ell)'u_\ell  \right]\right| \\
&\leq \sum_{\ell=1}^s \|\beta_j(\tilde{A}_T,t-r-\ell)\| E\left[\left|\xi_{1,t}(A_T,h_T)\xi_{1,s}(A_T,h_T)\tilde{y}_{j,s-r}(A_T,\epsilon) \right|\|u_\ell\|  \right] \\
&\leq \sum_{\ell=1}^s \|\beta_j(\tilde{A}_T,t-r-\ell)\| \times  \max_{\tau_1,\tau_2,\tau_3 \leq T} \left(E[\xi_{1,\tau_1}(A_T,h_T)^4]^2 \times E[\tilde{y}_{j,\tau_2}(A_T,\epsilon)^4] \times E[\|u_{\tau_3}\|^4] \right)^{1/4}.
\end{align*}
Using Assumption 2, Lemma A.7, and Lemma E.3 in MOPM (pp.\ 1805 and 1815, and Supplemental Material p.\ 8), we therefore have
\begin{align*}
&\frac{|E[\xi_{1,t}(A_T,h_T) \tilde{y}_{j,t-r}(A_T,\epsilon)\xi_{1,s}(A_T,h_T) \tilde{y}_{j,s-r}(A_T,\epsilon)]|}{v(A_T,h_T,w)^2} \leq \text{constant} \times \left(\sum_{\ell=1}^s \|\beta_j(\tilde{A}_T,t-r-\ell)\|\right) \\
&\leq \text{constant} \times \left(\sum_{\ell=1}^s (1-\epsilon)^{t-r-\ell} \right) \leq \text{constant} \times (1-\epsilon)^{t-s},
\end{align*}
where the constants do not depend on $j$, $r$, $s$, $t$, $h_T$, or $A_T$. \qed

\section*{Proof of Lemma A.6}
Lemma A.6 is correct as stated, but the proof contains an error on pp.\ 22--23 of the supplement. The issue is that the cumulant identity at the beginning of Footnote 3 in the supplement only holds for mean-zero variables, while $\tilde{u}_t$ is not mean zero.

We correct the proof of Lemma A.6 as follows. First, we establish an alternative bound on $\sum_{m=0}^{T-h_T} \tilde{\Gamma}_{1,T}(m)$ (see pp.\ 21--22 of the supplement). Define unit-length vectors $\omega_{A_T,h_T,\ell}$ and $\omega_w$ as on p.\ 22 of the supplement. Then
\begin{align*}
&\sum_{m=0}^{T-h_T} \tilde{\Gamma}_{1,T}(m) \\
&= \| w \|^4 \sum_{\ell_1=1}^{h_T}  \|\beta_{i}(A_{T},h_{T}-\ell_{1})\|^2 \sum_{\ell_3=1}^{h_T}  \|\beta_{i}(A_{T},h_{T}-\ell_{3})\|^2   \\
&\qquad\qquad\qquad\qquad\qquad \times \sum_{m=0}^{T-h_T} \left|\cov\left(( \omega_{A_{T}, h_{T}, \ell_1} ' u_{m+\ell_1})^2  (\omega_w'u_m)^2, (\omega_{A_{T}, h_{T}, \ell_3} 'u_{\ell_3})^2  (\omega_w' u_0)^2 \right)\right|.
\end{align*}
Consider the sum over $m$ in the last line above. By expanding the squares and using that the $\omega$ vectors have unit length, we can see that this sum is bounded above by at most $4n^2$ terms of the form
\[\sum_{m=0}^{\infty} \left|\cov\left( \tilde{u}_{j_1,m+\ell_1}\tilde{u}_{j_2,m},\tilde{u}_{j_3,\ell_3}\tilde{u}_{j_4,0}\right)\right|,\]
where $\tilde{u}_t = (\tilde{u}_{1,t},\dots,\tilde{u}_{n^2,t})' \equiv u_t \otimes u_t$, and $j_1,j_2,j_3,j_4 \in \lbrace 1,2,\dots,n^2 \rbrace$ are summation indices. \cref{thm:cov_squared} below implies that the sum in the above display is uniformly (in $\ell_1,\ell_3,j_1,\dots,j_4$) bounded by a constant $K$ that depends only on the process for the shocks $u_t$. It follows that
\begin{align*}
\sum_{m=0}^{T-h_T} \tilde{\Gamma}_{1,T}(m) &\leq 4n^2K\| w \|^4 \sum_{\ell_1=1}^{h_T}  \|\beta_{i}(A_{T},h_{T}-\ell_{1})\|^2 \sum_{\ell_3=1}^{h_T}  \|\beta_{i}(A_{T},h_{T}-\ell_{3})\|^2 \\
&= 4n^2K\| w \|^4 \left(\sum_{\ell=0}^{h_T-1}  \|\beta_{i}(A_{T},\ell)\|^2\right)^2.
\end{align*}

We now show that the new bound is sufficient to establish the ``Conclusion of proof'' on pp.\ 24--25 of the supplement. The bound for $\tilde{\Gamma}_{2,T}(m)$ is unaffected, so we just need to show that
\[\frac{1}{(T-h_T)v(A_T,h_T,w)^4}\sum_{m=0}^{T-h_T}\tilde{\Gamma}_{1,T}(m) \to 0\quad \text{uniformly as } T\to\infty.\]
Indeed, using our earlier derivations in this document and Lemma E.5, we have
\begin{align*}
\frac{1}{(T-h_T)v(A_T,h_T,w)^4}\sum_{m=0}^{T-h_T}\tilde{\Gamma}_{1,T}(m) &\leq \frac{4n^2K\| w \|^4 \left(\sum_{\ell=0}^{h_T-1}  \|\beta_{i}(A_{T},\ell)\|^2\right)^2}{(T-h_T) \delta^2 \lambda_{\min}(\Sigma)^2 \|w\|^4 \left(\sum_{\ell=0}^{h_T-1}  \|\beta_{i}(A_{T},\ell)\|^2\right)^2} \\
&= \frac{1}{T-h_T} \times \frac{4n^2 K}{\delta^2 \lambda_{\min}(\Sigma)^2}.
\end{align*}
The second fraction above depends only on the fixed process for the shocks $u_t$. This concludes the proof correction.

\begin{lem} \label{thm:cov_squared}
Let $\lbrace Z_t \rbrace$ be a strictly stationary, $k$-dimensional time series with finite 8th moments and absolutely summable joint cumulants up to order 4. Then there exists a constant $K<\infty$ such that for any indices $j_1,\dots,j_4 \in \lbrace 1,\dots,k\rbrace$ and any integers $\ell_1,\ell_2$,
\[\sum_{m=-\infty}^\infty |\cov(Z_{j_1,m+\ell_1}Z_{j_2,m},Z_{j_3,\ell_2}Z_{j_4,0})| \leq K.\]
\end{lem}
\begin{proof}
To lighten notation, we drop the $j_1,\dots,j_4$ indices and treat $Z_t$ as univariate, noting that the multivariate proof follows the same logical structure, with notational adjustments due to the presence of subscripts $j_1, \ldots, j_4$. Define $\mu \equiv E(Z_t)$ and $\bar{Z}_t \equiv Z_t-\mu$. Expanding $Z_t = \bar{Z}_t + \mu$, and using the linearity of the covariance operator,
\begin{align*}
&\sum_{m=-\infty}^\infty |\cov(Z_{m+\ell_1}Z_{m},Z_{\ell_2}Z_{0})| \\
&\leq \sum_{m=-\infty}^\infty |\cov(\bar{Z}_{m+\ell_1}\bar{Z}_{m},\bar{Z}_{\ell_2}\bar{Z}_{0})| +
2|\mu| \sup_{b \in \mathbb{Z}} \sum_{m=-\infty}^\infty |\cov(\bar{Z}_{m+\ell_1}\bar{Z}_m, \bar{Z}_b)| \\
&\qquad + 2|\mu| \sup_{b \in \mathbb{Z}} \sum_{m=-\infty}^\infty |\cov(\bar{Z}_{m+b}, \bar{Z}_{\ell_2}\bar{Z}_0)| + 4\mu^2 \sup_{b_1,b_2 \in \mathbb{Z}} \sum_{m=-\infty}^\infty |\cov(\bar{Z}_{m+b_1}, \bar{Z}_{b_2})|.
\end{align*}
We finish the proof by showing that each of the four terms on the right-hand side above is uniformly bounded.

First, using the cumulant identity in Footnote 3 in the supplement (which holds since $\bar{Z}_t$ has mean zero),
\begin{align*}
\sum_{m=-\infty}^\infty |\cov(\bar{Z}_{m+\ell_1}\bar{Z}_{m},\bar{Z}_{\ell_2}\bar{Z}_{0})| &\leq  \sum_{m=-\infty}^\infty |\cov(\bar{Z}_{m+\ell_1},\bar{Z}_{\ell_2})||\cov(\bar{Z}_m,\bar{Z}_0)| \\
&\qquad + \sum_{m=-\infty}^\infty |\cov(\bar{Z}_{m+\ell_1},\bar{Z}_0)||\cov(\bar{Z}_m,\bar{Z}_{\ell_2})| \\
&\qquad + \sum_{m=-\infty}^\infty|\cum(\bar{Z}_{m+\ell_1},\bar{Z}_{m},\bar{Z}_{\ell_2},\bar{Z}_{0})|.
\end{align*}
Since $|\cov(\bar{Z}_{m+\ell_1},\bar{Z}_{\ell_2})| \leq \var(\bar{Z}_0)$ for all $(\ell_1,\ell_2,m)$ by Cauchy-Schwarz and stationarity, and the covariances (2nd-order joint cumulants) are absolutely summable, the two first sums on the right-hand side above are uniformly bounded. The final sum on the right-hand side above is bounded above by
\[\sum_{m=-\infty}^\infty \sum_{b_1=-\infty}^\infty \sum_{b_1=-\infty}^\infty |\cum(\bar{Z}_{b_1},\bar{Z}_{m},\bar{Z}_{b_2},\bar{Z}_{0})|,\]
which is finite by the assumed absolute summability of 4th-order joint cumulants.

Second, note that
\begin{align*}
\sup_{b \in \mathbb{Z}} \sum_{m=-\infty}^\infty |\cov(\bar{Z}_{m+\ell_1}\bar{Z}_m, \bar{Z}_b)| &= \sup_{b \in \mathbb{Z}} \sum_{m=-\infty}^\infty |\cum(\bar{Z}_{m+\ell_1},\bar{Z}_m, \bar{Z}_b)| \\
&\leq \sum_{b_1=-\infty}^\infty \sum_{b_2=-\infty}^\infty |\cum(\bar{Z}_{b_1},\bar{Z}_{b_2}, \bar{Z}_0)|,
\end{align*}
where we use the ``$\cum$'' notation also to denote 3rd-order joint cumulants. The equality uses $\cum(X,Y,Z)=E[XYZ]$ for mean-zero variables $X,Y,Z$, and the inequality uses stationarity. The double sum is finite by the assumed absolute summability of 3rd-order joint cumulants. A similar argument shows that
\[\sup_{b \in \mathbb{Z}} \sum_{m=-\infty}^\infty |\cov(\bar{Z}_{m+b}, \bar{Z}_{\ell_2}\bar{Z}_0)|\]
is uniformly bounded.

Finally,
\[\sup_{b_1,b_2 \in \mathbb{Z}} \sum_{m=-\infty}^\infty |\cov(\bar{Z}_{m+b_1}, \bar{Z}_{b_2})| = \sum_{b=-\infty}^\infty |\cov(\bar{Z}_b, \bar{Z}_0)|\]
by stationarity. The latter sum is finite by the assumed absolute summability of covariances (2nd-order joint cumulants).
\end{proof}

\bigskip\bigskip
\phantomsection
\addcontentsline{toc}{section}{References}
\bibliography{ref}